\documentclass[twoside,11pt]{article}
\usepackage{jair, theapa, rawfonts}

\jairheading{1}{2020}{}{04/20}{?}
\ShortHeadings{Encoding Linear Constraints into SAT}
{Ab{\'\i}o, Mayer-Eichberger \& Stuckey}
\firstpageno{0}

\usepackage{etex}
\usepackage{fancybox}
\usepackage{yhmath}
\usepackage{latexsym}
\usepackage{amsmath}
\usepackage{amssymb}
\usepackage{stmaryrd}
\usepackage{xspace}
\usepackage{graphicx}
\usepackage{proof}
\usepackage{pstricks}
\usepackage{url}
\usepackage{epsf}
\usepackage{color}
\usepackage{algorithm}
\usepackage{algorithmic}
\usepackage{listings}
\usepackage{subfig}
\usepackage[all,cmtip]{xy}
\usepackage{amsthm}
\usepackage{rotating}

\lstdefinelanguage{minizinc} 
  {morekeywords={set,of,array,let,in,constraint,var,int,solve,minimize,not,where,forall,if,exists,then,else,true,endif,cummulative}
  ,classoffset=1
  ,sensitive=false
  ,morecomment=[l]{\%}
  ,morestring=[b]"
  ,literate=
  }

\newsavebox{\boxRCPSP}
\newsavebox{\boxSportLeagueA}
\newsavebox{\boxSportLeagueB}
\newsavebox{\boxSportLeagueC}
\newsavebox{\boxSportLeagueD}

\newcommand{\ignore}[1]{}

\newcommand{\pjs}[1]{ {Peter: \color{green}{#1} }}

\newcommand{\vale}[1]{ {Valentin: \color{gray}{#1} }}
\newcommand{\st}{\ \vert \ } 

\newcommand{\aterm}[2]{A_{#2,#1} x_{#2,#1}}
\newcommand{\aline}[1]{b^{#1} \Big( \aterm{#1}{1} &+& \aterm{#1}{2} &+& \cdots &+& \aterm{#1}{n} \Big) }
\newcommand{\atermX}[2]{A_{#2,#1} x_{#2}}
\newcommand{\alineX}[1]{b^{#1} \Big( \atermX{#1}{1} &+& \atermX{#1}{2} &+& \cdots &+& \atermX{#1}{n} \Big) }
\newcommand{\alineP}[1]{\aterm{#1}{1} + \aterm{#1}{2} + \cdots + \aterm{#1}{n} }

\usepackage{caption}
\usepackage{graphics}
\usepackage[showboxes]{textpos}

\usepackage{tikz}
\usetikzlibrary{arrows,shapes}
\usetikzlibrary{decorations.pathreplacing}
\usetikzlibrary{positioning}

\def\and{\wedge}

\def\amax{a_{\max}}

\newtheorem{theorem}{Theorem}{\bf}{\it}
\newtheorem{lemma}[theorem]{Lemma}{\bf}{\it}
\newtheorem{proposition}[theorem]{Proposition}{\bf}{\it}
\newtheorem{remark}[theorem]{Remark}
{\bf}{\it}
\newtheorem{exmp}[theorem]{Example}{\bf}{\it}
\DeclareMathOperator{\search}{\mathbf{search}}
\DeclareMathOperator{\ins}{\mathbf{insert}}
\DeclareMathOperator{\const}{\mathbf{MDDConstruction}}
\DeclareMathOperator{\sv}{SelVar}
\DeclareMathOperator{\child}{Child}
\DeclareMathOperator{\ordEnc}{Order-Encoding}
\DeclareMathOperator{\logEnc}{Logarithmic-Encoding}
\DeclareMathOperator{\mddenc}{MDD-Enc}

\DeclareMathOperator{\smerge}{SimplifiedMerge}
\DeclareMathOperator{\sn}{SortingNetwork}
\DeclareMathOperator{\cn}{CardinalityNetwork}
\DeclareMathOperator{\sm}{SM}

\newcommand{\solns}{\text{solns}}

\newcommand{\MDD}{\textsf{MDD}}
\newcommand{\SN}{\textsf{SN}}
\newcommand{\Sup}{\textsf{Support}}

\newcommand{\SNTARE}{\textsf{SN-Tare}}
\newcommand{\SNOPT}{\textsf{SN-Opt}}

\newcommand{\BDD}{\textsf{BDD}}
\newcommand{\BDDDec}{\textsf{BDD-Dec}}
\newcommand{\Adder}{\textsf{Adder}}
\newcommand{\Gurobi}{\textsf{Gurobi}}
\newcommand{\LCG}{\textsf{LCG}}

\newcommand{\LDMDD}{\textsf{LD-MDD}}
\newcommand{\LDSN}{\textsf{LD-SN}}

\newcommand{\LCGwC}{\textsf{LCG-Pre}}
\newcommand{\LCGnC}{\textsf{LCG}}
\newcommand{\MDDEncAwC}{\textsf{MDD-Pre}}
\newcommand{\MDDEncAnC}{\textsf{MDD}}
\newcommand{\CNwC}{\textsf{SN-Pre}}
\newcommand{\CNnC}{\textsf{SN}}
\newcommand{\SupwC}{\textsf{Support-Pre}}
\newcommand{\SupnC}{\textsf{Support}}

\newcommand{\tnode}{{\ensuremath{\cal T}}}
\newcommand{\fnode}{{\ensuremath{\cal F}}}

\begin{document}

\title{Encoding Linear Constraints into SAT}

\author{\name Ignasi Ab{\'\i}o \email ignasi@barcelogic.com \\
        \addr Barcelogic, K2M Building. Carrer de Jordi Girona, 1 Barcelona,  Spain
        \AND
        \name Valentin Mayer-Eichberger \email valentin@mayer-eichberger.de \\
        \addr Technische Universit\"at Berlin, Germany 
        \AND
        \name Peter Stuckey \email peter.stuckey@monash.edu \\
        \addr Faculty of Information Technology, Monash University, Australia 
       }

\maketitle

\begin{abstract}
Linear integer constraints are one of the most important constraints
in combinatorial problems since they are commonly found in many
practical applications. Typically, encodings to Boolean satisfiability (SAT)
format of conjunctive normal form
perform poorly in
problems with these constraints in comparison with SAT modulo theories
(SMT), lazy clause generation (LCG) or mixed integer programming (MIP)
solvers. 

In this paper we explore and
categorize SAT encodings for linear integer
constraints.
We define new SAT encodings based on multi-valued decision diagrams,
and sorting networks.
We compare different SAT encodings of linear constraints
and demonstrate where one may be preferable to another.
We also compare SAT encodings against other solving methods
and show they can be better than
linear integer (MIP) solvers and sometimes better than LCG/SMT solvers
on appropriate problems.
Combining the new encoding with lazy
decomposition, which during runtime only encodes constraints that are
important to the solving process that occurs, gives the best option
for many highly combinatorial problems involving linear constraints.
\end{abstract}

\section{Introduction}
\label{section-introduction}

In this paper we study \emph{linear integer (LI) constraints},
that is, constraints of the form $a_1 x_1 + \cdots + a_n x_n \;\#\;
a_0$, where the $a_i$ are integer given values, the $x_i$ are
finite-domain integer variables, and the relation operator $\#$
belongs to $\{ <, >, \leqslant, \geqslant, = \}$.

Linear integer constraints appear in almost every combinatorial problem,
including scheduling, planning and software verification, and, therefore, many
different Boolean satisfiability (SAT)
encodings~\shortcite<see e.g.>{Sugar,conf/ismvl/AnsoteguiBMV11}, 
SAT Modulo Theory (SMT) theory solvers~\shortcite{Yices,Z3}, and
propagators~\shortcite{harvey} 
for Constraint Programming (CP) solvers~\shortcite{toplas09}
have been suggested
for handling them.

In this paper we survey existing methods for encoding special cases of
linear constraints, in particular \emph{pseudo-Boolean (PB) constraints}
(where $x_1, \ldots, x_n$ are Boolean (or 0-1) variables), and
\emph{cardinality constraints (CC)} (where, moreover, $a_i = 1, 1 \leq i \leq n$).
We then show how these can be extended to encode general linear
integer constraints.

The first method proposed here roughly consists
in encoding a linear integer constraint into a Reduced Ordered
Multi-valued Decision Diagram (MDD for short), and then decomposing the MDD
to SAT. There are different reasons for choosing this approach:
firstly, most state-of-the-art encoding methods define one auxiliary
variable for every different possible value of the partial sum $s_i =
a_i x_i + a_{i+1} x_{i+1} + \cdots + a_n x_n$. However, some values of the
partial sums may be equivalent in the constraint. 
For example, the expression 
$s_2 = 2 x_2 + 5x_3$ with $x_2 \in [0,2]$ and $x_3 \in [0,3]$
cannot take the value 13 and hence the constraints 
$s_2 \leq 12$ and $s_2 \leq 13$ are equivalent, and hence
we don't need to encode both possible partial sum results.
With MDDs, due to the reduction process, we
can identify these situations, and encode all these indistinguishable
values with a single variable, producing a more compact encoding.

Secondly, BDDs are one of the best methods for encoding pseudo-Boolean
constraints into SAT by \shortciteA{Abio12}, and MDDs
seems the natural tool to generalize the pseudo-Boolean encoding.

The second method uses sorting networks to encode the LI
constraints. The encoding is a generalization of
the methods by \shortciteA{Bailleux09} and by \shortciteA{Een06} 
that have good propagation properties and
better asymptotic size than the BDD/MDD encodings.

The goal of these encodings is not for use in arbitrary problems
involving linear integer constraints. In fact, a specific linear
integer (MIP) solver or CP or SMT solver will usually outperform
any SAT encoding in problems with many
more linear integer constraints than Boolean clauses.

Nevertheless, a fairly common kind of combinatorial problem consists mainly
of Boolean variables and clauses, but also a few integer
variables and LI constraints. Among these problems, an important class
correspond to SAT problems with a linear integer optimization
function. In these cases, SAT solvers are the optimal tool for solving
the problem, but a good encoding for the linear integer constraints is
needed to make the optimization effective.  Therefore, in these
problems the decompositions presented here can make a significant
difference.

Note, however, that decomposing the constraint may not always be the
best option. In some cases the encoding might produce a large number
of variables and clauses, transforming an easy problem for a CP
solver into a huge SAT problem. In some other cases, nevertheless, the
auxiliary variables may give an exponential reduction of the search
space. Lazy decomposition~\shortcite{encodeOrPropagate,lazyDecomposition}
is a hybrid approach that has been successfully used to handle this
issue for cardinality and pseudo-Boolean constraints. Here, we show
that it also can be applied successfully on linear integer
constraints.

The methods proposed here use the order encoding~\shortcite{gent2004new,Ansotegui2004} for representing the
integer variables.  For some LI constraints, however, the domains of the integer
variables are too large to effectively use the order encoding. 
We also propose a new alternative method for encoding linear integer
constraints using a logarithmic encoding of the integer variables. 

In summary, the contributions of this paper are:
\begin{itemize}
  \item A precise definition of correct SAT encodings of constraints over
    non-Boolean variables and the consistency maintained by such a SAT
    encoding.  
  \item A new encoding (\MDD) for LI constraints using MDDs that can
    outperform other state-of-the-art encodings.
  \item A new encoding for Monotonic MDDs into CNF.
  \item A new encoding (\SN) for LI constraints using sorting networks
    that can outperform other state-of-the-art encodings.
  \item A new proof of consistency of direct sorting network encodings of LI
    (that is, without using ``tare'' trick to adjust the right hand side
    $a_0$ to be a power of 2). This is an open question in
    the MiniSAT+ paper~\cite{Een06}.
  \item An alternative encoding (\BDDDec) for LI constraints for large
    constraints or variables with huge domains.
  \item An extensive experimental comparison of our
    methods with respect to other decompositions to SAT and other
    solvers. A total of 14 methods are compared, on more than 5500 benchmarks,
    both industrial and crafted.
\end{itemize}

\ignore{
This paper extends work originally published in a conference~\cite{cp2014a},
by including  far more discussion of related work,
including a survey of existing methods for cardinality and PB constraints,
more rigorous formulation of key concepts, 
theorems with detailed proofs, a new translation of LI constraints based on
sorting networks, and a substantially 
extended experimental section.
}

The paper is organized as follows. First in Section \ref{section-preliminaries}
we introduce SAT
solvers, encodings of integer variables, and how to transform linear
constraints to a standard form. Next in Section~\ref{section-card}, because they are useful also for
encoding LI constraints,  we survey methods for encoding cardinality
constraints into CNF. 
Then in Section~\ref{section-PB}, because methods for encoding
linear constraints are typically extensions of methods for encoding PB
constraints, 
we review various methods for encoding
pseudo-Boolean constraints into CNF.

In Section~\ref{section-LI} we come to the core of the paper, which examines 
various methods for encoding general linear integer constraints.
We first concentrate on encodings using the order encoding of integers.
In Section~\ref{section-groupingCoeffs} we introduce
a simplification of LI constraints (and PB constraints)
that improves on their encoding (and also requires the use LI encodings for
what were originally PB constraints).
In Section~\ref{section-MDD} we define how to encode LI constraints using
multi-valued decision diagrams (MDDs).
In Section~\ref{section-sortingEncoding} we define how to encode LI
constraints using sorting networks.
In Section~\ref{section-sugar} we review the only existing encoding of
general linear constraints to SAT that maintains domain consistency, based
on using partial sums.
In Section~\ref{section-logarithmic} we review the existing encodings of
general linear constraints to SAT based on logarithmic encodings of
integers, and define a new approach \BDDDec.
In Section~\ref{section-experiments} we give detailed experiments
investigating the different encodings, and also compare them against other
solving techniques.
Finally in Section~\ref{section-conclusion} we conclude.

Figure \ref{graph1} gives an overview the different translations and
contributions in this paper. The graph shows a selection of different
translations from LI and PB constraints to intermediate data structures in
focus. The arrows are annotated by the respective publications and Sections. At
the root we have the Linear constraint and its decomposition to the PB
constraints by various methods. A second set of arrows connects the LI
constraints with the data structures directly. 

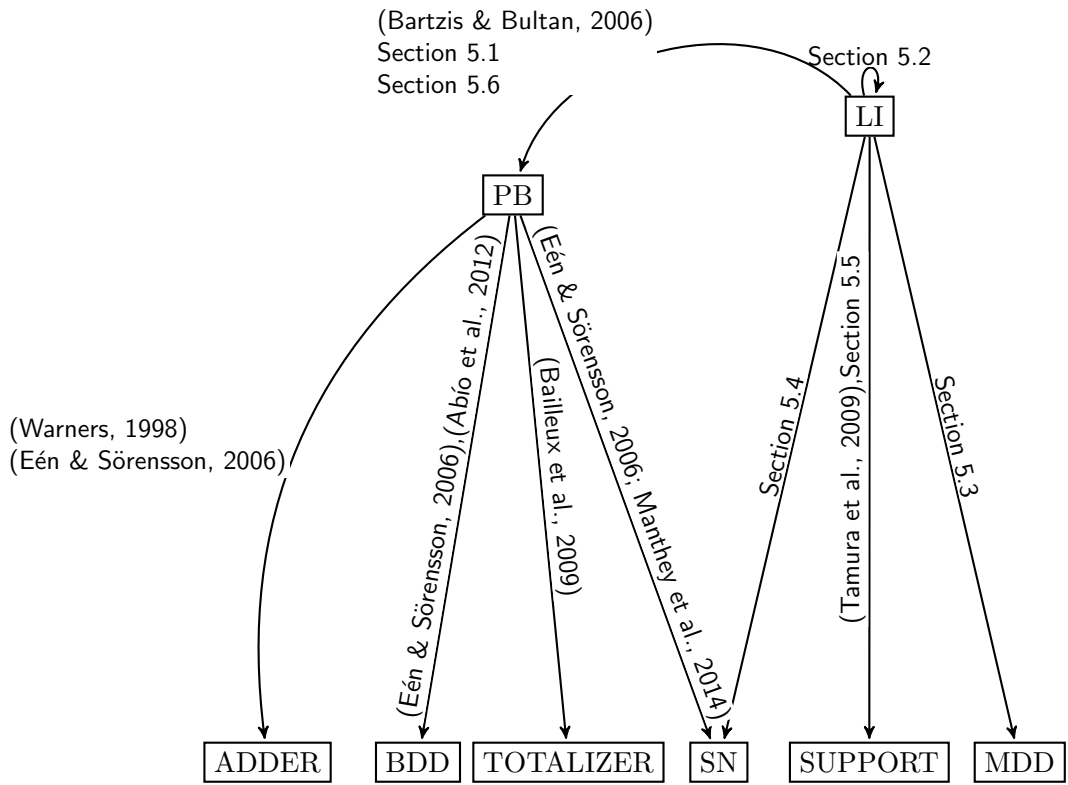
\begin{figure}
\begin{center}
\begin{tikzpicture}[->,>=stealth',shorten >=1pt,auto,node distance=2cm,
  thick,main node/.style={rectangle,draw}]
  \node[main node] (LI)  {LI};
  \node[main node] (PB) [below left=0.5cm and 4cm of LI] {PB};
  \node[main node] (ADD) [below left=7cm and 2cm of PB] {ADDER};
  \node[main node] (BDD) [right of = ADD] {BDD};
  \node[main node] (TOT) [right of = BDD] {TOTALIZER};
  \node[main node] (SN) [right of = TOT] {SN};
  \node[main node] (SUP) [right of = SN] {SUPPORT};
  \node[main node] (MDD) [right of = SUP] {MDD};

  \path[every node/.style={font=\sffamily\small,
  		fill=white,inner sep=1pt}]
        (LI) edge [loop above] node[] {Section \ref{section-groupingCoeffs}} (LI)
    (LI)    edge [bend right=60] node[align=left,left=1mm] {\shortcite{bartzis2006efficient} \\Section \ref{section-order}\\ Section \ref{section-logarithmic}} (PB)
            edge node [above, sloped]{Section \ref{section-MDD}} (MDD)
            edge node [above,sloped] {Section \ref{section-sortingEncoding}} (SN)
            edge node [above,sloped] {\shortcite{Sugar},Section~\ref{section-sugar}} (SUP)
    (PB)    edge [bend right] node[left=1mm,align=left] {\shortcite{Warners98} \\ \shortcite{Een06}} (ADD)
            edge node [above,sloped] {\shortcite{Een06},\shortcite{Abio12}} (BDD)
            edge node [above,sloped] {\shortcite{Bailleux09}} (TOT)
            edge node [above, sloped] {\shortcite{Een06,Manthey14}} (SN);

\end{tikzpicture}
\end{center}
\caption{An overview of various methods to encode LI and PB constraints. The
    majority of the encodings of LI translate first to PB and use existing
    encodings thereof, whereas others translate directly. The terminals symbolize
    the basic building structure of the translation. \label{graph1}}
\end{figure}

We summarize the state of the art in CC, PB and LI encodings 
of the constraint $a_1 x_1 + a_2 x_2 + \cdots + a_n x_n \leq a_0$ where
the number of variables in the LI is $n$,
$a_0$ is the right hand side coefficient, $\amax$ is the largest left hand side
coefficient, and $d$ is the
the size of the largest integer variable domain,
in Table~\ref{tab:all}.
The table shows the basis of  construction method 
using \textsl{Adder}s
, using
\textsl{Totalizer}s
, using sorting networks (\textsl{SN}),
using cardinality networks (\textsl{CN}), using \textsl{BDD}s, using
watchdogs (\textsl{GPW} and \textsl{LPW}), 
and for the general linear (LI) encodings the named encoding.
It gives a reference for the method, and a page where it is discussed.
It then shows the asymptotic size of the encoding and
the consistency maintained by unit propagation on the encoding (see Section~\ref{consistency}),
where --- indicates no consistency. 
Note that any method with a coefficient $a_0$ outside the $\log$ is
exponential in size, the remainder are polynomial.

\begin{sidewaystable}
\caption{Summary of different encodings for CC, PB and LI constraints.\label{tab:all}}
\begin{tabular}{|l|lll|ll|l|}
    \hline
    Cons & Construction & Reference & Page  & \# Clauses                   & Consistency \\
    \hline
    CC & \sl Adder     & \shortcite{Warners98} &\pageref{warners} & $O(n)$              & --- \\
    CC & \sl Totalizer & \shortcite{BailleuxBoufkhad2003CP}& \pageref{boufkhad} & $O(n^2)$ & Domain \\
    CC & \sl SN        & \shortcite{Een06} & \pageref{eensort} & $O(n \log^2 (n))$ & Domain \\
    CC & \sl CN        & \shortcite{AsinNOR11} &\pageref{asinnor}  & $O(n \log^2 (a_0))$ & Domain \\
    CC & \sl CN        & \shortcite{parametricCardinalityConstraint} &\pageref{abiosn}  & $O(n \log^2 (a_0))$ & Domain \\
    
    \hline \hline
    PB & \sl Adder & \shortcite{Warners98} &\pageref{warnerspb} & $O(n\log(a_0))$              & --- \\
    PB & \sl Adder & \shortcite{Een06}   & \pageref{eenadder}  & $O(n\log(a_0))$              & ---  \\
    PB & \sl BDD   & \shortcite{Een06}    & \pageref{eenbdd} & $O(na_0)$ 6/Node             & Domain\\
    PB & \sl BDD   & \shortcite{Bailleux06}& \pageref{b06} & $O(na_0)$ 4/Node             & Domain \\
    PB & \sl BDD   & \shortcite{Abio12}  & \pageref{abiobdd}  & $O(na_0)$ 2/Node             & Domain \\
    PB & \sl SN    & \shortcite{Een06}    & \pageref{eensn} & $O(n\log(\amax) \log^2 (n\log(\amax)))$ & Consistent \\
    PB & \sl GPW   & \shortcite{Bailleux09}& \pageref{gpw} & $O(n^2 \log^2(n) \log(\amax))$ & Consistent \\
    PB & \sl LPW   & \shortcite{Bailleux09}& \pageref{gpw} & $O(n^3 \log^2(n) \log(\amax))$ & Domain \\
    PB & \sl SN    & \shortcite{Manthey14} & \pageref{manthey} & $O(n^2 \log^2(n) \log(\amax))$ & Domain\\
    \hline 
    \hline
    LI & \sl Adder & \shortcite{DBLP:conf/cp/Huang08}& \pageref{huang} & $O(n \log(a_0))$ & --- \\
    LI & \sl SN    & \shortcite{DBLP:journals/jair/AMCS13}&  & $O(n d a_0)$ & --- \\
    LI & \SN       & \ref{section-sortingEncoding} & \pageref{section-sortingEncoding} &$O(n \log (d) \log (n) \log (\amax))$&  Consistent \\
    LI & \BDD      & \shortcite{bartzis2006efficient} & \pageref{bartzis} &$O(n \log (d) \log (a_0))$ & --- \\
    LI & \BDDDec   & \ref{section-logarithmic} & \pageref{section-logarithmic} &$O(n^2 \log(d) \log (\amax))$ & --- \\
    LI & \Sup      & \shortcite{Sugar} &\pageref{section-sugar} &$O(n d a_0)$ & Domain \\
    LI & \MDD      & \ref{section-MDD} & \pageref{section-MDD} &$O(n d a_0)$ & Domain
    \\
    \hline 
\end{tabular}
\end{sidewaystable}

\section{Propagation and Encodings}
\label{section-preliminaries}

In this section we introduce the concepts of variables, domains,
constraints, propagators and encodings into SAT. Mostly, the
terminology we use is standard. The exception is the encodings
into SAT: unfortunately, there is no standard definition for this.

In fact, most papers do not define what is an encoding or what it
means for an encoding to be consistent. In the case of encodings of
constraints of Boolean variables this is not a problem, but, in
general, when dealing with integer variables some encodings cannot
represent all possible domains. In this case, the meaning of
consistency is not clear. In this section we provide a precise
definition of encodings into SAT for both Boolean and integer
variables. With this definition, the concept of consistency can
naturally be extended to encodings.

\subsection{Domains, Constraints and Propagators}

We use $[l,u]$ to denote
the \emph{interval} of integers $\{ d ~|~ l \leq d \leq u \}$. 
Let ${\cal X}$ be a fixed set of variables. A \emph{domain} $D$ is a
complete mapping from ${\cal X}$ to a set of subsets of 
finite set of integers. Given a
domain $D$, and a variable $x \in {\cal X}$, the \emph{domain of the
  variable} $x$ is the set $D(x) \subset \mathbb{Z}$.
In the following, let us fix an initial domain $D$.

A \emph{false domain} $D$ is one where $D(x) = \emptyset$ for some $x \in
{\cal X}$. Let $\oslash$ be the false domain where
$\oslash(x) = \emptyset, \forall x \in {\cal X}$.
A domain $D_1$ is \emph{stronger} than a domain $D_2$, written $D_1
\sqsubseteq D_2$, if $D_1(x) \subseteq D_2(x)$ for all $x \in {\cal
  X}$. Given the domains $D_1$ and $D_2$, the domain $D_1 \sqcap D_2$
is the domain such that $D_1 \sqcap D_2(x) = D_1(x) \cap D_2(x)$ for
all $x \in {\cal X}$. In this paper we assume that the initial domain
is \emph{convex}, i.e., that the domain of every variable $x \in {\cal
  X}$ is an interval. A \emph{complete assignment} is a domain $D'$
such that $\vert D'(x)\vert = 1$ for all $x \in {\cal X}$.

A \emph{constraint} $c$ over the variables $x_1, x_2, \ldots, x_n \in
{\cal X}$ is a subset of the Cartesian product $D(x_1) \times D(x_2)
\times \cdots \times D(x_n)$. A complete assignment $D'$
\emph{satisfies the constraint} $c$ if 
$D'(x_i) = \{d_i\}, 1 \leq i \leq n$ and 
$(d_1, d_2, \ldots, d_n) \in c$. The \emph{solutions of a constraint} $c$, denoted as
$\solns(c)$, are the set of complete assignments that satisfy $c$. A
constraint $c$ is \emph{satisfiable} on the domain $D_1$ if there is a
complete assignment $D_2 \sqsubseteq D_1$ that satisfies
$c$. Otherwise, it is \emph{unsatisfiable} on $D_1$.

Given a constraint $c$, a \emph{propagator} $f$ is a monotonically
decreasing function from domains to domains such that $f(D')
\sqsubseteq D'$ for all domain $D'$; a monotonically decreasing 
function is such that if $D_1 \sqsubseteq D_2$ then $f(D_1) \sqsubseteq
f(D_2)$. A propagator $f$ is \emph{correct} if for all domains $D'$,
$\{ \solns(c) \st \solns(c) \sqsubseteq D' \} = \{ \solns(c) \st
\solns(c) \sqsubseteq f(D') \}$.


Constraint Programming (CP) solvers solve problems by maintaining a domain
$D$, and reducing the domain using a propagator for each constraint $c$ in 
the problem. When propagation can make no further reduction, the solver splits
the problem into two, typically by splitting the domain of a variable in two
disjoint parts, and examines each subproblem in turn.

\subsection{Consistency}\label{consistency}

The identity propagator, $id(D) = D$, is correct for any constraint. In
practice we want propagators to enforce some stronger condition than
correctness.  The usual conditions of interest are:
\begin{description}
\item[consistent]
A propagator $f$ is \emph{consistent} for $c$ if, given any domain $D'$ where
$c$ is unsatisfiable on it, $f(D')$ is a false domain. That is it detects
when the constraint can no longer be satisfied by the domain.
\item[domain consistent]
A propagator $f$ is
\emph{domain consistent} for $c$ if, given any domain $D_1$, for all $x \in
     {\cal X}$ and $d \in f(D_1)(x)$, $c$ is satisfiable on
$$D_2 =
\begin{cases}
  x' \mapsto f(D_1)(x') & \mbox{if } x' \neq x\\
  x' \mapsto \{d\} & \mbox{if } x' = x
\end{cases}$$
A domain consistent propagator infers the maximum possible information,
representable in the domain, from the constraint.
\item[bounds consistent]
A propagator $f$ is \emph{bounds consistent} for $c$ if, given any domain $D_1$,
for all $x \in {\cal X}$ with $f(D_1)$ not a false domain, $c$ is
satisfiable on
$$D_2 =
\begin{cases}
  x' \mapsto [ \min f(D_1)(x'), \max f(D_1)(x')] & \mbox{if } x' \neq x\\
  x' \mapsto \{l\} & \mbox{if } x' = x
\end{cases}$$ and
$$D_3 =
\begin{cases}
  x' \mapsto  [ \min f(D_1)(x'), \max f(D_1)(x')] & \mbox{if } x' \neq x\\
  x' \mapsto \{u\} & \mbox{if } x' = x
\end{cases}$$
where $l = \min f(D_1)(x)$ and $u = \max f(D_1)(x)$.
A bounds consistent propagator enforces that the upper and lower bounds of
each variable appear in some solution to the constraint.
\end{description}

\begin{exmp}
  Given $x_1, x_2, x_3 \in {\cal X}$ with initial domain $x_1 \in
  [0,4]$, $x_2 \in [0,2]$ and $x_3 \in [0,3]$, let us consider the
  constraint $c: 3 x_1 + 2 x_2 + 5 x_3 \leqslant 15$. The propagator $f$ defined by
  $$f(D_1) = \begin{cases}
    D_1 & \mbox{if } 3 \min D_1(x_1) + 2 \min D_1(x_2) + 5 \min D_1(x_3) \leqslant 15\\
    \oslash & \mbox{otherwise}
  \end{cases}$$
  is correct, since if $\{\solns(c) \sqsubset D_1\} \neq \emptyset$,
  then $f(D_1) = D_1$.  It is consistent since if $f(D_1) \neq
  \oslash$, then $\{x_i = \min D_1(x_i) \st i = 1,2,3\}$ is a
  solution of $c$. However, $f$ is not bounds consistent, since given
  $D_1 = \{ x_1 \in [0,4], \ x_2 \in [0,2], \ x_3 \in [3,3] \}$,
  $f(D_1) = D_1$, but $c$ is unsatisfiable in $D_3 = \{ x_1 \in [4,4],
  \ x_2 \in [0,2], \ x_3 \in [3,3] \}$. In the same way, the
  propagator is not domain consistent. 
\end{exmp}

\subsection{SAT Solving}
Let ${\cal Y}=\{y_1,y_2,\ldots \}$ be a fixed set of propositional
\emph{variables}. If $y\in {\cal Y}$ then $y$ and $\neg y$ are
\emph{positive} and \emph{negative literals}, respectively.  The
\emph{negation} of a literal $l$, written $\neg l$, denotes $\neg y$
if $l$ is $y$, and $y$ if $l$ is $\neg y$.  A \emph{clause} is a
disjunction of literals $\neg y_1 \lor \cdots \lor \neg y_p \lor
y_{p+1} \lor \cdots \lor y_n$, sometimes written as $y_1 \land \cdots
\land y_p \rightarrow y_{p+1} \lor \cdots \lor y_n$. A \emph{CNF
  formula} is a conjunction of clauses. Clauses and CNF formulas can
be seen as constraints as defined in the previous section.

A \emph{partial assignment} $A$ is a set of literals such that $\{y,
\neg y \} \not\subseteq A$ for any $y \in {\cal Y}$, i.e., no
contradictory literals appear.  A literal $l$ is \emph{true} in $A$ if
$l \in A$, is \emph{false} in $A$ if $\neg{l} \in A$, and is
\emph{undefined} in $A$ otherwise. True, false or undefined is the
\emph{polarity} of the literal $l$. A non-empty domain $D$ on ${\cal
  Y}$ defines a partial assignment $A$ in the obvious way: if $D(y) =
\{0\}$, then $\neg y \in A$; if $D(y) = \{1\}$, then $y \in A$; and if
$D(y) = \{0, 1\}$, $y$ is undefined in $A$.

Given a CNF formula $F$, \emph{unit propagation} is the propagator
defined as following: given an assignment $A$, it finds a clause in
$F$ such that all its literals are false in $A$ except one, say $l$,
which is undefined, add $l$ to $A$ and repeat the process until
reaching a fix-point.

We assume a basic model of a propagation based SAT solver which captures the
majority of SAT solvers: A system that decides whether a formula has a model by
extending partial assignments through deciding on unassigned variables and
reasoning via unit propagation. See~e.g. the work by
\shortciteA{Nieuwenhuisetal2006JACM} for more details. More advance concepts
such as heuristics or conflict clause learning are not explicitly needed for
our investigation. In our analysis of encodings
we do not consider SAT solvers that follow other paradigms, for instance local
search.

\subsection{Encoding of Integer Variables}
SAT solvers cannot directly deal with non-propositional variables. Therefore,
to tackle a general problem with SAT solvers, the non-propositional variables
must be transformed into propositional ones. This process is called
\emph{encoding the integer variables into SAT}.

Given a set of integer variables $\cal X$ with initial domain $D$, an encoding
of ${\cal X}$ into SAT is a set of propositional variables ${\cal Y}$,
a CNF formula $F$ and a monotonically decreasing function $e$ between
domains of $\cal X$ and partial assignments on ${\cal Y}$ such that:
\begin{itemize}
  \item If $D'$ is empty for some $x \in {\cal X}$, then $e(D')$
    cannot satisfy $F$.
  \item If $D'$ is a complete assignment of $\cal X$, $e(D')$ is a
    complete assignment of ${\cal Y}$ and it satisfies $F$.
  \item The restriction of $e$ to complete assignments is injective.
\end{itemize}

Given $({\cal Y}, F, e)$ an encoding of a set of integer variables
$\cal X$, we can define a monotonically decreasing function $e^{-1}$
between partial assignments of $\cal Y$ to domains on ${\cal X}$
as 
$$e^{-1} (A) = \bigsqcap {\{D' \st e(D') \subseteq A\}}.$$
Notice
that $e^{-1}(e(D')) = D'$. Also notice that $e^{-1}$ is the inverse of
the restriction of $e$ to complete assignments of ${\cal X}$.

Here we consider encodings of a single integer variable: these
encodings can be extended to sets of integer variables in the obvious
way.

There are different methods to encode finite domain variables to SAT that
maintain different levels of consistency. A methodical introduction to this
topic is given by \shortciteA{Walsh00} and \shortciteA{Gent02arcconsistency}.
In case of integer variables for our investigation we focus on the \emph{order}
and the \emph{logarithmic} encoding that we will properly define in this
section. We will not consider the \emph{direct} encoding since it performs
badly with LI constraints \shortcite<e.g.>{Sugar}: it requires a huge number of
clauses even in the simplest LI constraints. Logarithmic encoding produces the
most compact encodings of LI constraints at the expense of propagation
strength. Order encoding produces the smallest encodings among those with good
propagation properties.

Let $x$ be an integer variable with initial domain $[a,b]$. The \emph{order
  encoding}~\shortcite{gent2004new,Ansotegui2004} - sometimes called the
\emph{ladder} or \emph{regular} encoding -  introduces Boolean variables
$y^i$ for $a+1 \leqslant i \leqslant b$. A variable $y^i$ is true iff
$x \geqslant i$. The encoding also introduces the clauses $y^{i+1}
\rightarrow y^{i}$ for $a+1 \leqslant i < b$. In the following, we
denote $\ordEnc(x):=[y^{a+1}, y^{a+2}, \ldots, y^{b}]$. Given a domain $D'
\subseteq [a,b]$ of $x$, let us define $l = \min D'$ and $u = \max
D'$. Then, $e(D') = \{y^i \st i \geq l \} \cup \{\neg y^i \st i \geq
u+1 \}$.

Given $x$ be an integer variable with initial domain $[0, 2^n-1]$.
The \emph{logarithmic encoding}
introduces only $n$ variables $y_b^i$ which codify the binary
representation of the value of $x$, as $x = \sum_{i=0}^{n} 2^i y_b^i$.
In the following, we denote $\logEnc(x):=[y_b^0, y_b^1, \ldots,
  y_b^{n-1}$. 
It is a more compact encoding, but
it usually gives poor propagation performance. Given a domain 
$D'$ over ${\cal X} = \{x\}$ where $D'(x) \subseteq [0, 2^n-1]$, 
$e(D') = \{ y_b^i \st \forall v \in D'(x)\, v/2^i
\equiv 1 \ (\textrm{mod}~ 2) \} \cup \{ \neg y_b^i \st \forall v \in D'(x), v/2^i
\equiv 0 \ (\textrm{mod}~ 2) \}$.

We will only be interested in encoding linear constraints with variables
with initial domain $[0,d]$ (see Section~\ref{section-li}).
To generate the logarithmic encoding
for such a variable with $d \neq 2^n-1$.
We generate the encoding for an integer with initial domain $[0, 2^m-1]$
where $m = \lceil \log (d+1) \rceil$. 
We then add constraints encoding $x \leq d$ as
the lexicographic ordering constraint
$[y^{m-1}_b,y^{m-2}_b,\ldots,y^0_b] \leq
[bit(m-1,d), bit(m-2, d), \ldots, bit(0,d)]$
where $bit(i,d)$ returns the $i^{th}$ bit in the unsigned encoding of
positive integer $d$.
The clauses are
$$
     \bigwedge_{i=0..m-1, bit(i,d) = 0} ((\bigvee_{j=i+1..m-1, bit(i,d) = 0}
     y_b^j) \vee (\bigvee_{j=i+1..m-1, bit(i,d) = 1} \neg y_b^j) \vee \neg y_b^i)
$$
which encode that if the first $m-i-2$ elements in the list are equal,
and the $i^{th}$ bit of $d$ is a 0, then $\neg y_i$ must hold.

\begin{exmp}
  Consider encoding a variable $x$ taking values from $[0,9]$. Then $m = 4$
  and the bits of $d=9$ are $bit(3,d) = 1$, $bit(2,d) = 0$, $bit(1,d) = 0$,
  $bit(0,d) = 1$.  
  We generate Boolean encoding variables
  $[y_b^0, y_b^1, y^2_b, y^3_b]$.
  We encode $x \leq 9$ using the clauses
  $\neg y_b^3 \vee \neg y_b^2$, $\neg y_b^3 \vee y_b^2 \vee \neg y_b^1$.
  Notice these clauses can be simplified (by Krom subsumption) to
  $\neg y_b^3 \vee \neg y_b^2$, $\neg y_b^3 \vee \neg y_b^1$.
\end{exmp}

\subsection{Encoding constraints into SAT}
In the same way, SAT solvers cannot directly deal with general
constraints, so they must be encoded as well. In this section we
explain what is an encoding of a general constraint into SAT.

Let $D$ be a domain on the variables $\cal X$, and let $c$ be a
constraint on ${\cal X}$. Let $({\cal Y}, F, e)$ be an encoding of
$\cal X$. An \emph{encoding of $c$ into SAT} is a set of propositional
variables ${\cal Y}_c \supseteq {\cal Y}$ and a formula $F_c$ such
that given $D'$ a complete assignment on $\cal X$, $D'$ satisfies $c$
if and only if $F_c$ is satisfiable on $e(D')$.

Given an encoding of a constraint $c$ into SAT, unit propagation
defines a propagator of $c$:
\begin{proposition} \label{prop-up-propagator}
  Let $D$ be a domain on the variables $\cal X$, and let $c$ be a
  constraint on $\cal X$. Let $({\cal Y}, F, e)$ be an encoding of
  $\cal X$ and $({\cal Y}_c, F_c)$ an encoding of $c$. Then $D_1
  \mapsto D_1 \sqcap (e^{-1} \circ \pi_{\vert {\cal Y}} \circ
  \textrm{up}_{F_c} \circ e)(D_1)$ is a correct propagator of $c$,
  where $\pi_{\vert {\cal Y}}$ is the projection from ${\cal Y}_c$ to
  $\cal Y$ and $\textrm{up}_{F_c}$ is the unit propagation on $F_c$.
\end{proposition}

Since an encoding of a constraint into SAT defines a propagator of
that constraint, the notions of consistency, bound consistency and
domain consistency can be extended to encodings: an encoding is
consistent/bound consistent/domain consistent if the propagator it
defines is.

\begin{exmp}
  Let us consider again the constraint $c: 3 x_1 + 2 x_2 + 5 x_3
  \leqslant 15$, where $x_1 \in [0,4]$, $x_2 \in [0,2]$ and $x_3 \in
  [0,3]$. Let $({\cal Y}, F, e)$ be the order encoding of $\cal X$. We
  define $$F_c = \bigwedge \left\{ \neg y_1^i \vee \neg y_2^j \vee \neg
  y_3^k \st i \in [0,4], j \in [0,2], k \in [0,5], 3i+2j+5k >15 \right\}
  \wedge y_1^0 \wedge y_2^0 \wedge y_3^0$$ and ${\cal Y}_c = {\cal Y}
  \cup \{y_1^0, y_2^0, y_3^0\}$. Then, $({\cal Y}_c, F_c)$ is an
  encoding of $c$.

  Let $f= e^{-1} \circ \pi_{\vert {\cal Y}} \circ \textrm{up}_{F_c}
  \circ e$ be the propagator defined by unit propagation and
  $F_c$. Given $D_1 = \{ x_1 \mapsto \{1,2,3\}, x_2 \mapsto \{0, 2\}, x_3 \mapsto
  \{2,3\}\}$, $$e(D_1) = \{y_1^1, \neg y_1^4, y_3^1, y_3^2\}.$$

  Unit propagation propagates $y_1^0, y_2^0, y_3^0, \neg y_1^2$ (due
  to clause $\neg y_1^2 \vee \neg y_2^0 \vee \neg y_3^2$), $\neg y_1^3$
  (due to clause $\neg y_1^3 \vee \neg y_2^0 \vee \neg y_3^2$)and $\neg
  y_2^2$ (due to clause $\neg y_1^1 \vee \neg y_2^2 \vee
  \neg y_3^2$). So $$\textrm{up}_{F_c}(e(D_1)) = \{y_1^0, y_1^1, \neg
  y_1^2, \neg y_1^3, \neg y_1^4, y_2^0, \neg y_2^2, y_3^0, y_3^1,
  y_3^2\}.$$

  Therefore, $$\pi_{\vert {\cal Y}} (\textrm{up}_{F_c}(e(D_1))) = \{y_1^1, \neg
  y_1^2, \neg y_1^3, \neg y_1^4, \neg y_2^2, y_3^1, y_3^2\},$$

  so $$e^{-1}(\pi_{\vert {\cal Y}} (\textrm{up}_{F_c}(e(D_1)))) = \{x_1
  \mapsto \{1\}, x_2 \mapsto \{0,1\}, x_3 \mapsto \{2,3\}\}.$$

  Finally, $$f(D_1) = \{x_1 \mapsto \{1\}, x_2 \mapsto \{0\}, x_3 \mapsto
  \{2,3\}\}.$$ 
\end{exmp}

\subsection{Linear Integer Constraints}
\label{section-li}
In this paper we consider linear integer constraints of the form $a_1
x_1 + \cdots + a_n x_n \; \leqslant \; a_0$, where the $a_i$ are
positive integer coefficients and the $x_i$ are integer variables with
domains $[0, d_i]$. Other LI constraints can be easily reduced to this
one:

\vspace{-0.6cm}
$$\renewcommand{\arraystretch}{1.5}
\begin{array}{l @{\hskip 0.5cm} c @{\hskip 0.5cm} l}
  a_1 x_1 + \cdots + a_n x_n \; = \; a_0 & \Longrightarrow &
  \left\{ {\renewcommand{\arraystretch}{1.25}\begin{array}{l}
      a_1 x_1 + \cdots + a_n x_n \; \leqslant \; a_0 \ \wedge \\
      a_1 x_1 + \cdots + a_n x_n \; \geqslant \; a_0 \\
  \end{array}} \right. \\

  \rule{0pt}{3.5ex}
  a_1 x_1 + \cdots + a_n x_n \; < \; a_0 & \Longrightarrow &
  a_1 x_1 + \cdots + a_n x_n \; \leqslant \; a_0 - 1 \\

  a_1 x_1 + \cdots + a_n x_n \; \geqslant \; a_0 & \Longrightarrow  &
  -a_1 x_1 + \cdots + -a_n x_n \; \leqslant \; -a_0 \\

  a_1 x_1 + \cdots + a_n x_n \; > \; a_0 & \Longrightarrow  &
  -a_1 x_1 + \cdots + -a_n x_n \; \leqslant \; -a_0 - 1 \\

  \rule{0pt}{5.5ex}
  \left. {\renewcommand{\arraystretch}{1.25}\begin{array}{l}
      a_1 x_1 + \cdots + a_i x_i + \cdots \\
      + a_n x_n \; \leqslant \; a_0 \\
      \text{ when } x_i \in [l,u], l \neq 0, a_i > 0 \\
  \end{array}} \right\} & \Longrightarrow &
 
  \left\{ {\renewcommand{\arraystretch}{1.25}\begin{array}{l}
      a_1 x_1 + \cdots + a_i x'_i + \cdots  \\
      + a_n x_n \; \leqslant \; a_0 + a_i \times l \ \wedge \\
      x'_i \in [0,u-l] \ \wedge \ x_i' = x_i - l\\
  \end{array}} \right. \\

  \rule{0pt}{6ex}
  \left. {\renewcommand{\arraystretch}{1.25}\begin{array}{l}
      a_1 x_1 + \cdots + a_i x_i + \cdots \\
      + a_n x_n \; \leqslant \; a_0 \\
      \text{when } a_i < 0  \text{ and } x_i \in [l,u] \\
  \end{array}} \right\} & \Longrightarrow &
  \left\{ {\renewcommand{\arraystretch}{1.25}\begin{array}{l}
      a_1 x_1 + \cdots + -a_i x'_i + \cdots \\
      + a_n x_n \; \leqslant \; a_0 - a_i \times u \ \wedge\\
      x'_i \in [0,u-l] \ \wedge \ x'_i = u-x_i \\
  \end{array}} \right. \\

  \rule{0pt}{5ex}
  \left. {\renewcommand{\arraystretch}{1.25}\begin{array}{l}
        x_i \in [l, u], \ l \neq 0 \ \wedge \ x_i' = x_i - l \\
        \wedge \ y_i^j \equiv x_i \geqslant j \ \text{ for } l < j \leqslant u
  \end{array}} \right\} & \Longrightarrow &
  y_i^{j+l} \equiv x_i' \geqslant j \ \text{ for } 1 \leqslant j \leqslant u-l \\

  \rule{0pt}{5ex}
  \left. {\renewcommand{\arraystretch}{1.25}\begin{array}{l}
        x_i \in [l, u] \ \wedge \ x_i' = u - x_i \\
        \wedge \ y_i^j \equiv x_i \geqslant j \ \text{ for } l < j \leqslant u
  \end{array}} \right\} & \Longrightarrow &
  \neg y_i^{u-j+1} \equiv x_i' \geqslant j \ \text{ for } 1 \leqslant j \leqslant u-l \\
\end{array}
$$

All transformations but the first one maintain domain consistency. In
the first transformation, domain consistency is lost. Notice, however,
that a consistent propagator for linear equality would solve the
NP-complete problem subset sum (see
\shortcite{Abio12}), so, in principle, there is no
consistent propagator for equality constraints that runs in polynomial
time (unless P = NP).

\section{Encoding Cardinality Constraints into SAT}
\label{section-card}
In this section we consider the simplest LI constraints: cardinality
constraints. A cardinality constraint is a LI constraint where all
coefficients are 1 and variables are Boolean (i.e., their domains are
$\{0,1\}$). Notice that in this case we do not need to encode the
variables since they are already propositional. 
In this paper we do
not introduce any new encoding for these constraints, nor do we
improve them. However, for completeness, in this section we review the
most usual encodings of cardinality constraints into SAT. 
Practical
comparison of the different methods have been made by Asin~\emph{et al}~\shortcite{AsinNOR11}
and Abio~\emph{et al}~\shortcite{parametricCardinalityConstraint}.
Note that for the special case where $a_0 = 1$ a cardinality constraint 
is an \emph{at most one} constraint for which more efficient encodings exist
(see e.g.~\shortcite{Nguyen:2015:NME:2833258.2833293,modref13}).

\subsection{Encoding Cardinality Constraints with Adders}\label{sec:card-adders}

The first encoding of these constraints is due to Warners~\shortcite{Warners98}\label{warners}. 
Let us consider the cardinality constraint $C:
y_1 + y_2 + \cdots + y_n \leqslant a_0$ where $y_i \in \{0,1\}$. The
idea of the encoding is to compute the binary representation of the
left-side part of the constraint, i.e., the encoding introduces $\log_2
n$ Boolean variables $z_0, z_1, \ldots, z_k$ such that $D: y_1 + y_2 +
\cdots + y_n = z_0 + 2 z_1 + \ldots + 2^k z_k$. With these variables
the enforcement of the original constraint $C$ is trivial.

To enforce the constraint $D$, the encoding creates a circuit
composed of \emph{full-adders} and \emph{half-adders}, defined as
follows:

\begin{itemize}
\item A \emph{full-adder}, denoted by
     $\operatorname{FA}(y_1,y_2,y_3) = (z_0,z_1)$, is a
    circuit with three inputs $y_1, y_2, y_3$ and two outputs $z_0,
    z_1$ such that $y_1 + y_2 + y_3 = z_0 + 2 z_1$.
  \item A \emph{half-adder}, denoted by
    $\operatorname{HA}(y_1,y_2) = (z_0,z_1)$, is a
    circuit with two inputs $y_1, y_2$ and two outputs $z_0, z_1$ such
    that $y_1 + y_2 = z_0 + 2 z_1$.
\end{itemize}

Full and half adders can be naively encoded. This is, the
encoding half-adder consists of the following 7 clauses:
$$\begin{array}{lllllll}
    y_1 \wedge y_2 \rightarrow \neg z_0, & \hspace{0.4em}&
    y_1 \wedge \neg y_2 \rightarrow z_0, & \hspace{0.4em}&
    \neg y_1 \wedge y_2 \rightarrow z_0, &\hspace{0.4em}&
    \neg y_1 \wedge \neg y_2 \rightarrow \neg z_0, \\
    y_1 \wedge y_2 \rightarrow z_1, & \hspace{0.4em}&
    \neg y_1 \rightarrow\neg z_1, &\hspace{0.4em}&
    \neg y_2 \rightarrow  \neg z_1 \\
  \end{array}$$
and the encoding of a full-adder consist of the following 14 clauses:
$$\begin{array}{lllll}
    y_1 \wedge y_2 \wedge y_3 \rightarrow z_0, &\hspace{0.2em}&
    y_1 \wedge y_2 \wedge \neg y_3 \rightarrow \neg z_0, &\hspace{0.2em}&
    y_1 \wedge \neg y_2 \wedge y_3 \rightarrow  \neg z_0, \\
    y_1 \wedge \neg y_2 \wedge \neg y_3 \rightarrow  z_0, &\hspace{0.2em}&
    \neg y_1 \wedge y_2 \wedge y_3 \rightarrow  \neg z_0, &\hspace{0.2em}&
    \neg y_1 \wedge y_2 \wedge \neg y_3 \rightarrow  z_0, \\
    \neg y_1 \wedge \neg y_2 \wedge y_3 \rightarrow  z_0, &\hspace{0.2em}&
    \neg y_1 \wedge \neg y_2 \wedge \neg y_3 \rightarrow \neg z_0, \\
    y_1 \wedge y_2 \rightarrow z_1, &\hspace{0.2em}&
    y_1 \wedge y_3 \rightarrow z_1, &\hspace{0.2em}&
    y_2 \wedge y_3 \rightarrow z_1, \\
    \neg y_1 \wedge \neg y_2 \rightarrow \neg z_1, &\hspace{0.2em}&
    \neg y_1 \wedge \neg y_3 \rightarrow \neg z_1, &\hspace{0.2em}&
    \neg y_2 \wedge \neg y_3 \rightarrow \neg z_1 \\
  \end{array}$$

The circuit to enforce $D$ can be created in several ways. One of the
simplest ways is the recursive one. For $n > 1$, we want to define
$(z_0, \ldots, z_k) = f(y_1, \ldots, y_n)$, where $k = \log_2 n$, such
that constraint $D$ holds.

\begin{description}
\item[If $n=2$:] Then $f$ is a half-adder
   $\operatorname{HA}(y_1,y_2) = (z_0,z_1)$.
\item[If $n=3$:] Then $f$ is a full-adder
  $\operatorname{FA}(y_1,y_2,y_3) = (z_0,z_1)$.
  \item[If $n>3$:] Then
    $$\begin{array}{lll}
      (w_1, w_2, \ldots, w_k) &=& f(y_1, y_2, \ldots, y_{n/2})\\
      (w_{k+1}, w_{k+2}, \ldots, w_{2k}) &=& f(y_{n/2+1}, y_{n/2+2}, \ldots, y_{n})\\
      (z_0, c_2) &=& \operatorname{HA}(w_1,w_{k+1})\\
      (z_1, c_3) &=& \operatorname{FA}(w_2,w_{k+2},c_2)\\
      (z_2, c_4) &=& \operatorname{FA}(w_3,w_{k+3},c_3)\\
      &\cdots\\
      (z_{k-2}, c_{k}) &=& \operatorname{FA}(w_{k-1},w_{2k-1},c_{k-1})\\
      (z_{k-1}, z_{k}) &=& \operatorname{FA}(w_{k},w_{2k},c_{k})\\
  \end{array}$$
\end{description}

Figure \ref{figure-adders-CC} shows the recursive construction explained here.
\begin{figure}[t]
  \begin{center}
    \includegraphics[scale=0.7]{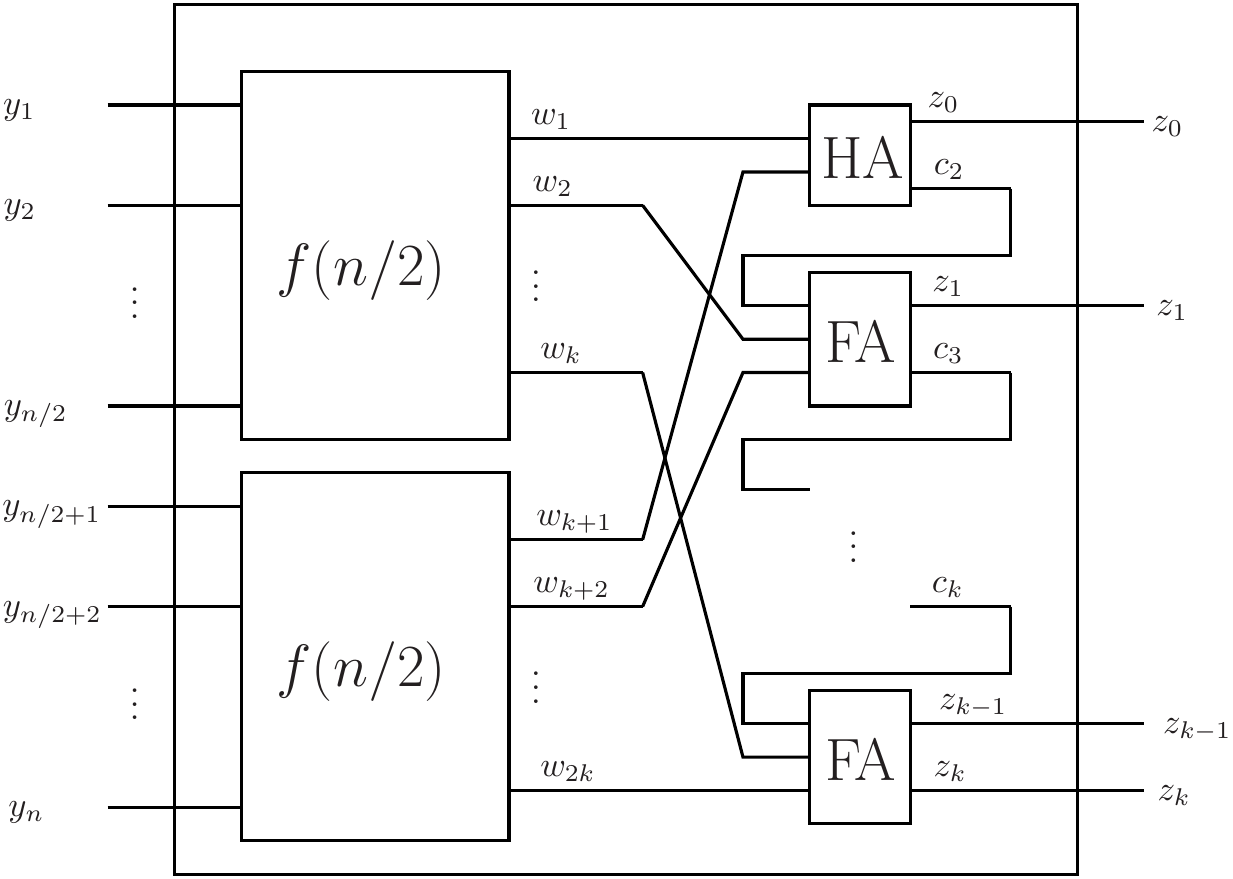}
    \caption{\label{figure-adders-CC}Recursive definition of an adder network.}
  \end{center}
\end{figure}

\begin{exmp}
  Consider the constraint $y_1 + y_2 + y_3 + y_4 + y_5 \leqslant
  2$. The adder encoding introduces variables $w_1, w_2, w_3, w_4,
  c_2, z_0, z_1, z_2$ defined as:
  $$\begin{array}{lll}
    (w_1,w_2)=\operatorname{FA}(y_1, y_2, y_3) &\quad&
    (z_0,c_2)=\operatorname{HA}(w_1, w_3) \\
    (w_3,w_4)=\operatorname{HA}(y_4, y_5) &\quad&
    (z_1,z_2)=\operatorname{FA}(c_2,w_2, w_4) \\
  \end{array}$$
  In addition, it enforces that $z_0 + 2 z_1 + 4 z_2 \leqslant 2$, so it
  produces clauses $\{\neg z_2, \neg z_0 \vee \neg z_1 \}$.

  All in all, the encoding consists of the following clauses:
  $$\begin{array}{lllll}
    \neg y_1 \vee \neg y_2 \vee \neg y_3 \vee w_1, &\hspace{0.2em}&
    \neg y_1 \vee \neg y_2 \vee  y_3 \vee \neg w_1, &\hspace{0.2em}&
    \neg y_1 \vee  y_2 \vee \neg y_3 \vee  \neg w_1, \\
    \neg y_1 \vee  y_2 \vee  y_3 \vee  w_1, &\hspace{0.2em}&
    y_1 \vee \neg y_2 \vee \neg y_3 \vee  \neg w_1, &\hspace{0.2em}&
    y_1 \vee \neg y_2 \vee  y_3 \vee  w_1, \\
    y_1 \vee  y_2 \vee \neg y_3 \vee  w_1, &\hspace{0.2em}&
    y_1 \vee  y_2 \vee  y_3 \vee \neg w_1 &\hspace{0.2em}&
    \neg y_1 \vee \neg y_2 \vee w_2, \\
    \neg y_1 \vee \neg y_3 \vee w_2, &\hspace{0.2em}&
    \neg y_2 \vee \neg y_3 \vee w_2, &\hspace{0.2em}&
    y_1 \vee  y_2 \vee \neg w_2, \\
    y_1 \vee  y_3 \vee \neg w_2, &\hspace{0.2em}&
    y_2 \vee  y_3 \vee \neg w_2, &\hspace{0.2em}&
    \neg y_4 \vee \neg y_5 \vee \neg w_3, \\
    \neg y_4 \vee  y_5 \vee w_3, & \hspace{0.2em}&
    y_4 \vee \neg y_5 \vee w_3, &\hspace{0.2em}&
    y_4 \vee  y_5 \vee \neg w_3, \\
    \neg y_4 \vee \neg y_5 \vee w_4, & \hspace{0.2em}&
    y_4 \vee\neg w_4, &\hspace{0.2em}&
    y_5 \vee  \neg w_4, \\
    \neg w_1 \vee \neg w_3 \vee \neg z_0, & \hspace{0.2em}&
    \neg w_1 \vee  w_3 \vee z_0, & \hspace{0.2em}&
    w_1 \vee \neg w_3 \vee z_0,  \\
    w_1 \vee  w_3 \vee \neg z_0, & \hspace{0.2em}&
    \neg w_1 \vee \neg w_3 \vee c_2, & \hspace{0.2em}&
    w_1 \vee\neg c_2, \\
    w_3 \vee  \neg c_2, & \hspace{0.2em}&
    \neg c_2 \vee \neg w_2 \vee \neg w_4 \vee z_1, &\hspace{0.2em}&
    \neg c_2 \vee \neg w_2 \vee  w_4 \vee \neg z_1, \\
    \neg c_2 \vee  w_2 \vee \neg w_4 \vee  \neg z_1, &\hspace{0.2em}&
    \neg c_2 \vee  w_2 \vee  w_4 \vee  z_1, &\hspace{0.2em}&
    c_2 \vee \neg w_2 \vee \neg w_4 \vee  \neg z_1, \\
    c_2 \vee \neg w_2 \vee  w_4 \vee  z_1, &\hspace{0.2em}&
    c_2 \vee  w_2 \vee \neg w_4 \vee  z_1, &\hspace{0.2em}&
    c_2 \vee  w_2 \vee  w_4 \vee \neg z_1 \\
    \neg c_2 \vee \neg w_2 \vee z_2, &\hspace{0.2em}&
    \neg c_2 \vee \neg w_4 \vee z_2, &\hspace{0.2em}&
    \neg c_2 \vee \neg w_4 \vee z_2, \\
    c_2 \vee  w_2 \vee \neg z_2, &\hspace{0.2em}&
    c_2 \vee  w_4 \vee \neg z_2, &\hspace{0.2em}&
    w_2 \vee  w_4 \vee \neg z_2, \\
    \neg z_2, &\hspace{0.2em}& \neg z_0 \vee \neg z_1\\
  \end{array}$$

  Consider now the partial assignment $\{y_1, y_3, y_4\}$. Unit
  propagation just enforces $w_2$, but it does not detect any
  conflict. Therefore, the encoding is not consistent.
\end{exmp}

\begin{theorem}
  The adder encoding defined in this section encodes cardinality
  constraints with $O(n)$ variables and clauses. The encoding does not
  maintain consistency.
\end{theorem}
\begin{proof}
  The encoding does not maintain consistency due to the previous example. That
    the encoding needs $O(n)$ clauses and variables is shown by
    \shortciteA{Warners98} in Lemma 2. 
\end{proof}

\subsection{Encoding Cardinality Constraints by Sorting the Input Variables}\label{sect-sorting-CC}
As before, let us consider the cardinality constraint $C: y_1 + y_2 +
\cdots + y_n \leqslant a_0$. Let $s$ be the integer variable $s = y_1
+ y_2 + \ldots + y_n$. The adder encoding introduced the Boolean
variables $\logEnc(s)$, and then easily encoded $s \leqslant
a_0$. Here, the idea is introduce the Boolean variables $\ordEnc(s)$
and then encode $s \leqslant a_0$.

The way to introduce $\ordEnc(s)$ is by sorting the input variables:
this is, given $y_1, y_2, \ldots, y_n$, we want to generate $z_1, z_2,
\ldots, z_n$ such that $z_1 \leqslant z_2 \leqslant \cdots \leqslant
z_n$ and $\{y_1, y_2, \ldots, y_n \} = \{z_1, z_2, \ldots, z_n\}$ as a multiset
(i.e., there are the same number of true and false variables in both sides).
There are different ways to construct encodings that perform such sorting. We
present two ways, first \emph{totalizers} and secondly \emph{odd-even sorters}
as an example for a comparator based sorting networks.

\subsubsection{Sorting Variables with Totalizers} \label{sec:totalizers}
An encoding for cardinality constraints through totalizers was given
by \shortciteA{BailleuxBoufkhad2003CP}\label{boufkhad}. The idea of the encoding is simple:
given the input variables $y_1, y_2, \ldots, y_n$, the method splits
the variables in two halves and recursively sorts both halves. Then,
with a quadratic number of clauses, the method produces the sorted
output. More specifically $(z_1, z_2, \ldots, z_n) =
\operatorname{Sort_{TOT}} (y_1, y_2, \ldots, y_n)$ is defined by:
\begin{description}
  \item[If $n=1$:] $z_1 = y_1$.
  \item[If $n>1$:] Let us define $$(w_1, w_2, \ldots, w_{n/2}) =
      \operatorname{Sort_{TOT}}(y_1, y_2, \ldots, y_{n/2})$$ and $$(w'_1, w'_2,
      \ldots, w'_{n/2}) = \operatorname{Sort_{TOT}}(y_{n/2+1}, y_{n/2+2},
    \ldots, y_{n}).$$ Then: $$z_i = \bigvee \{ w_j \wedge w'_k \st j + k
    = i\},$$ where $w_0$ and $w'_0$ are dummy true variables.
\end{description}
\begin{exmp}
  Let us consider again the constraint $y_1 + y_2 + y_3 + y_4 + y_5
  \leqslant 2$. The method of \shortciteA{BailleuxBoufkhad2003CP}
  introduces the following variables:
  $$\begin{array}{lll}
    w_1 = y_1 \vee y_2 &\hspace{0.2em}&
    w_2 = y_1 \wedge y_2 \\
    w_3 = w_1 \vee y_3 &\hspace{0.2em}&
    w_4 = y_2 \vee (w_1 \wedge y_3) \\
    w_5 = w_2 \wedge y_3 &\hspace{0.2em}&
    w_6 = y_4 \vee y_5 \\
    w_7 = y_4 \wedge y_5 &\hspace{0.2em}&
    z_1 = w_3 \vee w_6 \\
    z_2 = w_4 \vee (w_3 \wedge w_6) \vee w_7 &\hspace{0.2em}&
    z_3 = w_5 \vee (w_4 \wedge w_6) \vee (w_3 \wedge w_7) \\
    z_4 = (w_5 \wedge w_6) \vee (w_4 \wedge w_7) &\hspace{0.2em}&
    z_5 = w_5 \wedge w_7
  \end{array}$$
  Finally, the method adds the clause $\neg z_3$. All in all, the
  method produces the following set of clauses:
  $$\begin{array}{lllll}
    \neg y_1 \vee w_1, &\hspace{0.2em}&
    \neg y_2 \vee w_1, &\hspace{0.2em}&
    y_1 \vee y_2 \vee \neg w_1, \\
    \neg y_1 \vee \neg y_2 \vee w_2, &\hspace{0.2em}&
    y_1 \vee \neg w_2, &\hspace{0.2em}&
    y_2 \vee \neg w_2, \\
    \neg w_1 \vee w_3, &\hspace{0.2em}&
    \neg y_3 \vee w_3, &\hspace{0.2em}&
    w_1 \vee y_3 \vee \neg w_3, \\
    \neg y_2 \vee w_4, &\hspace{0.2em}&
    \neg w_1 \vee \neg y_3 \vee w_4, &\hspace{0.2em}&
    y_2 \vee w_1 \vee \neg w_4, \\
    y_2 \vee y_3 \vee \neg w_4, &\hspace{0.2em}&
    \neg w_2 \vee \neg y_3 \vee w_5, &\hspace{0.2em}&
    w_2 \vee \neg w_5, \\
    y_3 \vee \neg w_5, &\hspace{0.2em}&
    \neg y_4 \vee w_6, &\hspace{0.2em}&
    \neg y_5 \vee w_6, \\
    y_4 \vee y_5 \vee \neg w_6, &\hspace{0.2em}&
    \neg y_4 \vee \neg y_5 \vee w_7, &\hspace{0.2em}&
    y_4 \vee \neg w_7, \\
    y_5 \vee \neg w_7, &\hspace{0.2em}&
    \neg w_3 \vee z_1, &\hspace{0.2em}&
    \neg w_6 \vee z_1, \\
    w_3 \vee w_6 \vee \neg z_1, &\hspace{0.2em}&
    \neg w_4 \vee z_2, &\hspace{0.2em}&
    \neg w_3 \vee \neg w_6 \vee z_2, \\
    \neg w_7 \vee z_2, &\hspace{0.2em}&
    w_4 \vee w_3 \vee w_7 \vee \neg z_2, &\hspace{0.2em}&
    w_4 \vee w_6 \vee w_7 \vee \neg z_2, \\
    \neg w_5 \vee z_3, &\hspace{0.2em}&
    \neg w_4 \vee \neg w_6 \vee z_3, &\hspace{0.2em}&
    \neg w_3 \vee \neg w_7 \vee z_3, \\
    w_5 \vee w_4 \vee w_3 \vee \neg z_3, &\hspace{0.2em}&
    w_5 \vee w_4 \vee w_7 \vee \neg z_3, &\hspace{0.2em}&
    w_5 \vee w_6 \vee w_3 \vee \neg z_3, \\
    w_5 \vee w_6 \vee w_7 \vee \neg z_3, &\hspace{0.2em}&
    \neg w_5 \vee \neg w_6 \vee z_4, &\hspace{0.2em}&
    \neg w_4 \vee \neg w_7 \vee z_4, \\
    w_5 \vee w_4 \vee \neg z_4, &\hspace{0.2em}&
    w_5 \vee w_7 \vee \neg z_4, &\hspace{0.2em}&
    w_6 \vee w_4 \vee \neg z_4, \\
    w_6 \vee w_7 \vee \neg z_4, &\hspace{0.2em}&
    \neg w_5 \vee \neg w_7 \vee z_5, &\hspace{0.2em}&
    w_5 \vee \neg z_5, \\
    w_5 \vee \neg w_7, &\hspace{0.2em}&
    \neg z_3
  \end{array}$$
  and maintains domain consistency. For instance, given the partial
  assignment $\{y_1, y_3\}$ unit propagation enforces that $\neg w_5,
  \neg w_2, \neg y_2, w_1, w_3, w_4, \neg w_6, \neg w_7, \neg y_4,
  \neg y_5$. 
\end{exmp}

\begin{theorem}[\citeR{BailleuxBoufkhad2003CP}]
  The totalizer encoding defined in this section encode cardinality
  constraints with $O(n \log n)$ variables and $O(n^2)$ clauses. The
  encoding maintains domain consistency.

\end{theorem}

\subsubsection{Sorting Variables with Sorting Networks}\label{section-CC-SN} 

An improved version of the previous encoding was given by
~\shortciteA{Een06}\label{eensort} by using a \emph{odd-even} sorting network.
To sort $n$ Boolean variables Odd-Even Sorting networks split the variables in
two halves and recursively sort them. The \emph{merge} of these two already
sorted sets of variables is also done
recursively. This is, $(z_1, z_2, \ldots, z_n) = \operatorname{Sort_{OE}} (y_1,
y_2, \ldots, y_n)$ is defined by\footnote{In \shortcite{Een06} it is assumed that
    $n$ is a power-of-two: dummy false input variables can be added if needed.
A definition that works for arbitrary $n$ is presented by ~\shortciteA{parametricCardinalityConstraint}. Here, however, we consider the
power-of-two case for simplicity.}
\begin{description}
  \item[If $n=1$:] $z_1 = y_1$.
  \item[If $n>1$:] Let us define $$(w_1, w_2, \ldots, w_{n/2}) =
    \operatorname{Sort_{OE}}(y_1, y_2, \ldots, y_{n/2})$$ and $$(w'_1,
    w'_2, \ldots, w'_{n/2}) = \operatorname{Sort_{OE}}(y_{n/2+1},
    y_{n/2+2}, \ldots, y_{n}).$$ Then: $$(z_1,z_2, \ldots, z_n) =
    \operatorname{Merge_{OE}}(w_1, w_2, \ldots, w_{n/2}; w'_1, w'_2,
    \ldots, w'_{n/2}),$$
\end{description}
where $$(z_1, z_2, \dots, z_{2n}) = \operatorname{Merge_{OE}}(y_1, y_2,
\ldots, y_n; y'_1, y'_2, \ldots, y'_n)$$ is recursively defined as
\begin{description}
  \item[If $n=1$:] $(z_1,z_2) = (y_1 \vee y'_1, y_1 \wedge y'_1)$.
  \item[If $n>1$:] Let us define $$(w_1, w_3, \ldots, w_{2n-1}) =
    \operatorname{Merge_{OE}}(y_1, y_3, \ldots, y_{n-1}; y'_1, y'_3,
    \ldots, y'_{n-1})$$ and $$(w_2, w_4, \ldots, w_{2n}) =
    \operatorname{Merge_{OE}}(y_{2}, y_{4}, \ldots, y_{2n}; y'_{2}, y'_{4},
    \ldots, y'_{2n}).$$ Then:
    $$\begin{array}{rcl}
      z_1 &=& w_1,\\
      (z_2, z_3) &=& \operatorname{Merge_{OE}}(w_2, w_3),\\
      &\ldots\\
      (z_{2n-2}, z_{2n-1}) &=& \operatorname{Merge_{OE}}(w_{2n-2}, w_{2n-1}),\\
      z_{2n} &=& w_{2n}.
      \end{array}$$
\end{description}
Figure \ref{figure-SN} shows the recursive construction of sorting networks.
\begin{figure}[t]
  \begin{center}
      \includegraphics[scale=0.7]{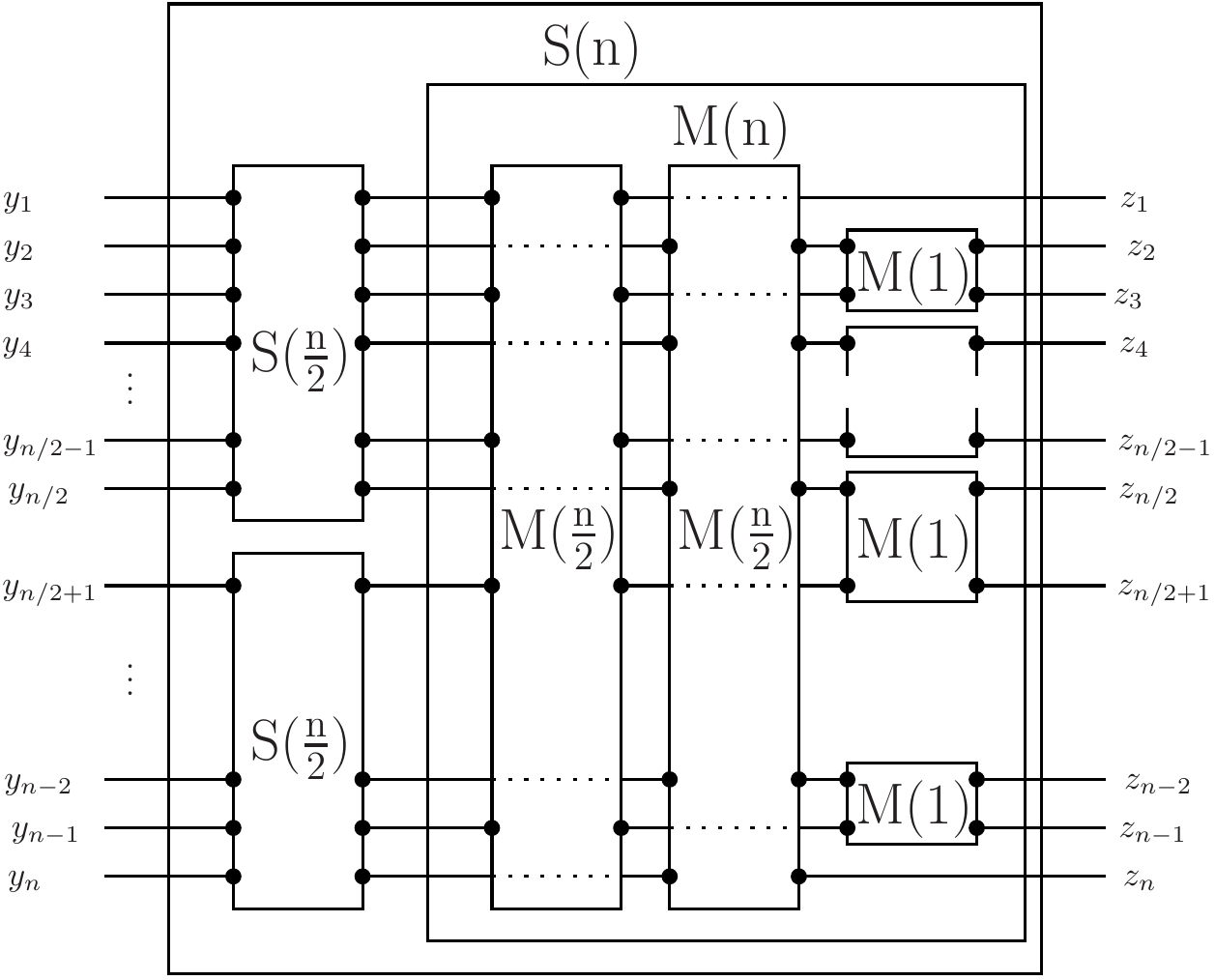}
    \caption{\label{figure-SN}A sorting network and a merge of size $n$.}
  \end{center}
\end{figure}

\begin{theorem}[\citeR{Een06}]
  The sorting network encoding defined in this section encodes
  cardinality constraints with $O(n \log^2 n)$ variables and
  clauses. The encoding maintains domain consistency.
  
\end{theorem}

This method was improved first by \shortciteA{AsinNOR11}\label{asinnor} and then
at by \shortciteA{parametricCardinalityConstraint}\label{abiosn}. 
In the first paper the
authors reduce the size of the network by computing only the $a_0+1$
most significant bits. In this case, the encoding needs $O(n
\log^2a_0)$ variables and clauses without losing any propagation
strength. Sorting Networks with fewer outputs than inputs are called
\emph{cardinality Networks} (which we denote by $\cn$),
and merge networks with fewer outputs than
inputs are called \emph{simplified Merges} (which we denote by $\smerge$).

The authors also change the base case of the definition of Merge from
$z_1 = y_1 \vee y_1'$ and $z_2 = y_1 \wedge y_1'$ to $y_1 \vee y_1'
\rightarrow z_1 $ and $y_1 \wedge y_1' \rightarrow z_2$, halving the
number of clauses needed. Note that this means that the outputs $z_1, z_2$
are no longer functionally defined by the inputs $y_1, y_1'$.  

In the second paper, the authors combine the recursive definitions of
Sorting Network and Merge with naive ones (this is, without auxiliary
variables). The naive encoding needs an exponential number of clauses,
but it is more compact than the recursive one for small input sizes:
therefore, the naive encoding can be used in some recursive
calls. Using a dynamic programming approach, the authors achieve a
much more compact encoding. The asymptotic size or propagation
strength is the same that of \shortciteA{AsinNOR11}, 
but in practice the encoding has better size and performance.

\begin{exmp}\label{example-cardinalitynetworks}
  Let us consider again the constraint $y_1 + y_2 + y_3 + y_4 + y_5
  \leqslant 2$. The encoding of this section defines 10 variables
  $$\begin{array}{rcl}
    (w_1,w_2,w_3) & = & \ordEnc(y_1 + y_2 + y_3)\\
    (w_4,w_5) & = & \ordEnc (y_4 + y_5)\\
    (z_1, z_2, z_3, z_4, z_5) &=& \operatorname{Merge}(w_1, w_2, w_3; w_4, w_5)\\
  \end{array}$$

  The encoding of \shortciteA{parametricCardinalityConstraint} defines just 8
  variables
  $$\begin{array}{rcl}
    (w_1,w_2,w_3) & = & \cn(y_1,y_2,y_3)\\
    (w_4,w_5) &=& \cn(y_4,y_5)\\
    (z_1, z_2, z_3)&=& \smerge(w_1, w_2, w_3; w_4, w_5)
  \end{array}$$
  and clauses
  $$\begin{array}{lllll}
    \neg y_1 \vee w_1, &\hspace{0.2em}&
    \neg y_2 \vee w_1, &\hspace{0.2em}&
    \neg y_3 \vee w_1, \\
    \neg y_1 \vee \neg y_2 \vee w_2, &\hspace{0.2em}&
    \neg y_1 \vee \neg y_3 \vee w_2, &\hspace{0.2em}&
    \neg y_2 \vee \neg y_3 \vee w_2, \\
    \neg y_1 \vee \neg y_2 \vee \neg y_3 \vee w_3, &\hspace{0.2em}&
    \neg y_4 \vee w_4, &\hspace{0.2em}&
    \neg y_5 \vee w_4, \\
    \neg y_4 \vee \neg y_5 \vee w_5, &\hspace{0.2em}&
    \neg w_1 \vee z_1, &\hspace{0.2em}&
    \neg w_4 \vee z_1, \\
    \neg w_2 \vee z_2, &\hspace{0.2em}&
    \neg w_1 \vee \neg w_4 \vee z_2, &\hspace{0.2em}&
    \neg w_5 \vee z_2, \\
    \neg w_3 \vee z_3, &\hspace{0.2em}&
    \neg w_2 \vee \neg w_4 \vee z_3, &\hspace{0.2em}&
    \neg w_1 \vee \neg w_5 \vee z_3, \\
    \neg z_3 \\
  \end{array}$$

  The method is domain consistent: for instance, given the partial
  assignment $\{y_1, y_3 \}$ unit propagation enforces that $\neg z_3,
  \neg w_3, \neg y_2, w_1, w_2, \neg w_4, \neg w_5, \neg y_4, \neg
  y_5$. 
\end{exmp}

A very similar encoding were given by 
\shortciteA{DBLP:conf/lpar/CodishZ10}
and \shortciteA{DBLP:journals/jair/AMCS13}. In this case the authors used a
different definition of Sorting Networks, called Pairwise Sorting
Networks by \shortciteA{DBLP:journals/ppl/Parberry92}. By means of partial
evaluation, this method also achieves $O(n \log^2a_0)$ variables and
clauses, and produce a similar encoding to those of \shortciteA{AsinNOR11} in
terms of size and propagation strength.

\section{Encoding Pseudo-Boolean Constraints into SAT}
\label{section-PB}
Pseudo-Boolean constraints are a well-studied topic in the SAT
community. In this section we review the literature and
provide a survey of the existing translations. PB constraints are a
natural extension to cardinality constraints. As in the previous
section, notice that the variables of these constraints are already
propositional, so we do not have to encode them.

Let us fix the pseudo-Boolean constraint $a_1 y_1 + a_2 y_2 + \ldots + a_n y_n
\leqslant a_0$. The direct translation without auxiliary variables generates
one clause for each minimal subset of $\{y_1 \ldots y_n\}$ such that the sum of
the respective $a_i$ exceeds $a_0$. This translation is unique, but not
practical since even simple PBs require $O(2^n)$ clauses, as shown by
~\shortciteA{Warners98}. 

All the remaining methods in the literature introduce
the integer variable $s = a_1 y_1 + a_2 y_2 + \ldots + a_n y_n$ and
then simply encode $s \leqslant a_0$. The main differences between the
encodings are the way to represent $s$ (either the order or the
logarithmic encoding) and the way to enforce the definition of $s$.

We classify the encodings of PB constraints into three groups: the
ones using binary adders to obtain the logarithmic encoding of $s$;
the ones that use some sorting method to obtain the logarithmic encoding
of $s$; and the ones that, incrementally, define the order encoding of
the partial sums, obtaining finally the order encoding of $s$.

\subsection{Encoding Pseudo-Boolean Constraints with Adders}
\label{section-PB-adders}
Given positive integer numbers $a_1, a_2, \ldots, a_n$, an adder
encoding for PB constraints defines propositional variables $z_0, z_1,
\ldots, z_m$ such that 
$$a_1 y_1 + a_2 y_2 + \ldots + a_n y_n = z_0 +
2z_1 + 2^2 z_2 + \ldots + 2^m z_m,$$ 
where $m= \left \lfloor \log (a_1 + a_2 + \ldots + a_n) \right \rfloor$. 
In Section \ref{sec:card-adders} we showed a method to
accomplish this when $a_1 = a_2 = \cdots = a_n = 1$. Here we extend
this method for arbitrary positive coefficients.

First of all, for all $i$ in $\{1, 2, \ldots, n\}$ let $(A_i^0, A_i^1,
\ldots, A_i^{m})$ be the binary representation of $a_i$; this is,
$A_i^j \in \{0,1\}$ and $a_i = A_i^0 + 2 A_i^1 + \cdots + 2^m A_i^m$

Therefore, we obtain that
\begin{multline*}
a_1 y_1 + a_2 y_2 + \ldots + a_n y_n =\\ = (A_1^0 y_1 + A_2^0 y_2 +
\cdots + A_n^0 y_n) + \cdots + 2^m (A_1^m y_1 + \cdots + A_n^m y_n).
\end{multline*}
The goal of the encoding is to obtain the logarithmic encoding of that
sum. This can be easily done by repeatedly using the Adder encoding
for cardinality constraints explained in Section
\ref{sec:card-adders}.

More specifically, let $z_0^0, z_0^1, \ldots, z_0^{m_0}$ be the
logarithmic encoding of $A_1^0 y_1 + A_2^0 y_2 + \cdots + A_n^0 y_n$,
obtained as shown in Section \ref{sec:card-adders}. This means
that 
$$z_0^0 + 2 z_0^1 + \ldots + 2^{m_0} z_0^{m_0} = A_1^0 y_1 +
A_2^0 y_2 +\cdots + A_n^0 y_n.$$ 
Let us define $z_0 = z_0^0$. Therefore,
\begin{multline*}
  (A_1^0 y_1 + \cdots + A_n^0 y_n) + 2 (A_1^1 y_1 + \cdots + A_n^1
  y_n) +2^2 (A_1^2 y_1 + \cdots + A_n^2 y_n) + \cdots = \\ = z_0^0 + 2
  z_0^1 + 2^2 z_0^2 + \cdots + 2 (A_1^1 y_1 + \cdots + A_n^1 y_n ) +
  2^2 (A_1^2 y_1 + \cdots + A_n^2 y_n) + \cdots = \\ =z_0 + 2 (A_1^1
  y_1 + \cdots + A_n^1 y_n + z_0^1) + 2^2 (A_1^2 y_1 + \cdots + A_n^2
  y_n + z_0^2) + \cdots
\end{multline*}
Again, let $z_1^0, z_1^1, \ldots, z_1^{m_1}$ be the logarithmic
encoding of $A_1^1 y_1 + \cdots + A_n^1 y_n + z_0^1$. Then, if we
define $z_1 = z_1^0$,
\begin{multline*}
  (A_1^0 y_1 + \cdots + A_n^0 y_n) + 2 (A_1^1 y_1 + \cdots + A_n^1
  y_n) +2^2 (A_1^2 y_1 + \cdots + A_n^2 y_n) + \cdots =\\= z_0 + 2
  z_1 + 2^2 (A_1^2 y_1 + \cdots +
  A_n^2 y_n + z_0^2 + z_1^1) + \cdots
\end{multline*}
And so on and so forth until we obtain $$(A_1^0 y_1 + A_2^0 y_2 +
\cdots + A_n^0 y_n) + \cdots + 2^m (A_1^m y_1 + \cdots + A_n^m y_n) =
z_0 + 2 z_1 + \cdots + 2^m z_m.$$


\begin{exmp}
  Let us consider the constraint $2y_1 + 3y_2 + 5 y_3 + 6 y_4
  \leqslant 9$. The Adder encoding first transforms the constraint
  into $y_2 +y_3 + 2y_1 + 2 y_2 + 2 y_4 + 4 y_3 +4 y_4 \leqslant
  9$. Then, using the circuit from Figure
  \ref{figure-PBadders-example}, defines the propositional variables
  $z_0^0, z_1^0, \ldots, z_4^0$. Finally, it enforces $z_0^0 +2 z_1^0
  + \cdots + 16 z_4^0 \leqslant 9$ by adding clauses $\neg z_4^0, \neg
  z_3^0 \vee \neg z_2^0, \neg z_3^0 \vee \neg z_1^0$.
  \begin{figure}[t]
    \begin{center}
      \includegraphics[scale=0.7]{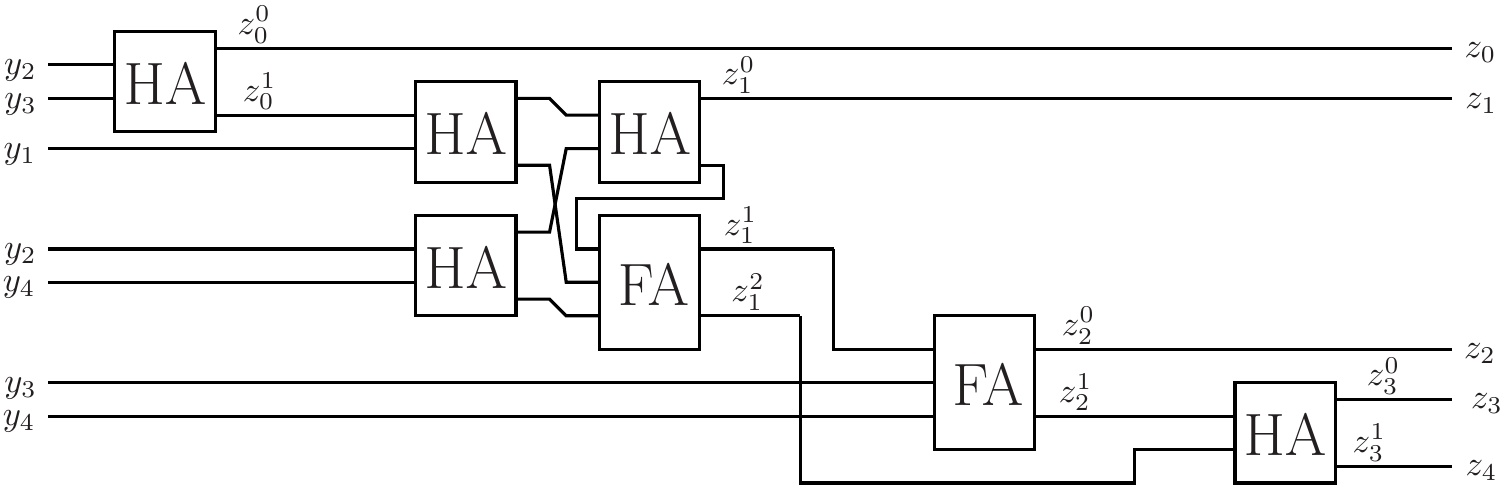}
      \caption{\label{figure-PBadders-example}Adder circuit to encode
        $2y_1 + 3y_2 + 5 y_3 + 6 y_4 \leqslant 9$.}
    \end{center}
  \end{figure}
\end{exmp}

The first encoding of PB constraints using some form of carry-adders was given
by \shortciteA{Warners98}\label{warnerspb}. \shortciteA{Een06} implemented a
similar encoding in
\emph{Minisat+}\label{eenadder},\footnote{\url{github.com/niklasso/minisatp}} a
PB extension to the award winning SAT solver \emph{Minisat}. 

\begin{theorem}[Lemma 2 of~\citeR{Warners98}]
  The Adder encoding for PB presented in this section needs
  $O(\log(a_0) n)$ variables and clauses, but does not maintain
  consistency.
  
\end{theorem}

\subsection{Encoding Pseudo-Boolean Constraints with Sorting
Methods}\label{section-PB-sorting}

In this section we give the basic idea of encodings of Pseudo-Boolean
constraints based on sorting methods~\shortcite{Een06,Bailleux09,Manthey14}.
These encodings are extended to LI constraints in Section
\ref{section-sortingEncoding}, so a more detailed explanation can be found
there.

As in the previous section, given a PB constraint $a_1 y_1 + \ldots, +
a_n y_n \leqslant a_0$, let $(A_i^0, A_i^1, \ldots, A_i^{m})$ be the
binary\footnote{We could use the representation of $a_i$ in any other
  base or even with mixed radix as in the work of \shortciteA{Zha17}. For simplicity, however, in this
  section we consider only binary basis. Section
  \ref{section-sortingEncoding} contains the method and proofs for
  arbitrary bases.} representation of $a_i$. 
Then, the PB constraint
can be written as $$(A_1^0 y_1 + A_2^0 y_2 + \cdots + A_n^0 y_n) +
\cdots + 2^m (A_1^m y_1 + \cdots + A_n^m y_n) \leqslant a_0.$$ We can
define Boolean variables $$(z_0^1, z_0^2, \ldots, z_0^{k_0}) =
\operatorname{Sort}(A_1^0 y_1, A_2^0 y_2, \ldots, A_n^0 y_n),$$ in a
similar way as we did in Section \ref{sect-sorting-CC}. In this case,
$z_0^1 \leqslant z_0^1 \leqslant \cdots \leqslant z_0^{k_0}$
and $$A_1^0 y_1 + A_2^0 y_2 + \cdots + A_n^0 y_n = z_0^1 + z_0^2 +
\cdots + z_0^{k_0}.$$ 

Notice that given any $j > 1$,
$$z_0^{j-1} \leqslant z_0^{j} \quad \Rightarrow \quad z_0^{j-1} +
z_0^{j} = (\neg z_0^{j-1} \wedge z_0^{j}) + 2 z_0^{j-1}.$$ So, if
we define $w_0^{j} = \neg z_0^{j-1} \wedge z_0^{j}$,
\begin{multline*}
  A_1^0 y_1 + A_2^0 y_2 + \cdots + A_n^0 y_n = z_0^1 + z_0^2 +
  \cdots + z_0^{k_0} = \\ = (w_0^{k_0} + w_0^{k_0-2} + \cdots )
  + 2 (z_0^{k_0-1} + z_0^{k_0-3} + \cdots)
\end{multline*}

Also notice that, by construction, $w_0^{k_0} + w_0^{k_0-2} + \cdots
 \leqslant 1$, so $w_0^{k_0} + w_0^{k_0-2} + \cdots = w_0^{k_0} \vee
 w_0^{k_0-2} \vee \cdots$ Let us define $z_0 = w_0^{k_0} \vee
 w_0^{k_0-2} \vee \cdots$.

Then, we have that
\begin{multline*}
  a_1 y_1 + \ldots, + a_n y_n = (A_1^0 y_1 + \cdots + A_n^0 y_n) + 2
  (A_1^1 y_1 + \cdots + A_n^1 y_n) + \cdots = \\ = z_0 + 2 (A_1^1 y_1
  + \cdots + A_n^1 y_n + z_0^{k_0-1} + z_0^{k_0-3} + \cdots) + 2^2 (A_1^2 y_1 + \cdots
\end{multline*}
Repeating this process, we can define $z_1, z_2, \ldots, z_l$ so $a_1
y_1 + \cdots + a_n y_n = z_0 + 2 z_1 + \cdots + 2^l z_l$. With these
variables it is easy to enforce $z_0 + 2 z_1 + \cdots + 2^l z_l
\leqslant a_0$.

In summary, the method defines the variables:
$$\begin{array}{rcl}
  (z_0^1, z_0^2, \ldots, z_0^{k_0}) &=& \operatorname{Sort}(A_1^0 y_1, A_2^0 y_2, \ldots, A_{n+1}^0 y_{n+1}) \\
  (z_1^1, z_1^2, \ldots, z_1^{k_1}) &=& \operatorname{Sort}(A_1^1 y_1, A_2^1 y_2, \ldots, A_{n+1}^1 y_{n+1}, z_0^{k_0-1}, z_0^{k_0-3}, \ldots) \\
  &\cdots \\
  (z_l^1, z_l^2, \ldots, z_l^{k_l}) &=& \operatorname{Sort}(A_1^l y_1, A_2^l y_2, \ldots, A_{n+1}^l y_{n+1}, z_{l-1}^{k_{l-1}-1}, z_{l-1}^{k_{l-1}-3}, \ldots) \\
\end{array}$$
$$w_i^{j} = \neg z_i^{j-1} \wedge z_i^{j}$$
$$z_i = w_i^{k_i} \vee w_i^{k_i-2} \vee \cdots$$
and then encodes $z_0 + 2 z_1 + \cdots + 2^l z_l \leqslant a_0$.

\begin{exmp}
  Let us consider again the constraint $2y_1 + 3y_2 + 5 y_3 + 6 y_4
  \leqslant 9$. As before, the constraint is rewritten as $y_2 +y_3 +
  2y_1 + 2 y_2 + 2 y_4 + 4 y_3 +4 y_4 \leqslant 9$. Then, using the
  circuit from Figure \ref{figure-SN-PB-example}, the encoding defines
  the propositional variables $z_0, z_1, \ldots, z_4$. Finally,
  it enforces $z_0 +2 z_1 + \cdots + 16 z_4 \leqslant 9$ by
  adding clauses $\neg z_4, \neg z_3 \vee \neg z_2, \neg z_3
  \vee \neg z_1$.
  \begin{figure}[t]
    \begin{center}
      \includegraphics[scale=0.7]{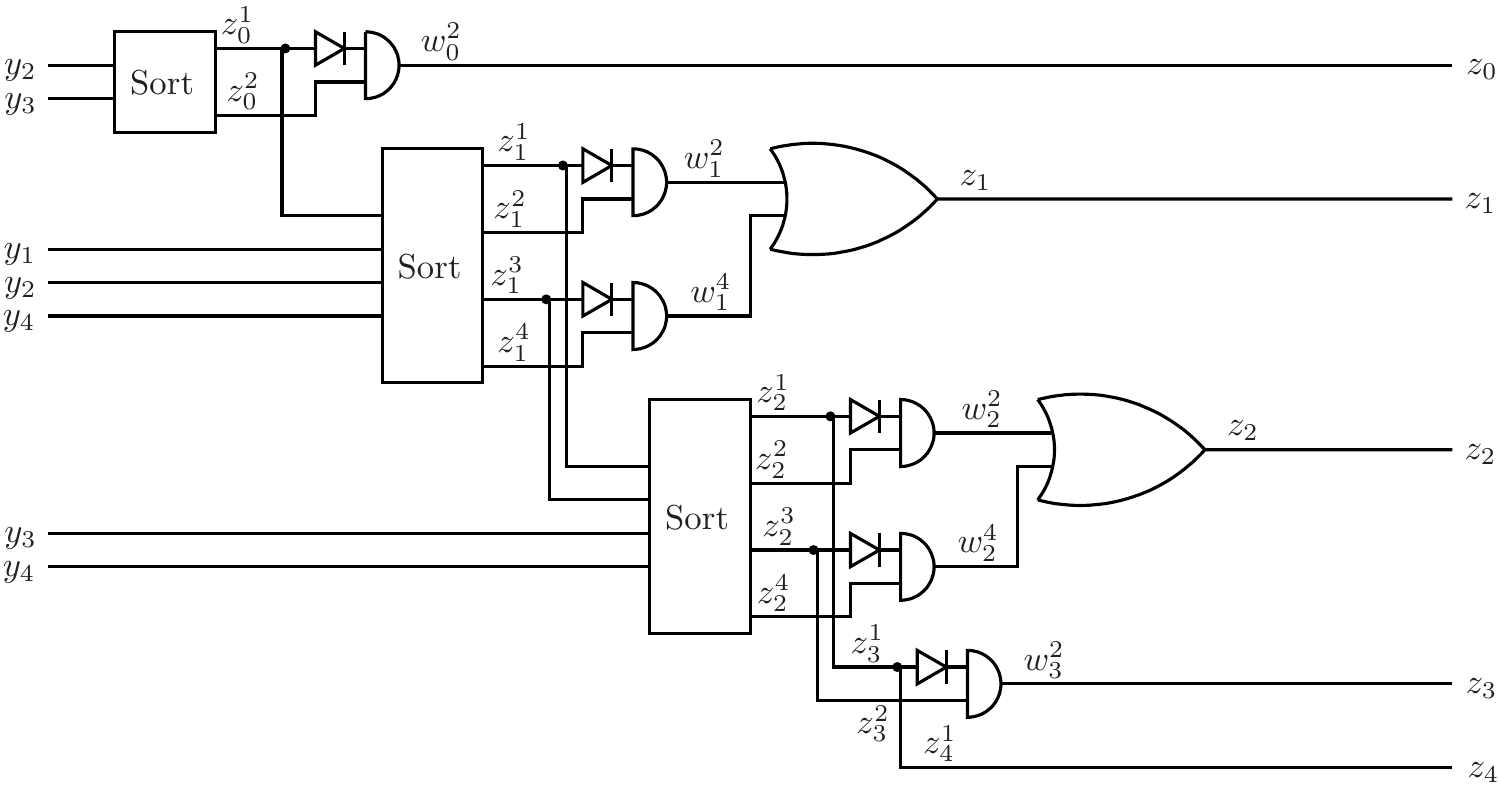}
      \caption{\label{figure-SN-PB-example}Circuit to encode $2y_1 +
        3y_2 + 5 y_3 + 6 y_4 \leqslant 9$ with sorting methods without
        tare.}
    \end{center}
  \end{figure}
\end{exmp}

An alternative method starts by introducing the so-called
\emph{tare}. This is, let $k$ be the minimal positive integer such
that $2^k > a_0$. We can introduce a dummy true variable $y_{n+1}$ and
encode the constraint $a_1 y_1 + \cdots + a_n y_n + (2^k-a_0-1)
y_{n+1} \leqslant 2^k-1$. As before, we introduce the variables
$$\begin{array}{rcl}
  (z_0^1, z_0^2, \ldots, z_0^{l_0}) &=& \operatorname{Sort}(A_1^0 y_1, A_2^0 y_2, \ldots, A_{n+1}^0 y_{n+1}) \\
  (z_1^1, z_1^2, \ldots, z_1^{l_1}) &=& \operatorname{Sort}(A_1^1 y_1, A_2^1 y_2, \ldots, A_{n+1}^1 y_{n+1}, z_0^1, z_0^3, \ldots) \\
  \cdots \\
  (z_k^1, z_k^2, \ldots, z_k^{l_k}) &=& \operatorname{Sort}(A_1^k y_1, A_2^k y_2, \ldots, A_{n+1}^k y_{n+1}, z_{k-1}^1, z_{k-1}^3, \ldots) \\
\end{array}$$
Now the original constraint can be enforced by simply adding the
clause $\neg z_k$. Notice that in this case we do not have to
introduce the variables $w_i^j$ and $z_i$, so the encoding is more
compact. However, the encoding is not incremental: this is, if we
encoded a PB constraint with bound $a_0$ and now we want to encode the
same constraint with a new bound $a_0' < a_0$, we have to re-encode
the constraint from scratch. In the non-tare case, we would only
need to enforce $z_0 + 2 z_1 + \cdots + 2^l z_l \leqslant a_0'$ for
the new value $a_0'$.

\begin{exmp}
  Again, let us consider the constraint $2y_1 + 3y_2 + 5 y_3 + 6 y_4
  \leqslant 9$. The tare is $15-9 = 6$, so we rewrite the constraint
  as $y_2 +y_3 + 2y_1 + 2 y_2 + 2 y_4 + 2 + 4 y_3 +4 y_4 + 4 \leqslant
  15$. Then, using the circuit from Figure
  \ref{figure-SN-tare-example}, the encoding defines the propositional
  variable $z_4$. Now the constraint can be enforced by simply adding
  the clause $\neg z_4$.
  \begin{figure}[t]
    \begin{center}
      \includegraphics[scale=0.7]{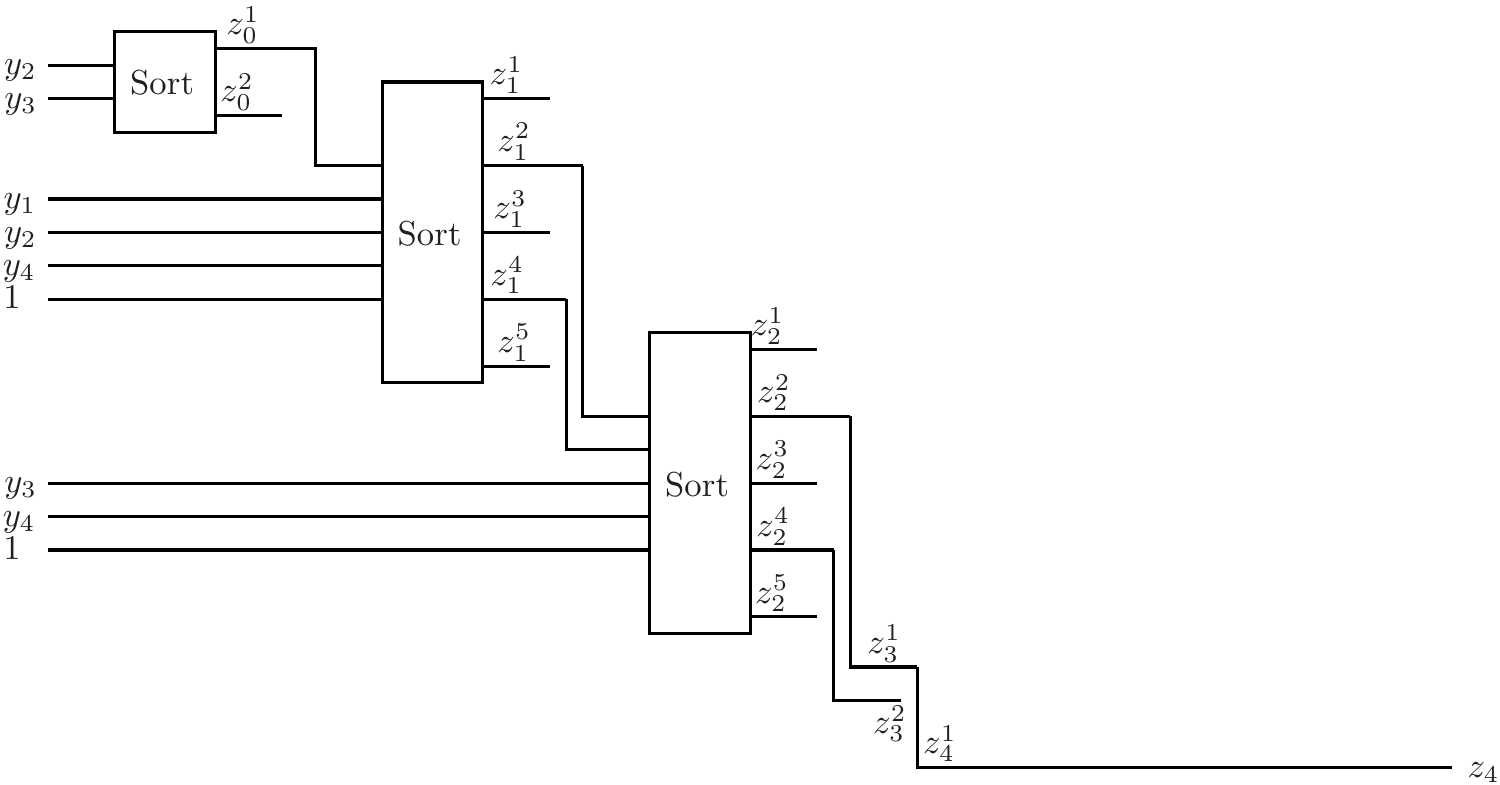}
      \caption{\label{figure-SN-tare-example}Circuit to encode $2y_1 +
        3y_2 + 5 y_3 + 6 y_4 \leqslant 9$ with sorting methods with
        tare.}
    \end{center}
  \end{figure}
\end{exmp}

\citeA{Een06} designed and implemented this method in
MiniSAT+. The coefficients are decomposed in a mixed
radix form and Odd-Even sorting networks are used as the sorting method. They
prove that their method maintains domain consistency on cardinality
constraints, but it is not known what consistency is
maintained in the pseudo-Boolean case. To our
knowledge, this is still an open question. In Section
\ref{section-sortingEncoding} we answer this question: the encoding is
consistent but not domain consistent. The method generates $O(N \log^2
(N))$ variables and clauses, where $N = n \log (\max \{a_i\})$.

Another version of this encoding is presented by
~\shortciteA{Bailleux09}. Their method, called \emph{global polynomial
  watchdog (GPW)}\label{gpw}, adds the tare, decomposes the coefficients in binary and
uses totalizers (see Section \ref{sec:totalizers}) as sorting
method. The number of clauses of this encoding is
$O(n^2\log(n)\log(a_0))$. It is proven that GPW detects inconsistencies
and that through increasing the encoding by a factor of $n$, it also
maintains domain consistency. The extended form is 
referred to as \emph{local polynomial watchdog (LPW)}. 

\shortciteA{Manthey14}\label{manthey} improve the encoding of
~\shortciteA{Bailleux09} by using sorting networks instead of totalizers. Both
constructions have the same consistency but Manthey's save a factor of $n$ in
size. 

\subsection{Encoding Pseudo-Boolean Constraints by Incremental Partial Sums}

In this section we describe the encodings of PB constraints that
introduce the order encoding of the partial sums.

The underlying idea of these encodings is simple. Given a PB constraint
and $i \in \{0, 1, \ldots, n\}$, let us define the integer variable
$s_i = a_1 y_1 + \cdots + a_i y_i$. Since $s_i = s_{i-1} + a_i y_i$,
where $s_0=0$, we can enforce the order encoding of these variables
with the following clauses:
\begin{equation}
\begin{array}{ll}
  (s_{i-1} \geqslant j) \wedge y_i \rightarrow (s_i \geqslant j + a_i), &
  (s_{i-1} \geqslant j)  \rightarrow (s_i \geqslant j,) \\
  (s_{i-1} < j) \wedge \neg y_i \rightarrow (s_i < j), &
  (s_{i-1} < j)  \rightarrow (s_i < j + a_i), \\
\end{array}\label{eq:ps}
\end{equation}
and then simply enforce $s_n \leqslant a_0$.

Notice, however, that the encoding contains some redundant
variables. For instance, the variables $s_1 \geqslant 1$, $s_1
\geqslant 2$, $\ldots$ and $s_1 \geqslant a_1$ are all equivalent. In
fact, $s_i \geqslant j-1$ and $s_i \geqslant j$ are equivalent if
there is no value of $y_1, y_2, \ldots, y_{i-1}$ with $a_1 y_1 + a_2
y_2 +\cdots + a_{i-1} y_{i-1} = j$. Therefore, we can just consider
variables $$\{ s_i \geqslant j \st \exists \ y_1, y_2, \ldots, y_{i-1}
\in \{0,1\}, \text{ with } a_1 y_1 + a_2 y_2 +\cdots + a_{i-1} y_{i-1}
= j \}$$ and obtain an equivalent and more compact encoding.

The encoding of \shortciteA{Bailleux06}\label{b06} computes all the possible values
of the partial sums with diagrams like that shown in
Figure~\ref{fig:bailleux}.

\begin{figure}
$$\xymatrix@R=5mm{
  &&s_2=0\\
  &s_1=0\ar[ru]^{\neg y_2}\ar[r]^{y_2}& s_2=a_2\\
  s_0=0\ar[ru]^{\neg y_1}\ar[dr]^{y_1}&&&\cdots \\
  &s_1=a_1\ar[r]^{\neg y_2}\ar[dr]^{y_2}& s_2=a_1\\
  &&s_2=a_1+a_2\\
  }
$$
\caption{A representation of partial sums calculated by the method of
\shortciteA{Bailleux06}.\label{fig:bailleux}}
\end{figure}
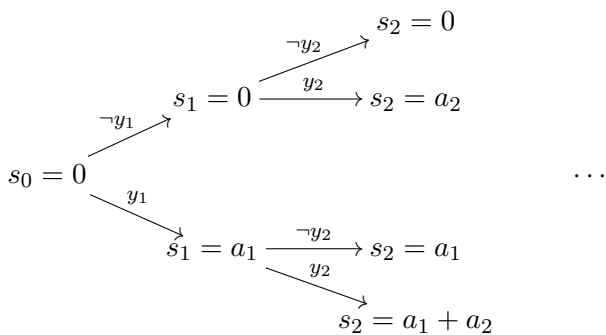
 
The method then adds a Boolean variable and the clauses of Equation~(\ref{eq:ps})
for each of these values. The encoding maintains domain consistency,
but it is exponential in the worst case.

Notice that this diagram is a non-reduced BDD (see Section
\ref{MDD-def}), and, therefore, the encoding of ~\shortciteA{Bailleux06}
still contains redundant variables. The encoding is improved 
by \shortciteA{Een06}\label{eenbdd}, 
where the BDD is reduced before encoding it. The authors
then encode the BDD with the Tseytin transformation by
~\citeA{Tseytin1968},
which requires
one variable and six clauses per node. The method is domain consistent
if the PB was sorted by the coefficients size.

Several improvements were made by \shortciteA{Abio12}\label{abiobdd}. First, a
faster algorithm for constructing the reduced BDD is presented. The authors can
directly generate the reduced BDD, rather than constructing first the
non-reduced one. The authors also introduce a CNF encoding for monotone BDDs
that uses two clauses per BDD node and still maintains
domain consistency via unit propagation. This encoding is more compact
than the encodings used for general BDDs by previous work~\shortcite{Een06,Bailleux06}
and shows improvement for practical SAT solving. The paper has also a
deep analysis on the conditions under which the size of the BDD is
polynomial in the number of literals.

Finally, the work by \shortciteA{BDDsOptimization} generalizes the method of
\shortciteA{Abio12} for optimization functions: the authors study how to
efficiently encode one PB constraint if the constraint with a different bound
was already encoded.

Section \ref{section-MDD} contains a detailed explanation of the generalization
of PB encodings to LI constraints.

\section{Encoding Linear Integer Constraints into SAT}\label{section-LI}

In this section we study several translations of general LI constraints into
CNF. One way to translate LI constraints is to encode the integer variables and
uses PB translations from the previous Section. A different way extends the
encodings from the previous section to integer variables. Following this path
we introduce two new encodings, {\MDD} in Section \ref{section-MDD} and the
family {\SN} of encodings in Section \ref{section-sortingEncoding}. 

Surprisingly, there are few publications that study the direct translations of
non-pseudo-Boolean Linear Constraints to CNF. 

In context of bounded model checking, \shortciteA{Bartzis03} study
the translation of LI constraints to BDDs via reformulating as PBs. We
generalize and discuss this method in Section \ref{section-logarithmic}.

A different method by \shortciteA{Sugar} using an order encoding of integers, and
referred to as the {\Sup} encoding, is explained in detail in Section
\ref{section-sugar}.

Many approaches to encoding linear constraints break the constraint into
component, additions and multiplication by a constant.  Hence
$a_1 x_1 + \cdots + a_n x_n \leq a_0$ is encoded as
$y_i = a_i x_i, 1 \leq i \leq n$ and $y_1 + \cdots + y_n \leq a_0$.
Hence they effectively encode a summation constraint and multiplication by
a constant. The systems below all follow this approach.

FznTini by ~\shortciteA{DBLP:conf/cp/Huang08} uses a  logarithmic encoding of
signed integers. It uses ripple carry adders to encode summation, and shift and
add to encode multiplication by a constant.  FznTini uses an extra bit to check
for overflow of arithmetic operations 

BEE by \shortciteA{DBLP:journals/jair/AMCS13} use an order encoding of integers. 
It encodes summation using odd-even sorting networks, and multiplication by
a constant using repeated addition.   BEE uses a equi-propagation to reason
about relations amongst the Boolean variables created during encoding, and
hence improve the encoding. 

Picat SAT by \shortciteA{Zhou2017} use a sign plus logarithmic encoding of
magnitude to encode integers. Like Fzntini it uses ripple carry adders, and
shift plus add to encode multiplication by a constant. Picat SAT applies
equivalence optimization to remove duplicate Booleans that occur during the
encoding procedure. 

\ignore{
\subsection{Encoding Linear Integer Constraints using Adders}

\label{huang}
Perhaps the most obvious way to encode linear constraints is using a binary
encoding of integers and using ripple carry adders to encode both addition
and multiplication by a constant.  This is the method used by
both FznTini~\cite{DBLP:conf/cp/Huang08} and Picat SAT~\cite{Zhou2017}.
Interestingly more complex adder circuits like carry look ahead adders,
or parallel prefix adders, introduced by circuit designers to make addition
circuits faster, appear to be worse for encoding arithmetic in
SAT.~\cite{ppdp09}. \vale{FIX REFERENCE}

A ripple carry adder encoding the addition of two non-negative $n$-bit logarithmic
integers $x$ and $y$, $z = x + y$, is simply
$$
\begin{array}{rcl}
  (z_b^0, c_1) & = & HA(x_b^0, y_b^0) \\
  (z_b^1, c_2) & = & FA(x_b^1, y_b^1, c_1) \\
  & \vdots  \\
  (z_b^{n-1}, c_n) & = & FA(x_b^{n-1}, y_b^{n-1}, c_{n-1}) \\
\end{array}
$$
where $c_n$ represents the overflow bit. It can be ignored (to implement
fixed width arithmetic), or set to 0 (to force no overflow to occur).

Multiplication by a constant is implemented by binary addition of shifted
inputs.  For example to compute $11 x$ we add $x + 2x + 8x$ where
we compute $2x$ by right-shifting the logarithmic encoding of
$x$ once, and $8x$ by right-shifting it 3 times.
}

\subsection{Encoding Linear Integer Constraints through the Order Encoding}
\label{section-order}
In this section, we describe the encodings of LI constraints that use
the order encoding.  

Let $x$ be an integer variable with domain $[0, d]$, and let
$\ordEnc(x)=[y^1, y^2, \ldots, y^{d}]$ be its order encoding. Notice
that $x = y^1 + y^2 + \cdots + y^d$, since if $x=v$, then $y^1 = y^2 =
\cdots = y^v = 1$ and $y^{v+1} = y^{v+1} = \cdots = y^d =
0$. Therefore, given a LI constraint $C: a_1 x_1 + \cdots + a_n x_n
\leqslant a_0$, we can replace $C$ for the PB constraint $$C': a_1
(y_1^1 + y_1^2 + \cdots + y_1^{d_1}) + \cdots + a_n (y_n^1 + y_n^2 +
\cdots + y_n^{d_n}) \leqslant a_0$$ If we now encode $C'$ with a
standard domain consistent PB encoding, the resulting encoding of $C$
will not be domain consistent. The reason for the loss of domain
consistency can be seen in the next example:

\begin{exmp}
  Let us consider the constraint $C: x_1 + x_2 \leqslant 2$, where
  $x_i \in [0,2]$. Let $y_i^1, y_i^2$ be the order encoding of $x_i$,
  and let us rewrite the constraint as $C': y_1^1 + y_1^2 + y_2^1 +
  y_2^2 \leqslant 2$.

  Let us consider the naive encoding of $C'$ with no auxiliary variables:
  $$\begin{array}{l}
    \neg y_1^1 \vee \neg y_1^2 \vee \neg y_2^1 \\
    \neg y_1^1 \vee \neg y_1^2 \vee \neg y_2^2 \\
    \neg y_1^1 \vee \neg y_2^1 \vee \neg y_2^2 \\
    \neg y_1^2 \vee \neg y_2^1 \vee \neg y_2^2 \\
  \end{array}$$
  This is obviously a domain consistent encoding of $C'$: when two
  variables are set to true, unit propagation sets the other two to
  false. However, the resulting encoding is not a domain consistent
  encoding of $C$: if $y_1^1$ is set to true, unit propagation does
  not propagate anything; however, $y_2^2$ should be set to false
  (this is, if $x_1 \geqslant 1$, then $x_2 \leqslant 1$).

  The reason is that we generated the encoding of $C'$ without
  considering the relations on the variables of the constraint. Notice
  that, since $y_i^j$ are the order encoding of integer variables,
  they satisfy $y_i^2 \rightarrow y_i^1$. If we use these clauses to
  simplify the previous encoding, we obtain:
  $$\begin{array}{l}
    \neg y_1^2 \vee \neg y_2^1 \\
    \neg y_1^2 \vee \neg y_2^2 \\
    \neg y_1^1 \vee \neg y_2^2 \\
  \end{array}$$
  which is a domain consistent encoding of $C$.
\end{exmp}

This example shows that if we want to obtain an encoding with good
properties, we have to consider the clauses from the order encoding of
the integer variables. As we have seen in Section \ref{section-PB},
there are basically two different approaches to 
domain consistent encodings of PB:
using BDDs or using SNs. In this paper we adapt
both of them to LI constraints: the BDD based approach is defined in Section
\ref{section-MDD}, and the SN based approaches are defined in Section
\ref{section-sortingEncoding}.

\subsection{Preprocessing PB and LI constraints}
\label{section-groupingCoeffs}
Before describing the different approaches, there is an easy
preprocessing step that reduces the encoding size without compromising
the propagation strength:

Let us fix an LI constraint $C: a_1 x_1 + \cdots + a_n x_n \leqslant
a_0$. Assume that some coefficients are equal; for simplicity, let us
assume $a_1 = a_2 = \cdots = a_r$. In this case, we can define the
integer variable $s = x_1 + \cdots + x_r$ and decompose the constraint
$C': a_1 s + a_{r+1} x_{r+1} + \cdots + a_n x_n \leqslant a_0$ instead
of $C$. The domain of $s$ is $[0, d_s]$ with $d_s = \min\{a_0 / a_1,
d_1 + \cdots + d_r \}$. 

Notice that we do not need to encode the constraint $s = x_1 + \cdots
+ x_r$ defining the integer variables $s$, instead we can encode $D
\equiv s \geqslant x_1 + \cdots + x_r$ since we are only interested in
lower bounds.  The encoding of $D$ can be done with sorting
networks, which usually gives a more compact encoding than encoding
$D$ as a LI constraint.

In our implementation we use the sorting networks defined by 
\shortciteA{parametricCardinalityConstraint}. However, any other methods
based on computing the order encoding of $x_1 + x_2 + \ldots + x_n$,
(e.g \shortcite{BailleuxBoufkhad2003CP,DBLP:conf/lpar/CodishZ10}),
could be used instead.

In industrial problems where constraints are not randomly generated,
the coefficients have some meaning. Hence it is likely that a large LI
constraint has only a few different coefficients. In this case this
technique can be very effective.

Notice that this method can be used when $C$ is a pseudo-Boolean:
however, $C'$ is still a linear integer constraint. This technique can
be used as a preprocessing method for pseudo-Boolean and linear
constraints. After it, any method for encoding linear integer
constraints, like the ones explained in the following sections, can be
used.

\subsection{Encoding Linear Integer Constraints through MDDs}\label{section-MDD}

In this section we will adapt the BDD construction of PB constraints
of \shortciteA{Abio12} to LI constraints, giving the encoding \MDD. 
First, we develop the key notion
for the MDD construction algorithm, the interval of a node. We then
explain how to construct the MDD and present the encoding of MDDs in
this context. Finally, we extend the method to optimization problems.

In the following, let us fix an LI constraint
$$C: a_1 x_1 + a_2 x_2 + \cdots + a_n x_n \leqslant a_0.$$ Let us
define the Boolean variables $y_i^j$ such that
$$x_i = \ordEnc([y_i^1, y_i^2, \ldots, y_i^{d_i}]), \quad 1 \leqslant i
\leqslant n.$$

\subsubsection{Multi-valued Decision Diagrams}\label{MDD-def}
A directed acyclic graph is called an \emph{ordered Multi-valued Decision
  Diagram} if it satisfies the following properties:
\vspace*{-1mm}
\begin{itemize}
  \item It has two terminal nodes, namely \tnode{} (true) and \fnode{} (false).
  \item Each non-terminal node is labeled by an integer variable $\{x_1, x_2,
    \cdots, x_n\}$. This variable is called \emph{selector variable}.
  \item Every node labeled by $x_i$ has the same number of outgoing
    edges, namely $d_i+1$, each labeled by a distinct number in $\{0, 1,
    \ldots, d_i\}$. 
  \item If an edge connects a node with a selector variable $x_i$ and
    a node with a selector variable $x_j$, then $j > i$.
\end{itemize}
\vspace*{-1mm}
The MDD is \emph{quasi-reduced} if no isomorphic subgraphs exist. It
is \emph{reduced} if, moreover, no nodes with only one child exist. A
\emph{long edge} is an edge connecting two nodes with selector
variables $x_i$ and $x_j$ such that $j > i+1$. In the following we
only consider quasi-reduced ordered MDDs without long edges, and we
just refer to them as MDDs for simplicity.

An MDD represents a function
$$f: \{0,1, \ldots, d_1\} \times \{0, 1, \ldots, d_2\} \times \cdots
\times \{0, 1, \ldots, d_n\} \rightarrow \{0, 1\}$$
in the obvious way. Moreover, given the variable ordering, there is
only one MDD representing that function. For further details about MDDs 
see e.g.~\shortciteA{MDDsPaper}.

An MDD where all the non-terminal nodes have exactly two edges is
called \emph{Binary Decision Diagram} or simply \emph{BDD}. This is,
BDDs are MDDs that non-terminal nodes are labeled by Boolean
variables.

\subsubsection{Motivation for using MDDs}

As seen in Section~\ref{section-order}, using the order encoding of
integer variables, an LI constraint can be replaced by a PB
constraint. This PB could be encoded as in ~\shortcite{Abio12}, but this is
not a good idea since the resulting encoding does not consider the
binary clauses of the order encoding. The next example shows that MDDs
are the natural way to generalize BDDs:

\begin{exmp}
Let us consider $n=2$, $a_1=1$, $d_1=3$, $a_2=2$, $d_2=2$ and
$a_0=4$. Notice that $y_1^1+y_1^2+y_1^3 = x_1$ and
$y_2^1+y_2^2=x_2$. Therefore, we can rewrite the LI constraint as the
pseudo-Boolean
$$C': y_1^1+y_1^2+y_1^3 + 2(y_2^1+y_2^2) \leqslant 4.$$

The BDD of $C'$ defined by \shortciteA{Abio12} is shown in the
upper leftmost diagram of Figure \ref{figure-example-bddmdd}. Notice
that some paths are incompatible with the binary clauses of the order
encoding: for instance, $y_1^1=0, y_1^2=1$ is not possible. If we
remove all the incompatible paths, we obtain the diagram shown at the
upper rightmost diagram of Figure \ref{figure-example-bddmdd}. This
is, in fact, equivalent to the MDD of $C: x_1 + 2x_2 \leqslant 4$
shown in the lowest diagram of Figure \ref{figure-example-bddmdd}.
\begin{figure}[t]
  \begin{center}
    \includegraphics[scale=0.75]{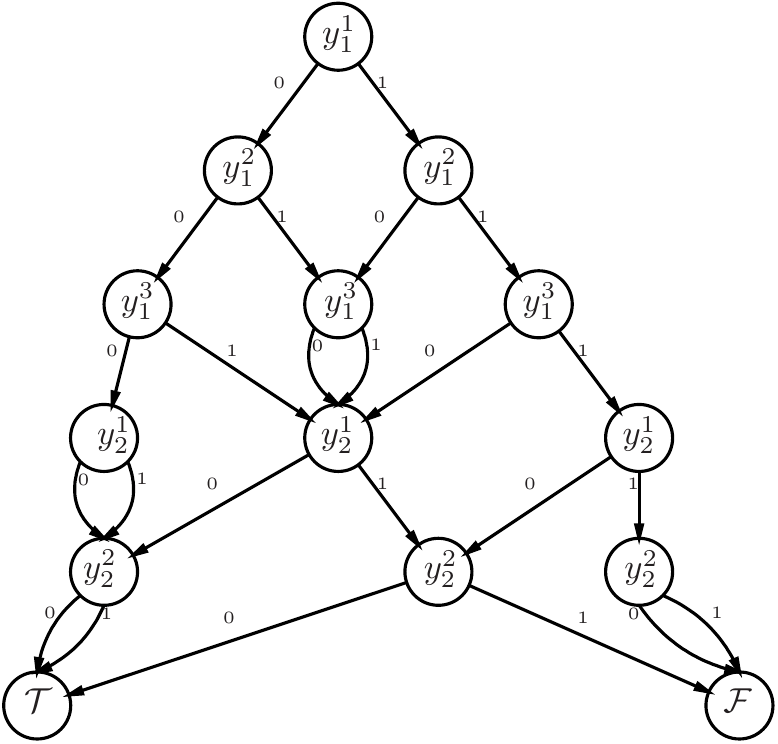}
    \includegraphics[scale=0.75]{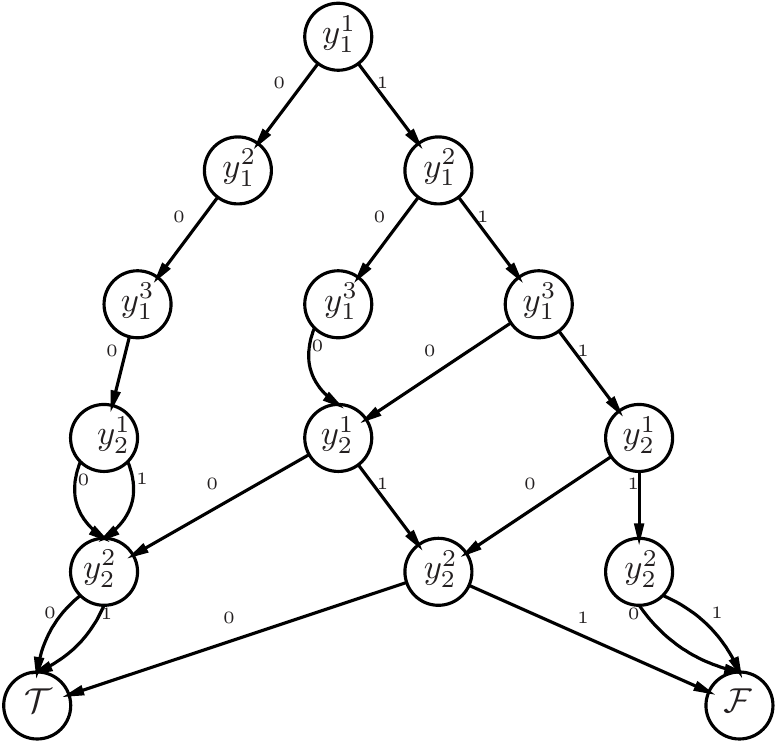}

    \includegraphics[scale=0.75]{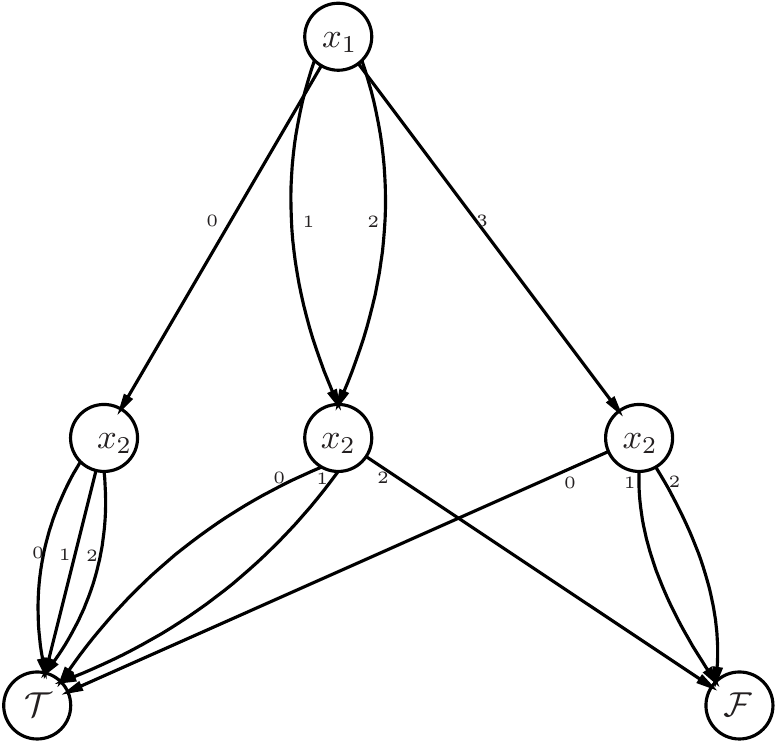}
    \caption{\label{figure-example-bddmdd}Decision Diagrams of $C':
      x_1 + 2x_2 \leqslant 4$.}
  \end{center}
\end{figure}
\end{exmp}

The previous example motivates the use of the MDD for LI
constraints.

\subsubsection{Interval of an MDD Node}
\label{subsection-intervals}
In this section we define the interval of an MDD node. The definition
is very similar to the BDD interval defined by \shortciteA{Abio12}; in
fact, they coincide if the MDD is a BDD (i.e., $d_i=1$ for all $i$).

Let $\mathcal M$ be the MDD of $C$ and let $\nu$ be a node of
$\mathcal M$ with selector variable $x_i$. We define the
\emph{interval} of $\nu$ as the set of values $\alpha$ such that the
MDD rooted at $\nu$ represents the LI constraint $a_i x_i + \cdots +
a_n x_n \leqslant \alpha$. It is easy to see that this definition
corresponds in fact to an interval.

\begin{exmp}\label{ex:1}
  Figure \ref{figure-mdd} contains the MDD of $3x_1 + 2 x_2 + 5x_3
  \leqslant 15$, where $x_1 \in [0,4]$, $x_2 \in [0,2]$ and $x_3 \in
  [0,3]$.
  The root interval is $[15,15]$: this means that the root does not
  correspond to any constraint $3x_1 + 2 x_2 + 5x_3 \leqslant \alpha$,
  apart from $\alpha = 15$. In particular, it means that this
  constraint is not equivalent to $3x_1 + 2 x_2 + 5x_3 \leqslant 14$
  or $3x_1 + 2 x_2 + 5x_3 \leqslant 16$.
  However, the left node with selector variable $x_2$ has interval
  $[15,16]$. 
  This means that $2x_2 + 5x_3 \leqslant 15$ and $2x_2 +
  5x_3 \leqslant 16$ are both represented by the MDD rooted at that
  node. In particular, that means that $2x_2 + 5x_3 \leqslant 15$ and
  $2x_2 + 5x_3 \leqslant 16$ are two equivalent constraints.

\end{exmp}

\begin{figure}[t]
  \begin{center}
    \includegraphics[scale=0.8]{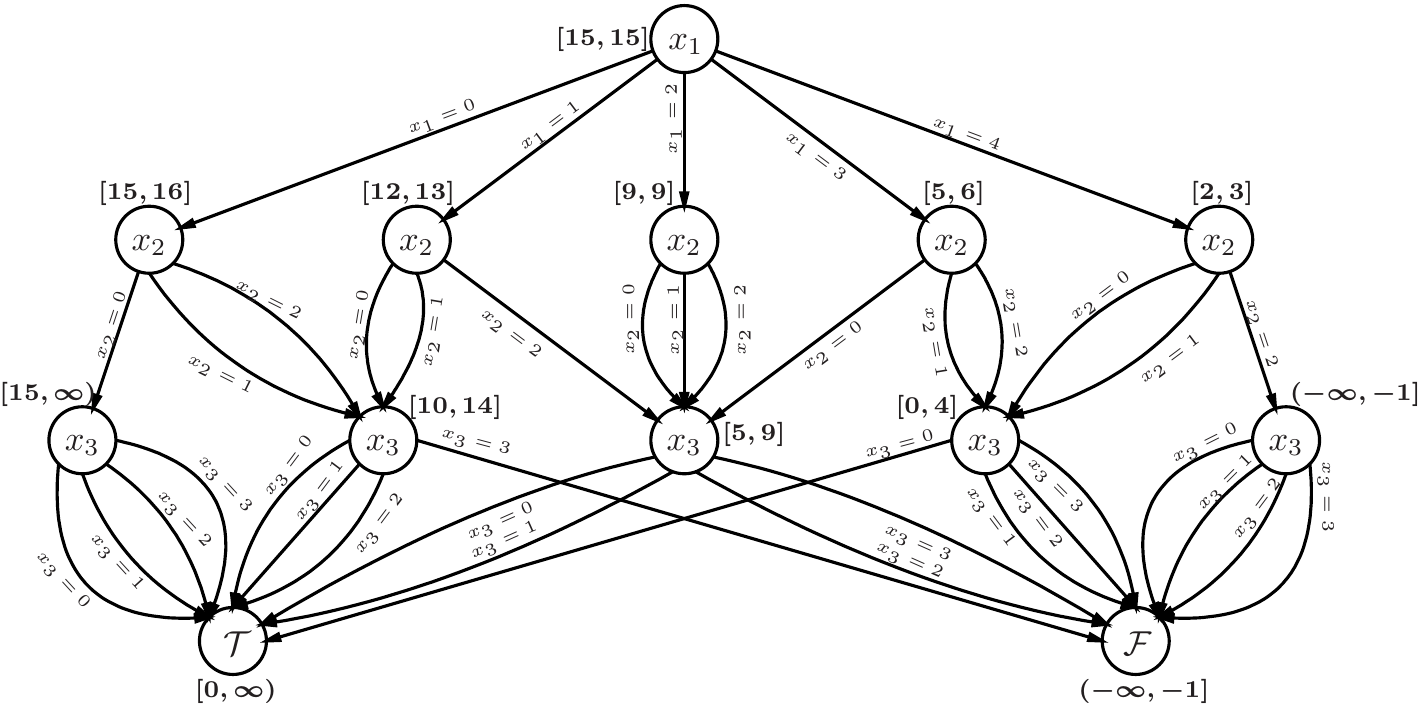}
    \caption{\label{figure-mdd}MDD of $3x_1 + 2 x_2 + 5x_3 \leqslant 15$.}
  \end{center}
  \vspace{-0.5cm}
\end{figure}

The next proposition shows how to compute the intervals of every node:
\begin{proposition}
  \label{prop-intervals}
  Let $\mathcal M$ be the MDD of a LI constraint $a_1 x_1 + \cdots + a_n x_n
  \leqslant a_0$. Then, the following holds:
  \begin{enumerate}
  \item The interval of the true node $\tnode$ is $[0, \infty)$.
  \item The interval of the false node $\fnode$ is $(-\infty, -1]$.
  \item Let $\nu$ be a node with selector variable $x_i$ and
    children $\{\nu_0, \nu_1, \ldots, \nu_{d_i} \}$. Let $[\beta_j,
      \gamma_j]$ be the interval of $\nu_j$. Then, the interval of
    $\nu$ is $[\beta, \gamma]$, with
    $$\beta = \max \{\beta_r + r a_i \st 0 \leqslant r \leqslant d_i \}, \qquad
    \gamma = \min \{\gamma_r + r a_i \st 0 \leqslant r \leqslant d_i \}.$$
  \end{enumerate}
\end{proposition}

\begin{exmp}
  Again, let us consider the constraint $3x_1 + 2 x_2 + 5x_3 \leqslant
  15$, whose MDD is represented at Figure \ref{figure-mdd}. By the
  previous Proposition, $\tnode$ and $\fnode$ have, respectively,
  intervals $[0, \infty)$ and $(-\infty, -1]$. Applying again the same
  proposition, we can compute the intervals of the nodes having $x_3$
  as selector variable. For instance, the interval from the left node is
  $$[0, \infty) \cap [5, \infty) \cap [10, \infty) \cap [15, \infty) =
  [15, \infty),$$
  and the interval from the node having selector variable
  $x_3$ in the middle is
  $$[0, \infty) \cap [5, \infty) \cap (-\infty, 9] \cap (-\infty,
    14] = [5, 9].$$
  After computing all the intervals from the nodes with selector
  variable $x_3$, we can compute the intervals of the nodes with
  selector variables $x_2$ in the same way, and, after that, we can
  compute the interval of the root.
  
\end{exmp}

\subsubsection{Construction of the MDD}
\label{subsection-algorithm-MDD}
In this section we describe an efficient algorithm for building MDDs
given an LI constraint $C$.

The key point of the MDDCreate algorithm, detailed in
Algorithm~\ref{alg:main} and Algorithm~\ref{alg:proc}, is to label
each node of the MDD with its interval $[\beta, \gamma]$.

In the following, for every $i \in \{1, 2, \ldots, n+1\}$, we use
a set $L_i$ consisting of pairs $([\beta, \gamma], \mathcal M)$, where
$\mathcal M$ is the MDD of the constraint $a_i x_i + \cdots + a_n x_n
\leqslant a_0'$ for every $a_0' \in [\beta, \gamma ]$ (i.e., $[\beta,
  \gamma ]$ is the interval of $\mathcal M$). All these sets are
kept in a tuple $\mathcal L = (L_1, L_2, \ldots, L_{n+1})$.

Note that by definition of the MDD's intervals, if both $( [\beta_1,
  \gamma_1], \mathcal{M}_1 )$ and $( [\beta_2, \gamma_2],
\mathcal{M}_2 )$ belong to $L_i$ then either $[\beta_1, \gamma_1] =
        [\beta_2, \gamma_2]$ or $[\beta_1, \gamma_1] \cap [\beta_2,
          \gamma_2] = \emptyset$. Moreover, the first case holds if
        and only if $\mathcal{M}_1 = \mathcal{M}_2$.  Therefore, $L_i$
        can be represented with a \emph{binary search tree-like} data
        structure, where insertions and searches can be done in
        logarithmic time. The function $\search(K, L_i)$ searches
        whether there exists a pair $( [\beta, \gamma], \mathcal M )
        \in L_i$ with $K \in [\beta, \gamma]$. Such a tuple is
        returned if it exists, otherwise an empty interval is returned
        in the first component of the pair. Similarly, we also use
        function $\ins(([\beta,\gamma],\mathcal M),L_i)$ for
        insertions.  The size of the MDD in the worst case is $O(n
        a_0)$ (exponential in the size of the rhs coefficient) and
        algorithm complexity is $O(n w \log w)$ where $w$ is the
        maximum width of the MDD ($w \leq a_0$).

\begin{algorithm}[t]
\caption{Procedure MDDCreate}\label{alg:main}
\begin{algorithmic}[1]
\REQUIRE Constraint $C: a_1 x_1 + \cdots + a_n x_n \leqslant a_0$
\ENSURE returns $\mathcal M$ the MDD of $C$.
\FORALL{$i$ such that $1 \leq i \leq n$}
\STATE $L_i \gets \emptyset$.
\ENDFOR
\STATE $L_{n+1} \gets \Big\{\; \big( (-\infty, -1], \fnode \big),\;\; \big(
[0, \infty), \tnode \big)\; 
    \Big\}\Big\}$.
\STATE $\mathcal L \gets (L_1, \ldots, L_{n+1})$.
\STATE $( [\beta, \gamma], \mathcal M ) \gets \const( 1, a_1 x_1 + \cdots + a_n x_n \leqslant a_0, \mathcal L)$.
\RETURN $\mathcal M$.
\end{algorithmic}
\end{algorithm}

\begin{algorithm}[t]
\caption{Procedure MDDConstruction}\label{alg:proc}
\begin{algorithmic}[1]
\REQUIRE $i \in \{1, 2, \ldots, n+1 \}$, constraint $\tilde C: a_i x_i + \cdots + a_n x_n \leqslant a_0'$ and tuple $\mathcal L$
\ENSURE returns $[\beta,\gamma]$ interval of $\tilde C$ and $\mathcal M$ its MDD
\STATE $( [\beta, \gamma], \mathcal M ) \gets \search(a_0', L_i)$.
\IF{$[\beta,\gamma]\neq \emptyset$}
\RETURN $( [\beta,\gamma], \mathcal M)$.
\ELSE 
\FORALL{$j$ such that $0 \leq j \leq d_i$}
\STATE $( [\beta_j, \gamma_j], \mathcal{M}_j ) \gets \const ( i+1, a_{i+1} x_{i+1} + \cdots + a_n x_n \leqslant a_0' - j a_i, \mathcal L)$.
\ENDFOR
\STATE $\mathcal M \gets$ {\bf mdd}$(x_i,[{\mathcal M}_0,\ldots,{\mathcal M}_{d_i}])$.
\STATE $[\beta,\gamma] \gets [\beta_0, \gamma_0] \cap [\beta_1 + a_1, \gamma_1 + a_1] \cap \cdots \cap [\beta_{d_i} + d_i a_i, \gamma_{d_i} + d_i a_1]$.
\STATE $\ins ( ( [\beta, \gamma], \mathcal M ), L_i)$. 
\RETURN $( [\beta,\gamma], \mathcal M )$.
\ENDIF
\end{algorithmic}
\end{algorithm}

The MDD creation works by initializing the $\mathcal L$ data structure for
the terminal nodes $\tnode$ and $\fnode$, and the calling the MDD
construction function.  This checks if the MDD requires is already in the
$\mathcal L$ structure in which case it is returned, otherwise if
recursively builds the child MDDs for this node (adjusting the rhs of the
constraint appropriately), and then constructs the MDD for this call using
the function \textbf{mdd} to construct and nodes labelled by $x_i$ and the
child MDDs. It calculates the interval for this MDD and returns the interval
and MDD.

\subsubsection{Encoding MDDs into CNF}
\label{subsection-encoding-MDDs-into-CNF}
In this section we generalize the encoding for monotonic BDDs
described by \shortciteA{Abio12} to monotonic MDDs. The encoding assumes
that the selector variables are encoded with the order encoding.

Let $\mathcal M$ be an MDD with the variable ordering $[x_1, \ldots,
  x_n]$. Let $[0, d_i]$ be the domain of the $i$-th variable, and let
$\{y_i^1, \ldots, y_i^{d_i} \}$ be the variables of the order encoding
of $x_i$ (i.e., $y_i^j$ is true iff $x_i \geqslant j$). Let $\mu$ be
the root of $\mathcal M$, and let $\tnode$ and $\fnode$ be
respectively the true and false terminal nodes. In the following,
given a non-terminal node $\nu$ of $\mathcal M$, we define $\sv(\nu)$
as the selector variable of $\nu$, and $\child(\nu, j)$ as the $j$-th
child of $\nu$, i.e., the child of $\nu$ defined by $\{ \sv(\nu) = j
\}$.

The encoding introduces the variables $\{ z_\nu \st \nu \in
\mathcal M \}$; and the clauses
\begin{multline*}
\{ z_\tau, \ \neg z_{\fnode} \} \cup \Big\{ \ \neg z_\nu \lor \neg
y_i^{j} \lor z_{\nu'} \st \nu \in {\mathcal M} \setminus \{\tnode,
\fnode\}, \\ \sv(\nu)=x_i, \ 0 \leqslant j \leqslant d_i, \ \nu' =
\child(\nu,j) \Big\}, \end{multline*} where $y_i^0$ is a dummy true
variable ($x_i \geqslant 0$ since $x_i \in [0, d_i]$). This encoding
will be denoted by $\mddenc(\mu)$.

Notice that this encoding coincides with the BDD encoding
of \shortciteA{Abio12} if the MDD is a BDD.

\begin{lemma} \label{lemma-mdd-consistency}
  Let $A = \{ x_j \geqslant v_j \}_{j=i}^n$ be a partial assignment on
  the last variables. Let $\nu$ be a node of $\mathcal M$ with
  selector variable $x_i$.

  Then, $\mddenc(\mu)$ and $A$ propagates (by unit propagation) $\neg
  z_\nu$ if and only if $A$ is incompatible with $\nu$ (this is, the
  constraint defined by an MDD rooted at $\nu$ does not have any
  solution satisfying $A$).
\end{lemma}
Notice that the Lemma implies that $\mddenc(\mu)$ is consistent.

\begin{theorem}\label{th-consistency}
  Unit propagation on $\mddenc(\mu)$ is domain consistent.
\end{theorem}
Proofs are in the Appendix. 

\begin{exmp}
  Let us consider the MDD represented in Figure \ref{figure-mdd}. The
  encoding introduces the variables $z_1, z_2, \ldots, z_{11},
  z_{\tnode}, z_{\fnode}$, one for each node of the MDD; and the following clauses:
\vspace*{-2mm}
  $${\fontsize{8pt}{1em}\selectfont
    \begin{array}{l @{\hskip 15pt} l @{\hskip 15pt} l @{\hskip 15pt} l}
    z_1,&
    z_{\tnode}, &
    \neg z_{\fnode}, &
    \neg z_1 \lor z_2, \\

    \neg z_1 \lor x_1 \leqslant 0 \lor z_3, &
    \neg z_1 \lor x_1 \leqslant 1 \lor z_4, &
    \neg z_1 \lor x_1 \leqslant 2 \lor z_5, &
    \neg z_1 \lor x_1 \leqslant 3 \lor z_6, \\

    \neg z_2 \lor z_7, &
    \neg z_2 \lor x_2 \leqslant 0 \lor z_8, &    
    \neg z_2 \lor x_2 \leqslant 1 \lor z_8, &
    \neg z_3 \lor z_8, \\

    \neg z_3 \lor x_2 \leqslant 0 \lor z_8, &
    \neg z_3 \lor x_2 \leqslant 1 \lor z_9, &
    \neg z_4 \lor z_9, &
    \neg z_4 \lor x_2 \leqslant 0 \lor z_9, \\

    \neg z_4 \lor x_2 \leqslant 1 \lor z_9, &
    \neg z_5 \lor z_9, &
    \neg z_5 \lor x_2 \leqslant 0 \lor z_{10}, &
    \neg z_5 \lor x_2 \leqslant 1 \lor z_{10}, \\

    \neg z_6 \lor z_{10}, &
    \neg z_6 \lor x_2 \leqslant 0 \lor z_{10}, &
    \neg z_6 \lor x_2 \leqslant 1 \lor z_{11}, &
    \neg z_7 \lor z_{\tnode}, \\

    \neg z_7 \lor x_3 \leqslant 0 \lor z_{\tnode}, &
    \neg z_7 \lor x_3 \leqslant 1 \lor z_{\tnode}, &
    \neg z_7 \lor x_3 \leqslant 2 \lor z_{\tnode}, &
    \neg z_8 \lor z_{\tnode}, \\

    \neg z_8 \lor x_3 \leqslant 0 \lor z_{\tnode}, &
    \neg z_8 \lor x_3 \leqslant 1 \lor z_{\tnode}, &
    \neg z_8 \lor x_3 \leqslant 2 \lor z_{\fnode}, &
    \neg z_9 \lor z_{\tnode}, \\

    \neg z_9 \lor x_3 \leqslant 0 \lor z_{\tnode}, &
    \neg z_9 \lor x_3 \leqslant 1 \lor z_{\fnode}, &
    \neg z_9 \lor x_3 \leqslant 2 \lor z_{\fnode}, &
    \neg z_{10} \lor z_{\tnode}, \\

    \neg z_{10} \lor x_3 \leqslant 0 \lor z_{\fnode}, &
    \neg z_{10} \lor x_3 \leqslant 1 \lor z_{\fnode}, &
    \neg z_{10} \lor x_3 \leqslant 2 \lor z_{\fnode}, &
    \neg z_{11} \lor z_{\fnode}, \\

    \neg z_{11} \lor x_3 \leqslant 0 \lor z_{\fnode}, &
    \neg z_{11} \lor x_3 \leqslant 1 \lor z_{\fnode}, &
    \neg z_{11} \lor x_3 \leqslant 2 \lor z_{\fnode}. \\
  \end{array}}$$

In essence, for every node an auxiliary variable is introduced and for
each edge a clause. Notice that some clauses are redundant. This issue
is handled in the following Section.
\end{exmp}

\paragraph{Removing Subsumed Clauses}
The MDD encoding explained here can easily be improved by removing
some unnecessary clauses. We apply the following rule when producing
the encoding:

Given a non-terminal node $\nu$ with $\sv(\nu) = x_i$, if
$\child(\nu,j) = \child(\nu, j-1)$, then the clause $\neg z_z\nu \lor
y_i^{j-1} \lor z_{\nu'}$ is subsumed by the clause $\neg z_\nu \lor
y_i^{j-2} \lor z_{\nu'}$; therefore, we can remove it.

Additionally, we also improve the encoding by reinstating long edges (since the
dummy nodes used to eliminate long edges do not provide any information); that
is, we encode the reduced MDD instead of the quasi-reduced MDD.

\begin{exmp}
  Let us consider again the MDD represented in Figure
  \ref{figure-mdd}. The encoding introduces the variables $z_1, z_2,
  z_3, z_5, z_6, z_8, z_9, z_{10}, z_{\tnode}, z_{\fnode}$, one for
  each non-dummy node of the MDD; and the following clauses:
\vspace*{-2mm}
  $${\fontsize{8pt}{1em}\selectfont
    \begin{array}{l @{\hskip 15pt} l @{\hskip 15pt} l @{\hskip 15pt} l}
    z_1,&
    z_{\tnode}, &
    \neg z_{\fnode}, &
    \neg z_1 \lor z_2, \\

    \neg z_1 \lor x_1 \leqslant 0 \lor z_3, &
    \neg z_1 \lor x_1 \leqslant 1 \lor z_9, &
    \neg z_1 \lor x_1 \leqslant 2 \lor z_5, &
    \neg z_1 \lor x_1 \leqslant 3 \lor z_6, \\

    \neg z_2 \lor z_{\tnode}, &
    \neg z_2 \lor x_2 \leqslant 0 \lor z_8, &    
    \neg z_3 \lor z_8, &
    \neg z_3 \lor x_2 \leqslant 1 \lor z_9,\\

    \neg z_5 \lor z_9, &
    \neg z_5 \lor x_2 \leqslant 0 \lor z_{10}, &
    \neg z_6 \lor z_{10}, &
    \neg z_6 \lor x_2 \leqslant 1 \lor z_{\fnode}, \\

    \neg z_8 \lor z_{\tnode}, &
    \neg z_8 \lor x_3 \leqslant 2 \lor z_{\fnode}, &
    \neg z_9 \lor z_{\tnode}, &
    \neg z_9 \lor x_3 \leqslant 1 \lor z_{\fnode}, \\

    \neg z_{10} \lor z_{\tnode}, &
    \neg z_{10} \lor x_3 \leqslant 0 \lor z_{\fnode}, \\
  \end{array}}$$

\end{exmp}

\subsubsection{Encoding Objective Functions with MDDs}
\label{section-optimization-mdds}
In this section we describe how to deal with combinatorial problems
where we minimize a linear integer objective function.  A similar
idea is used by \shortciteA{BDDsOptimization}, 
where the authors use BDDs
for encoding problems with pseudo-Boolean objectives.
Combinatorial optimization problems can be efficiently solved with a
branch-and-bound strategy. In this way, all the lemmas learned in the
previous steps are reused for finding the next solutions or proving
the optimality. For implementing branch-and-bound, we need to be
able to create a decomposition of 
the constraint $a_1 x_1 + \cdots + a_n x_n \leqslant a_0'$ from 
the decomposition of $a_1 x_1 + \cdots + a_n x_n \leqslant
a_0$ where $a'_0 < a_0$.

This is easy for cardinality constraints, since, when we have encoded
a constraint $x_1 + \cdots + x_n \leqslant a_0$ with a sorting
network, we can encode $x_1 + \cdots + x_n \leqslant a_0'$ by
adding a single clause see ~\shortcite<see>{AsinNOR11}.
 
\begin{exmp}\label{example-CC}
Let us consider the easier case of a cardinality constraint
objective. Assume we want to find a solution of a formula $F$ that
minimizes the function $x_1 + \cdots + x_n$, where $x_i$ are Boolean
variables.

First, we launch a SAT solver with the input formula $F$. After
finding a solution of cost $K$, the constraint $x_1 + \cdots + x_n
\leqslant K-1$ must be added. The encodings of cardinality constraints
based on sorting networks introduce some variables $y_1, \ldots,
y_{K}$, where $y_i \equiv x_1 + \cdots + x_n \geqslant i$. Let $S$ be
such an encoding. We then launch the SAT solver with $F \cup S \cup \{
\neg y_K \}$.

If now the SAT solver finds a solution of cost $K'$, we just have to
add the clause $\neg y_{K'}$. 
\end{exmp}

Example~\ref{example-CC} shows that in optimization problems we do not
have to re-encode the new constraints from scratch: we should reuse
as much as the previous encodings as possible. In this way, not only
do we generate fewer clauses and variables, but, more importantly, all the
learned clauses about the previous encodings can be reused.

In order to reuse the previous encodings for the MDD encoding of an LI
constraint, we have to save the tuple $\mathcal L$ used in Algorithm
\ref{alg:main}. When a new solution of cost $a'_0+1$ is found,
Algorithm \ref{alg:opt} is called.

\begin{algorithm}[t]
\caption{MDD Construction: Optimization version}\label{alg:opt}
\begin{algorithmic}[1]
\REQUIRE Constraint $C: a_1 x_1 + \cdots + a_n x_n \leqslant a'_0$ and tuple $\mathcal L$.
\ENSURE returns $\mathcal M$ the MDD of $C$.
\STATE $( [\beta, \gamma], \mathcal M ) \gets \const( 1, a_1 x_1 + \cdots + a_n x_n \leqslant a'_0, \mathcal L)$.
\RETURN $\mathcal M$.
\end{algorithmic}
\end{algorithm}

\begin{theorem}
  Algorithm \ref{alg:opt} provides a domain consistent encoding of the
  LI constraint $C$. The sum of all variables created by any call of
  Algorithm \ref{alg:opt} is bound by $n a_0$, and the number of
  clauses is bound by $n a_0 d$, where $d = \max \{d_i\}$ and $a_0$ is
  the cost of the first solution found.
\end{theorem}
\begin{proof}
  The encoding is domain consistent due to Theorem
  \ref{th-consistency}. Notice that the encoding creates at most one
  variable for every element of $L_i \in \mathcal L$, $1 \leqslant i
  \leqslant n$. Therefore, after finding optimality, the encoding has
  generated at most $n a_0$ variables in total. In the same way, the
  number of clauses generated can be bounded by $n a_0 d$.
\end{proof}

In practice the optimization version is very useful. The new MDD
construction typically only adds a few nodes near the top of the MDD, and
then reuses nodes below. 

\subsection{Encoding Linear Integer Constraints through Sorting 
Networks}\label{section-sortingEncoding}

In this section we introduce the methods \SNTARE{} and \SNOPT{} to encode LI
constraints using Sorting Networks. We prove that they maintain consistency and
discuss their size. 

\subsubsection{Background}

The encodings in this section are a generalization of previous work encoding
threshold functions to monotone circuits as in \shortciteA{Beimel06} and PB
constraints to SAT \shortcite{Een06,Manthey14}, that use different types of
sorting networks to encode pseudo-Boolean constraints, explained in Section
\ref{section-PB-sorting}. 

All these encodings work more or less in the same way: given a
pseudo-Boolean constraint
$$a_1 x_1 + a_2 x_2 + \cdots + a_n x_n \leqslant a_0$$ and an integer
number $b>1$, let $y_0, y_1, \ldots, y_m$ be the digits of $s=a_1 x_1
+ a_2 x_2 + \cdots + a_n x_n$ in base $b$ (or, in fact, with a fixed
mixed radix). The methods introduce the Boolean variables $y_j^i$
corresponding to the order encoding of $y_j$, and then encode $y_0 + b
y_1 + \cdots + b^m y_m \leqslant a_0$.

As we have seen in Section \ref{section-PB-sorting}, there are two ways to
encode PB constraints with sorting networks: either adding the tare or not.

When using a tare the methods add a dummy true variable $x_{n+1}$ with
coefficient $a_{n+1}=b^{m+1}-a_0-1$ such that the bound in the constraint is
$b^{m+1}-1$. In this case, the encoding is more compact, but it is not
incremental. We generalize the tare case to LI constraints in Sections
\ref{subsection-powerOfB} and \ref{subsection-general} and introduce encoding
\SNTARE. The non-tare case, needed to encode objective functions, is studied in
Section \ref{section-optimization-networks} and referred to as \SNOPT. 

These methods are consistent but not domain consistent. Our
implementation uses merge and simplified-merge networks
\shortcite{parametricCardinalityConstraint}, but any domain-consistent encoding of
sorting networks can be used instead.

\subsubsection{Encoding LI Constraints with Logarithmic Coefficients.}
\label{subsection-powerOfB}

First, let us consider the simpler case where all the coefficients
have a single digit in a fixed base $b>1$, and the bound is
$a_0=b^{m+1}-1$ for some integer value $m$. In the next section we
show that a general LI constraint can be reduced to this case.

  




Let us consider the constraint
$$\begin{array}{lccccccccc}
C:       & \aline{0}                  & +         &   \\
         & \aline{1}                  & +         &   \\
         & \multicolumn{7}{c}{\ldots} & +         &   \\
         & \aline{m}                  & \leqslant & b^{m+1}-1,
\end{array}$$
where the variables $x_{i,j}$ are integer with domain $[0, d_{i,j}]$ and $0 \leqslant A_{i,j} <
b.$

Given such a constraint, let us define
$$y_j = \begin{cases}
\alineP{0} &\text{if } j = 0 \\ 
\left \lfloor{\frac{y_{j-1}}{b}}\right \rfloor +  \alineP{j} &\text{if } j > 0 \\ 
\end{cases}$$

\begin{proposition}\label{prop-encoding-networks}
  \begin{enumerate}
    \item Given an integer $z$ with $\ordEnc(z)=[z^1, z^2, \ldots, z^d]$, then
      $$\ordEnc(\left \lfloor{\frac{z}{b}}\right \rfloor) = [z^b,
      z^{2b}, \ldots, z^{b \left \lfloor{\frac{d}{b}}\right
        \rfloor}]$$ is a domain consistent encoding of $\left
      \lfloor{\frac{z}{b}}\right \rfloor$.
      \item Given integers $z$ and $a$ with $\ordEnc(z) = [z^1, z^2,
        \ldots, z^d]$ and $a>0$,
        $$\ordEnc(az) = [\overbrace{z^1, z^1, \ldots, z^1}^{(a)},
        \overbrace{z^2, z^2, \ldots, z^2}^{(a)}, \ldots,
        \overbrace{z^d, z^d, \ldots, z^d}^{(a)}]$$
        is a domain consistent encoding of $az$.
      \item Given $n$ integer variables $z_1, z_2, \ldots, z_n$ with $\ordEnc(z_i)
        = [z_i^1, z_i^2, \ldots, z_i^{d_i}]$,
        \begin{multline*}
          \ordEnc(z_1+z_2+\cdots+z_n)=\\=\sn(z_1^1, z_1^2, \ldots, z_1^{d_1},
        z_2^1, z_2^2, \ldots, z_2^{d_2}, \ldots, z_n^{d_n})
        \end{multline*}
        is a domain consistent encoding of $z_1+z_2+\cdots+z_n$. 
  \end{enumerate}
\end{proposition}
The proof of this proposition is trivial using the definitions of
order encoding and sorting networks.

Now, given $(i,j)\in[1,n]\times[0,m]$, let us define the tuple $X_{i,j}$ as
$$X_{i,j} = \left (\overbrace{x_{i,j}^1, x_{i,j}^1, \ldots,
  x_{i,j}^1}^{(A_{i,j})}, \overbrace{x_{i,j}^2, x_{i,j}^2, \ldots,
  x_{i,j}^2}^{(A_{i,j})}, \ldots, \overbrace{x_{i,j}^{d_{i,j}},
  x_{i,j}^{d_{i,j}}, \ldots, x_{i,j}^{d_{i,j}}}^{(A_{i,j})} \right ),$$ where
$(x_{i,j}^1, \ldots, x_{i,j}^{d_{i,j}}) = \ordEnc(x_{i,j})$.

The encoding of this section introduces Boolean variables $y_i^j$ defined as
$$\begin{array}{l}
  (y_{0}^1, y_{0}^2, \ldots) = \sn( X_{1,0}, X_{2,0}, \ldots, X_{n,0}) \\
  (y_{1}^1, y_{1}^2, \ldots) = \sn( y_0^b, y_0^{2b}, \ldots, X_{1,1}, X_{2,1}, \ldots, X_{n,1}) \\
  \cdots \\
  (y_{m}^1, y_{m}^2, \ldots) = \sn( y_{m-1}^b, y_{m-1}^{2b}, \ldots, X_{1,m}, X_{2,m}, \ldots, X_{n,m}) \\
\end{array}$$

\begin{lemma}\label{lemma-yiff}
  Let $A = \{ x_{i,j} \geqslant v_{i,j} \}_{1 \leqslant i \leqslant n, 0
    \leqslant j \leqslant m}$ be an assignment. Then,
  $$\sum\limits_{j=0}^m b^j \Big( A_{1,j} v_{1,j} + A_{2,j} v_{2,j} +
  \cdots + A_{n,j} v_{n,j} \Big) > b^{m+1}-1$$ if and only if
  $y_m^{b}$ is propagated to true.
\end{lemma}

Finally, the encoding introduces the clause $\neg y_m^b$. By the
previous lemma, the encoding is consistent.

\begin{exmp}\label{example-withSN-powerOfB}
  Let us fix $b=3$, and consider the constraint
  $$C: 2 x_{1,0} + 2 x_{2,0} + 2x_{3,0} + 3 x_{1,1} + 3 x_{2,1} + 9
  x_{1,2} \leqslant 26,$$ where $x_{1,0} \in [0,2]$, $x_{2,0} \in
  [0,3]$, $x_{3,0} \in [0,1]$, $x_{1,1} \in [0,4]$, $x_{2,1} \in
  [0,3]$ and $x_{1,2} \in [0,1]$. Notice that $A_{3,1} = A_{2,2} =
  A_{3,2} = 0$. Let us denote $(x_{i,j}^1,x_{i,j}^2, \ldots,
  x_{i,j}^{d_{i,j}}) = \ordEnc(x_{i,j})$ as usual.

  The encoding of this section defines Boolean variables
  $$\{y_0^j \st 1 \leqslant j \leqslant 12 \} \cup \{y_1^j \st 1
  \leqslant j \leqslant 11 \} \cup \{y_2^j \st 1 \leqslant j \leqslant
  4 \}$$
  as
  \begin{multline*}
    (y_0^1, y_0^2, \ldots, y_0^{12}) =\\ = \sn(x_{1,0}^1, x_{1,0}^1,
    x_{1,0}^2, x_{1,0}^2, x_{2,0}^1, x_{2,0}^1, x_{2,0}^2, x_{2,0}^2,
    x_{2,0}^3, x_{2,0}^3,x_{3,0}^1, x_{3,0}^1),
  \end{multline*}
  $$(y_1^1, y_1^2, \ldots, y_1^{11}) = \sn(y_0^3,y_0^6,y_0^9,
  y_0^{12},x_{1,1}^1, x_{1,1}^2, x_{1,1}^3, x_{1,1}^4, x_{2,1}^1,
  x_{2,1}^2,x_{2,1}^3),$$
  
  $$(y_2^1, y_2^2, y_2^3, y_2^{4}) = \sn(y_1^3,y_1^6,y_1^9, x_{1,2}^1),$$

  and the clause $\neg y_2^3$.

  Figure \ref{figure-example-sn-powerOfB} shows the definition of
  these variables. Notice that, by Proposition
  \ref{prop-encoding-networks}, $(y_i^1, y_i^2, \ldots, ) = \ordEnc(y_i)$.

  \begin{figure}[t]
    \begin{center}
      \includegraphics[scale=0.42]{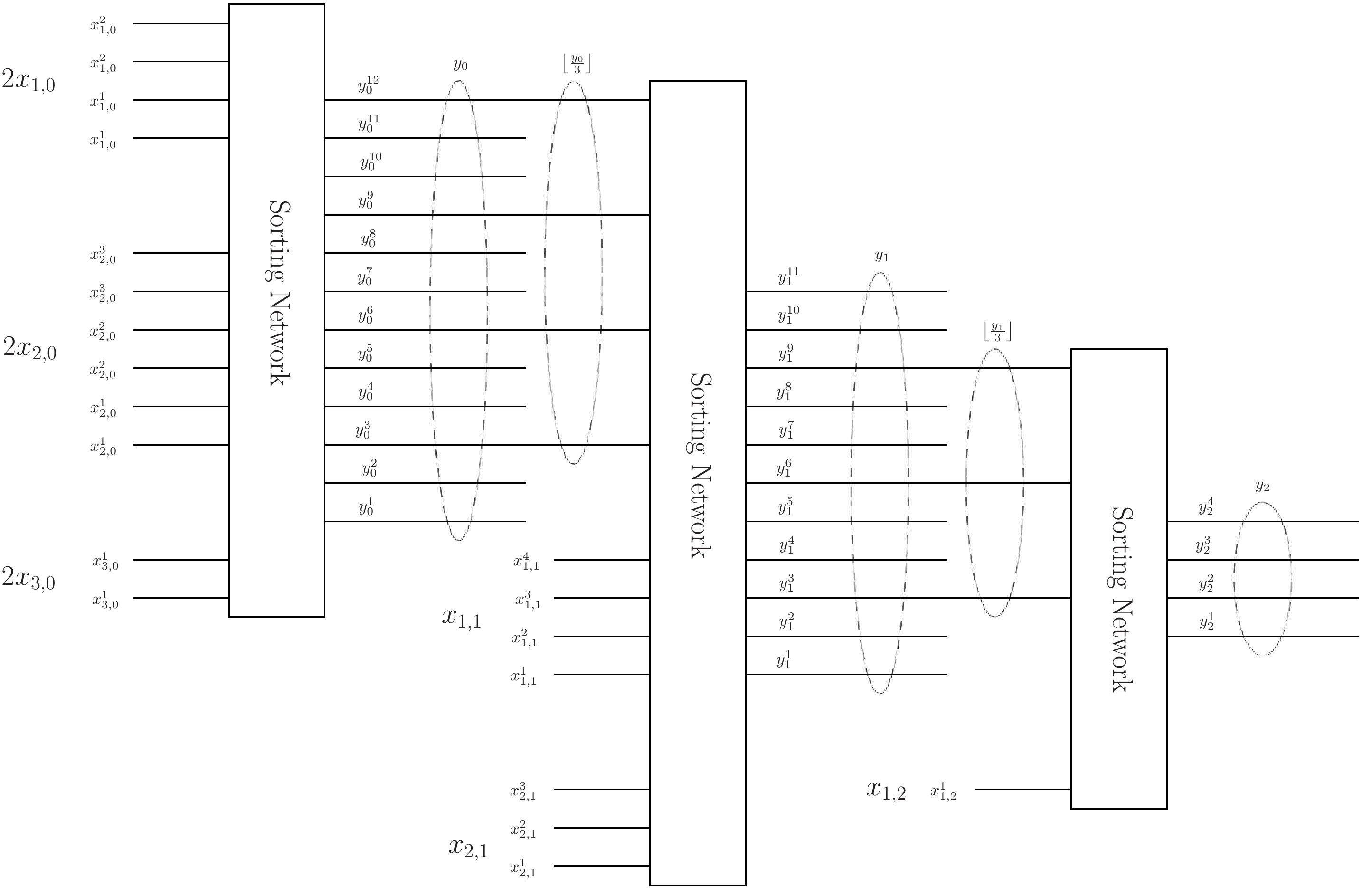}
      \caption{\label{figure-example-sn-powerOfB} Network of the constraint $2
        x_{1,0} + 2 x_{2,0} + 2x_{3,0} + 3 x_{1,1} + 3 x_{2,1} + 9
        x_{1,2} \leqslant 26$.}
    \end{center}
  \end{figure}

  The encoding maintains consistency: for instance, given the partial assignment
  $$\{x_{1,0}=1, x_{2,0}\geqslant 1, x_{1,1} \geqslant 3, x_{2,1} = 3,
  x_{3,1} = 1\}$$ the first sorting network has 4 true inputs (two
  copies of $x_{1,0}^1$ and two copies of $x_{2,0}^1$): therefore,
  $y_0^1, y_0^2, y_0^3$ and $y_0^4$ are propagated to true (indeed,
  $y_0 \geqslant 4$).

  Now, the second sorting network has 7 true inputs: $y_0^3,
  x_{1,1}^1, x_{1,1}^2, x_{1,1}^3, x_{2,1}^1, x_{2,1}^2$ and
  $x_{2,1}^3$. Therefore, $y_1^j$ is propagated to true for $1
  \leqslant j \leqslant 7$.

  Finally, the last sorting network has 3 true inputs: $y_1^3, y_1^6$
  and $x_{1,2}^1$. Therefore, $y_2^1$, $y_2^2$ and $y_2^3$ are
  propagated. That conflicts with clause $\neg y_2^3$.
\end{exmp}

\subsubsection{The \SNTARE{} Encoding for Linear Integer Constraints}\label{subsection-general}

In this section we transform a general LI constraint into a constraint
where all the coefficients have a single digit in base $b$ and the
bound is $b^{m+1}-1$. Then, the new constraint is encoded as in the
previous section. We finally show that consistency is not lost.

Given a constraint $C: a_1 x_1 + \cdots + a_{n-1} x_{n-1} \leqslant a_0$,
let $b$ be a fixed integer larger than 1.\footnote{All the results of
  this section are done with a fixed base $b>1$, where the digits
  $(d_0, d_1, \ldots, d_m)$ represent the number $d_0 + b (d_1 + b
  (d_2 + \cdots + b (d_{m-1} + b d_m)))$; the results can be trivially
  adapted, however, for mixed radix $(b_0, b_1, b_2, \ldots)$ with
  $b_i >1$, where the digits $(d_0, d_1, \ldots, d_m)$ represents the
  number $d_0 + b_0 (d_1 + b_1 (d_2 + \cdots + b_{m-2} (d_{m-1} +
  b_{m-1} d_m)))$.}  We define $m$ as the integer such that $b^m < a_0
\leqslant b^{m+1}$, and $a_{n} = b^{m+1} - 1 - a_0$. Let $x_{n}$
be a dummy variable which is fixed to 1, i.e. $x_{n} = 1$, called
the \emph{tare}.  Then, the constraint is equivalent to

$$\begin{array}{lccccccccc}
C':       & \alineX{0}                  & +         &   \\
         & \alineX{1}                  & +         &   \\
         & \multicolumn{7}{c}{\ldots} & +         &   \\
         & \alineX{m}                  & \leqslant & b^{m+1}-1,
\end{array}$$

where $(A_{i,0},A_{i,1},\ldots, A_{i,m})$ is the
representation of $a_i$ in base $b$; this is,
$$a_i = \sum_{j=0}^m b^j A_{i,j}, \text{ with } 0 \leqslant
A_{i,j} < b.$$

Constraint $C'$ can be encoded as in the previous section.
We refer to this encoding as \SNTARE. 

\begin{exmp}\label{example-withSN}
  Consider the constraint $C: 3x_1 + 2 x_2 + 5x_3 \leqslant 15$, where
  $x_1 \in [0,4]$, $x_2 \in [0,2]$ and $x_3 \in [0,3]$.  For base $b =
  3$, $m = 2$ and the tare is $a_4 = 27 - 1 - 15 = 11$.

  $C$ is therefore rewritten as $C': 2 x_2 + 2 x_3 + 2 x_4 + 3 x_1 +
  3x_3 + 9 x_4 \leqslant 26$.

  \SNTARE{} introduces Boolean variables $y_i^j$ as shown in Figure
  \ref{figure-example-sn-tare}. Then, it adds the clauses $\neg y_2^3
  \wedge x_4^1$.

  \begin{figure}[t]
    \begin{center}
      \includegraphics[scale=0.42]{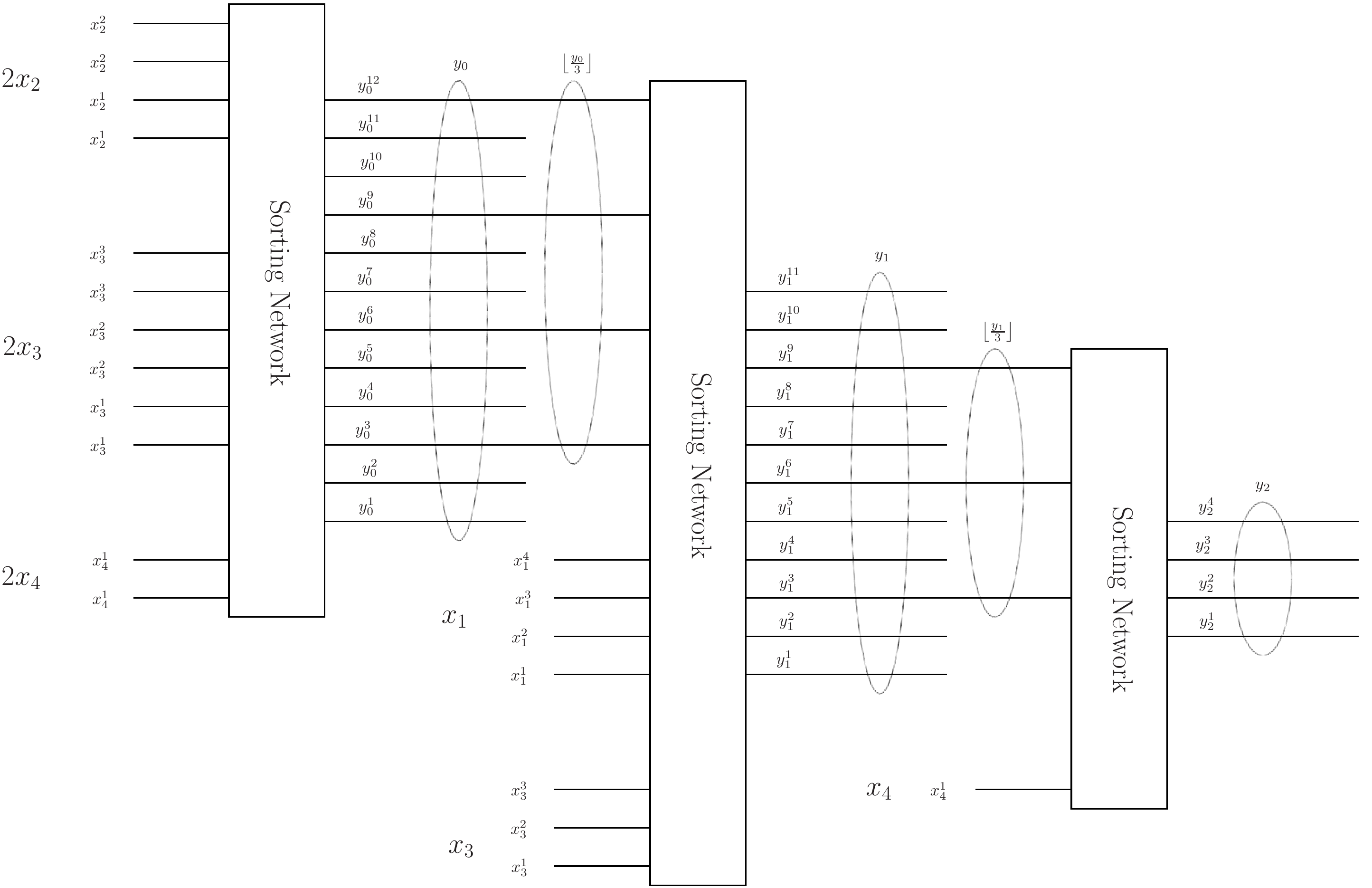}
      \caption{\label{figure-example-sn-tare}
        Network of the constraint
        $C: 3x_1 + 2 x_2 + 5x_3 \leqslant 15$ with base $b=3$ using a tare as $x_4$.}
    \end{center}
  \end{figure}

  The encoding is consistent. For example, if we take the assignment
  $\{x_1 \geqslant 3, x_3 \geqslant 2\}$, the first sorting network
  has 6 true input variables (two copies of $x_3^1$, two copies of
  $x_3^2$ and two copies of $x_4^1$). Therefore, $y_0^j$ will be
  propagated for $j \leqslant 6$.

  Now, the second sorting network has 7 true input variables: $y_0^3,
  y_0^6, x_1^1, x_1^2, x_1^3, x_3^1$ and $x_3^2$. Therefore, $y_1^j$
  is propagated for $j \leqslant 7$.

  Finally, the third network has 3 true inputs: $y_1^3, y_1^6$ and
  $x_4^1$. This causes a conflict with clause $\neg y_2^3$.

  However, the encoding is not domain consistent. If we take the
  assignment $\{x_2 \geqslant 1\}$, the encoding propagates $x_4^1,
  y_0^1, y_0^2, y_0^3, y_0^4, y_1^1, y_2^1, \neg y_2^3$. However,
  $\neg x_3^3$ is not propagated.
\end{exmp}

As shown in the previous example, domain consistency is lost due to
the duplication of variables. The encoding, however, is consistent:

\begin{theorem}\label{theorem-cons-SN}
Let $$C = C(x_{1,1}, x_{1,2}, \ldots, x_{1,m_1}, x_{2,1}, \ldots,
x_{2,m_2}, \ldots, x_{n,m_n})$$ be a monotonic constraint. Let
$$P = P(x_{1,1}, x_{1,2}, \ldots, x_{1,m_1}, x_{2,1}, \ldots, x_{2,m_2},
\ldots, x_{n,m_n})$$ be a consistent propagator of $C$; i.e., given a
partial assignment $A$ on the variables $\{x_{i,j}\}$, the propagator
$P$ finds an inconsistency iff $A$ is inconsistent with
$C$. Then, $$P' = P(x_{1}, x_{1}, \ldots, x_{1},
x_{2}, \ldots, x_{2}, \ldots, x_{n})$$ is a consistent propagator of
$$C' = C'(x_1, x_2, \ldots, x_n) = C(x_1, x_1, \ldots, x_1, x_2,
\ldots, x_2, \ldots, x_n).$$
\end{theorem}
\begin{proof}
The key point in the proof is that, given the monotonicity of $C$, a
partial assignment $\{ x_{i,j} \geqslant v_{i,j} \}$ is inconsistent
with $C$ if and only if $C(v_{1,1}, \ldots, v_{n,m_n})$ is false.

Let $A = \{ x_i \geqslant v_i \}_{i=1}^n$ be a partial assignment on
the variables $x_1, x_2, \ldots, x_n$ inconsistent with $C'$. That means that
$$C'(v_1, \ldots, v_n) = C(v_1, v_1, \ldots, v_1, v_2,\ldots, v_2,
\ldots, v_n)$$ is false. Therefore, $P(v_1, v_1, \ldots, v_n) = P'$ will
find a conflict.
\end{proof}

Notice that the result does not extend to non-monotonic constraints:
\begin{exmp}
Let us consider a constraint $C: x_1 + x_2 = 1$, where $x_i$ are
Boolean variables. Let $P$ be a propagator that, given a complete
assignment of the variables, return a conflict if $C$ does not
hold. Notice that $P$ is consistent. However, constraint $C': x + x =
1$ is unsatisfiable. $P'$ cannot find a conflict until $x$ is given a
value and, therefore, is not a consistent propagator for $C'$.  
\end{exmp}

Also, notice that the result cannot be extended to domain consistency:
\begin{exmp}
Consider the constraint $C: x_1 + x_2 \leqslant 1$, where $x_i$ are
Boolean variables. Let $P$ be domain consistent encoding
of $C$ created by the method of \shortciteA{Abio12}: it includes the 
auxiliary variable $y$ and the clauses $$ \neg x_1 \vee y, \qquad \neg y
\vee \neg x_2$$

Notice that $P$ is domain-consistent: if $x_1$ is assigned to true,
the first clause propagates $y$ and the second one propagates $\neg
x_2$. If $x_2$ is propagated to true, the second clause propagates
$\neg y$ and the first one propagates $\neg x_1$.

If the constraint is replaced by $C': x + x \leqslant 1$, a domain
consistent propagator would propagate $\neg x$. However, $P'$ does not:
clauses are $\{ \neg x \vee y, \neg y \vee \neg x \}$, so unit
propagation cannot propagate $\neg x$.
\end{exmp}

\begin{theorem}\label{theo-sn-tare-consistency}
    The encoding \SNTARE{} is consistent. 
\end{theorem} 
In \SNTARE{} variables are duplicated in the construction that consistency is
shown for by Lemma \ref{lemma-yiff}. Thus by Theorem \ref{theorem-cons-SN}
consistency is maintained. 

In the following Section we show an alternative
proof that uses the fact that the underlying circuit only consist of AND and OR
gates. 

\subsubsection{Monotone Circuits and Sorting Networks}

The encoding \SNTARE{} is the CNF translation of a network of sorting networks.
A sorting network is a network of comparators and a comparator computes the
AND and OR of its inputs. 

A circuit of AND and OR gates is called a \emph{monotone circuit}. By
introducing the tare in \SNTARE{} the underlying structure becomes a monotone
circuit. The output variable of this circuit $y^b_m$ ( meaning $y_m \geq b$) is
true if the partial assignment to the linear is greater than $b^{m+1}-1$, i.e.
$\sum a_i x_i > b^{m+1}-1 $, and otherwise undefined (see previous section). 

We can take advantage of the fact that the circuit is monotone to show
consistency of the translation to CNF. The key insight comes from the
connection between CNF encodings of constraint propagators and monotone
circuits as established by Bessiere \emph{et al} in ~\shortcite{Bessiere09}. A
partial assignment to the encoding of a monotone circuit can be interpreted as
an assignment to the input of the circuit. Input variables to the circuit are
set to true if they are true in the partial assignment, and false otherwise. 

Using this connection, we show an alternative proof to Theorem
\ref{theo-sn-tare-consistency}, that is more compact than the proof in the
previous Section or the similar result in context of PBs \shortcite{Bailleux09,Manthey14} :

\addtocounter{theorem}{-1}
\begin{theorem}
    The encoding \SNTARE{} is consistent. 
\end{theorem} 
\begin{proof} 
By contradiction: Assume UP does not detect an inconsistency, i.e. assume
a partial assignment $A$ such that the sum of the linear constraint exceeds
$b^{m+1}-1$ and there is no conflict. The conflict can only occur between the
unit clause $\neg y^b_m$ introduced by the encoding and a clause containing the
literal $y^b_m$ that is propagated to true by UP under $A$ (see Lemma
\ref{lemma-yiff}). Since there is no conflict, $y^b_m$ is unassigned and not
forced to true by unit propagation . Consider now the total assignment $A'$
which extends $A$ in a way that all unassigned variables  are set to false,
i.e. the remaining inputs of the circuit are set to false. The sum of the
linear expression under $A'$ and $A$ is the same. 

It follows that all auxiliary variables introduced by the Tseytin encoding
corresponding to output of gates that were unassigned, will also be forced
to false by UP. There can be no inversion from false to true since the
circuit does not contain negation.  It follows that also the output gate of
the circuit will be false. However, since all complete extensions of the
partial assignment of $A$ must set the output gate of the circuit to true,
there is a contradiction to the assumption. 
\end{proof}





\subsubsection{Encoding Objective Functions with Sorting Networks}
\label{section-optimization-networks}
The encoding of the previous section works for any LI constraint, but
it is not incremental: this is, we cannot use the encoding of an LI
constraint $C: \sum a_i x_i \leqslant a_0$ to construct the encoding
of $C': \sum a_i x_i \leqslant a_0'$. This is an issue in optimization
problems, where a single constraint with different bounds is encoded.

In this section, we adapt our method to deal with optimization
problems. As explained in Section \ref{section-optimization-mdds},
once we find a solution we do not want to encode the new constraint
from scratch: we want to reuse the encoding of the previous
constraint. As far as we know, this result is novel even for PB: there
is no incremental encoding for pseudo-Booleans (or LI) through sorting
networks.

The main difference between the encoding proposed here and the one for
LI constraints described in the previous section is that here the tare
cannot be used: the right hand side bound on the constraint is not a
fixed value.  Instead, we compute the value of the sum in the left
side and compare it with the right side bound.

As in the previous sections, given a linear integer constraint
$$\sum_{i=0}^n a_i x_i \leqslant a_0,$$ let us rewrite it as
$$\sum_{j=0}^m \sum_{i=1}^n b^j A_{ij} x_i \leqslant \sum_{j=0}^m
b^j \varepsilon_j,$$ where $b>1$ is the chosen base, $0\leqslant
A_{ij} < b$, $0 \leqslant \varepsilon_j < b$ and $m$ is large enough
such that $\sum a_i d_i < b^{m+1}$ (i.e., the computed value $y_m < b$
for any input value of the variables $x_i$).

As in the previous sections, given $0\leqslant j \leqslant m$, we define
$$y_j = \begin{cases}

\sum\limits_{i=1}^n A_{i,0} x_i &\text{if } j = 0 \\ 
\left \lfloor{\frac{y_{j-1}}{b}}\right \rfloor + \sum\limits_{i=1}^n A_{i,j} x_i &\text{if } j > 0 \\ 
\end{cases}$$
Variables $y_j$ are encoded as before with sorting networks; the
input of these networks represents the order encoding of $y_j$. In the
following, we denote
$$\ordEnc(y_j)=[y_j^1, y_j^2, \ldots, y_j^{e_j}]$$ the output
variables of these networks.

To encode the optimization function, besides these variables $y_j$, we
also encode the following variables:
\begin{equation}\label{def-ojk}
  o_j^k := \bigvee_{\substack{1\leqslant l \leqslant e_j\\l \equiv k \ (\text{mod } b)}} \Big(y_j^l \wedge \neg y_j^{l+b-k}\Big)
  \qquad 0 \leqslant j \leqslant m, \ 1 \leqslant k < b,
\end{equation}
where $y_j^{l+b-k}$ is false if the domain of $y_j < l+b-k$. These
variables $o_j^k$ can be easily defined through Tseytin
transformation~\shortcite{Tseytin1968}.

Finally, when we want to encode the constraint with a new bound
$\sum\limits_{j=0}^m b^j \varepsilon_j$, the method just adds the
following clauses:
\begin{equation}\label{clauses-bound}
\bigwedge_{\substack{0 \leqslant j_1 \leqslant m\\ \varepsilon_{j_1} < b-1}} \Big( \neg o_{j_1}^{\varepsilon_{j_1}+1} \vee \bigvee_{\substack{j_1 < j_2 \leqslant m\\ \varepsilon_{j_2} >0}} \neg o_{j_2}^{\varepsilon_{j_2}} \Big)
\end{equation}

\SNOPT{} consists of the clauses encoding the sorting networks computing $y_j$ for $j=1 \ldots m$,
together with clauses (\ref{def-ojk}) and (\ref{clauses-bound}).

Before proving that this encoding is consistent, we need the
following result:

\begin{lemma}\label{lemma-opt-consistency}
  Given a partial assignment $A = \{ x_i \geqslant v_i \}$ such that
  $$\sum_{j=0}^m \sum_{i=1}^n b^j A_{ij} v_i = \sum_{j=0}^m b^j
  \varepsilon_j,$$ the following variables are assigned due to unit
  propagation:
  \begin{enumerate}
    \item $o_j^{\varepsilon_j}$ for all $0 \leqslant j \leqslant m$ with $\varepsilon_j >0$.\label{aaaa}
    \item $\neg o_j^{\varepsilon_j+1}$ for all $0 \leqslant j \leqslant
      m$ with $\varepsilon_j < b-1$.\label{bbbb}
    \item $\neg x_i^{v_i+1}$ for all $1 \leqslant i \leqslant n$ with
      some $A_{ij} \neq 0$ (i.e., $x_i \leqslant v_i$).\label{cccc}
  \end{enumerate}
\end{lemma}
\begin{theorem}\label{theo-network-optimization}
The encoding \SNOPT{} is consistent.
\end{theorem}
This Theorem answers an open question of \shortciteA{Een06} for
the PB case, see appendix for the full proof of the general case for LI. 

\begin{exmp}
  Consider again the constraint $C: 3x_1 + 2 x_2 + 5x_3 \leqslant 15$,
  where $x_1 \in [0,4]$, $x_2 \in [0,2]$ and $x_3 \in [0,3]$. Let us
  fix $b = 3$. Since $\sum a_i d_i = 31$, we can take $m=3$ ($31 <
  b^{m+1} = 81$).
  
  The encoding introduces Boolean variables $y_i^j$ and $o_i^j$ as
  follows (see Figure \ref{figure-example-sn-notare}):

  $$\begin{array}{rcl}
    (y_0^1, y_0^2, \ldots, y_0^{10}) &=& \sn(x_{2}^1, x_{2}^1,
    x_{2}^2, x_{2}^2, x_{3}^1, x_{3}^1, x_{3}^2, x_{3}^2,
    x_{3}^3, x_{3}^3),\\
  (y_1^1, y_1^2, \ldots, y_1^{10}) &=& \sn(y_0^3,y_0^6,y_0^9,
  x_{1}^1, x_{1}^2, x_{1}^3, x_{1}^4, x_{3}^1,
  x_{3}^2,x_{3}^3),\\
  (y_2^1, y_2^2, y_2^3) &=& \sn(y_1^3,y_1^6,y_1^9), \\
  y_3^1 &=& y_2^3.\\
  \end{array}$$

  $$
  \begin{array}{rcl}
    o_0^1 &=& (y_0^1 \wedge \neg y_0^3) \vee (y_0^4 \wedge \neg y_0^6) \vee (y_0^7 \wedge \neg y_0^9) \vee y_0^{10}, \\
    o_0^2 &=& (y_0^2 \wedge \neg y_0^3) \vee (y_0^5 \wedge \neg y_0^6) \vee (y_0^8 \wedge \neg y_0^9), \\
    o_1^1 &=& (y_1^1 \wedge \neg y_1^3) \vee (y_1^4 \wedge \neg y_1^6) \vee (y_1^7 \wedge \neg y_1^9) \vee y_1^{10}, \\
    o_1^2 &=& (y_1^2 \wedge \neg y_1^3) \vee (y_1^5 \wedge \neg y_1^6) \vee (y_1^8 \wedge \neg y_1^9), \\
    o_2^1 &=& y_2^1 \wedge \neg y_2^3, \\
    o_2^2 &=& y_2^2 \wedge \neg y_2^3), \\
    o_3^1 &=& y_3^1, \\
    o_3^2 &=& 0. \\
    
  \end{array}$$

  \begin{figure}[t]
    \begin{center}
      \includegraphics[scale=0.42]{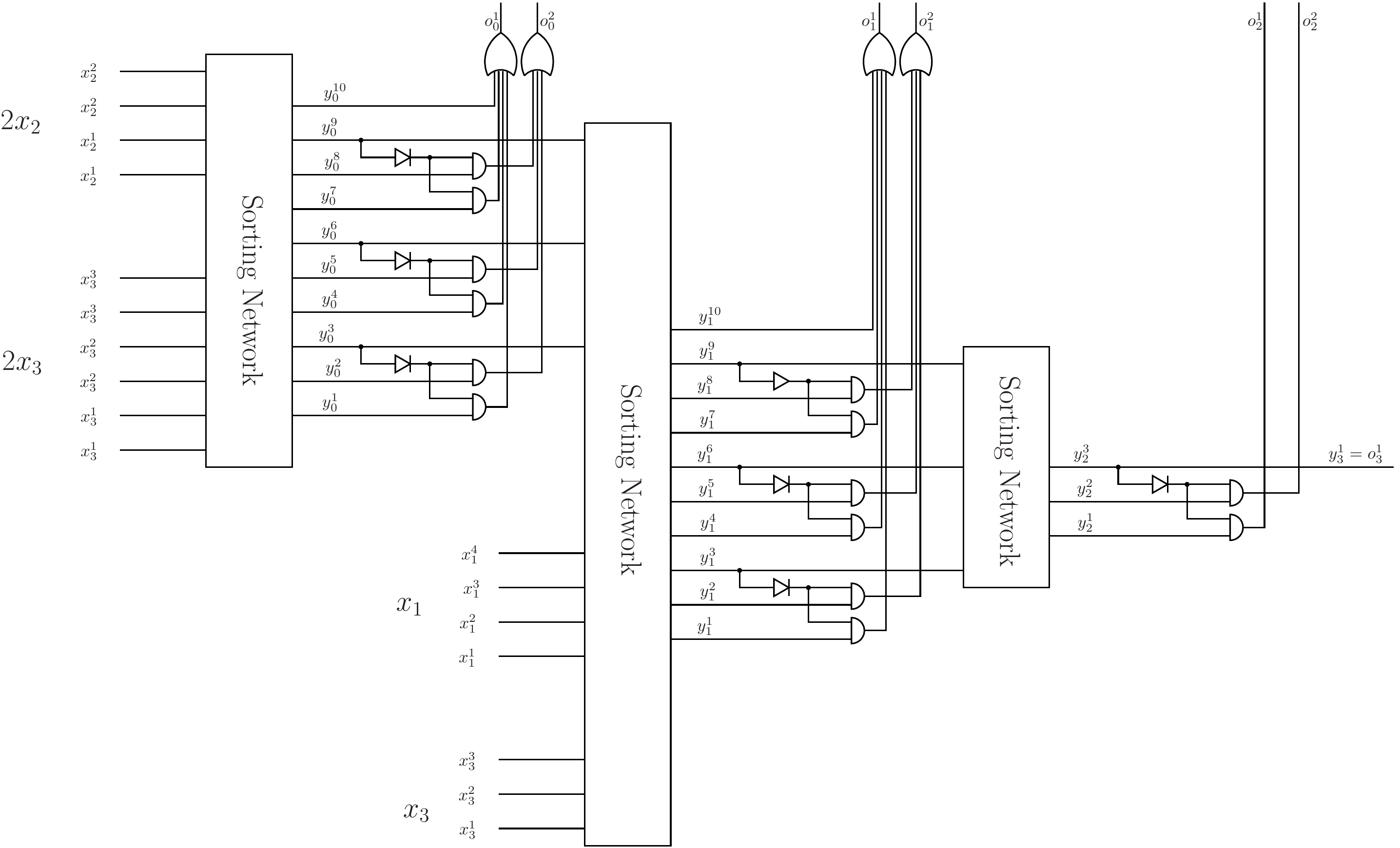}
      \caption{\label{figure-example-sn-notare} Network of the
        constraint $C: 3x_1 + 2 x_2 + 5x_3 \leqslant 15$ with base
        $b=3$ without tare.}
    \end{center}
  \end{figure}

  Since $15 = 0 \cdot 3^0 + 2 \cdot 3^1 + 1 \cdot 3^2 + 0 \cdot 3^3$,
  in this case $\varepsilon_0 = 0, \varepsilon_1 = 2, \varepsilon_2 =
  1, \varepsilon_3 = 0$. Therefore, the method introduces the clauses $$(\neg
  o_0^1 \vee \neg o_1^2 \vee \neg o_2^1) \wedge \neg o_2^2 \wedge \neg
  o_3^1$$ (see Equation (\ref{clauses-bound})).
  
  The encoding is consistent. For example, if we take the assignment
  $\{x_1 \geqslant 3, x_3 \geqslant 2\}$, the first sorting network
  has 4 true input variables (two copies of $x_3^1$ and two copies of
  $x_3^2$). Therefore, $y_0^j$ will be
  propagated for $j \leqslant 4$.

  Now, the second sorting network has 6 true input variables: $y_0^3,
  x_1^1, x_1^2, x_1^3, x_3^1$ and $x_3^2$. Therefore, $y_1^j$
  is propagated for $j \leqslant 6$.

  Finally, the third network has 2 true inputs: $y_1^3$ and $y_1^6$:
  therefore, $y_2^1$ and $y_2^2$ are propagated.

  By Equation (\ref{clauses-bound}), $\neg o_3^1$ is propagated:
  therefore, $y_2^3$ is set to false. Since $y_2^2$ is true, $o_2^2$
  is propagated. That causes a conflict in Equation
  (\ref{clauses-bound}).

  However, the encoding is not domain consistent. If we take the
  assignment $\{x_2 \geqslant 1\}$, the encoding propagates $y_0^1,
  y_0^2, \neg o_3^1, \neg o_3^2, \neg o_2^2, \neg y_2^3$ and $\neg
  y_2^2$. However, the encoding cannot propagate $\neg x_3^3$.

  If now we wish to encode the constraint $C': 3x_1 + 2x_2 + 5x_3 \leqslant 10$ we only
  have to add the clauses $$(\neg o_0^2 \vee \neg o_2^1) \wedge (\neg
  o_1^1 \vee \neg o_2^1) \wedge \neg o_2^2 \wedge \neg o_3^1.$$
\end{exmp}

\subsubsection{Practical Improvements and Size}
\label{section-pract-improvements}

In this section we describe improvements that can be applied to both \SNTARE{}
and \SNOPT. We then prove the asymptotic size for both encodings using these
improvements. 

Notice that these encodings can use any domain consistent implementation of
sorting networks; the concrete implementation or properties have not been used
in any result. Our implementation uses the networks defined by
\shortciteA{parametricCardinalityConstraint}, but this method can be
replicated with any other implementation of sorting networks.

First of all, note that we do not have to encode all the bits of
$y_m$: we only need the $b$ last bits. We can therefore
replace the sorting networks by cardinality networks: when computing
$y_{j}$, we need a $b^{m-j+1}$ cardinality network. For the lowest
values of $j$, this value is larger than the input sizes of the
network: in that case, the cardinality network is a usual sorting
network. However, for the largest values of $j$, cardinality networks
produce a more compact encoding.

Another important improvement is that we do not have to sort all the
variables: some of them are already sorted.  For instance , if we are
computing $z_1+z_2$, then $z_1^1 \leqslant z_1^2 \leqslant \ldots
\leqslant z_1^{d_1}$; and $z_2^1 \leqslant z_2^2 \leqslant \ldots
\leqslant z_2^{d_2}$. Therefore, we can replace $\cn(z_1^1, z_1^2,
\ldots, z_1^{d_1}, z_2^1, z_2^2, \ldots, z_2^{d_2})$ by $\smerge(z_1,
z_2)$.

Also notice that if $x_n$ is the tare variable, $x_n=1$ by
construction. Therefore, $\smerge(X; x_n, \ldots, x_n) = (X,x_n,
\ldots, x_n)$: that is, we can remove the simplified merges involving
the tare variable. 

Furthermore, notice that the sum $y_j = \left
\lfloor{\frac{y_{j-1}}{b}}\right \rfloor + \alineP{j}$ can be computed
in several ways (using the associativity and commutativity properties
of the sum). While the result is the same, the encoding size is not,
since each way uses simplified merge networks of different sizes.
(see Example \ref{example-different-networks}). Finding the optimal
order with respect to the size is hard; however, a greedy algorithm,
where in each step we compute the sum of the two smallest terms, 
in practice gives good results.

\begin{exmp}\label{example-different-networks}
  Consider again constraint $C: 3x_1 + 2 x_2 + 5x_3 \leqslant 15$,
  where $x_1 \in [0,4]$, $x_2 \in [0,2]$ and $x_3 \in [0,3]$. Consider
  the tare case. Figures \ref{figure-example-merge} and
  \ref{figure-example-simpMerge} contain the implementations of the
  method with different term orders in the computation of $y_j$. In
  Figure \ref{figure-example-merge}, we compute the values $y_j$
  without reordering the terms, whereas in Figure
  \ref{figure-example-simpMerge} we reorder them to generate smaller
  networks.

  \begin{figure}[t]
    \begin{center}
      \includegraphics[scale=0.4]{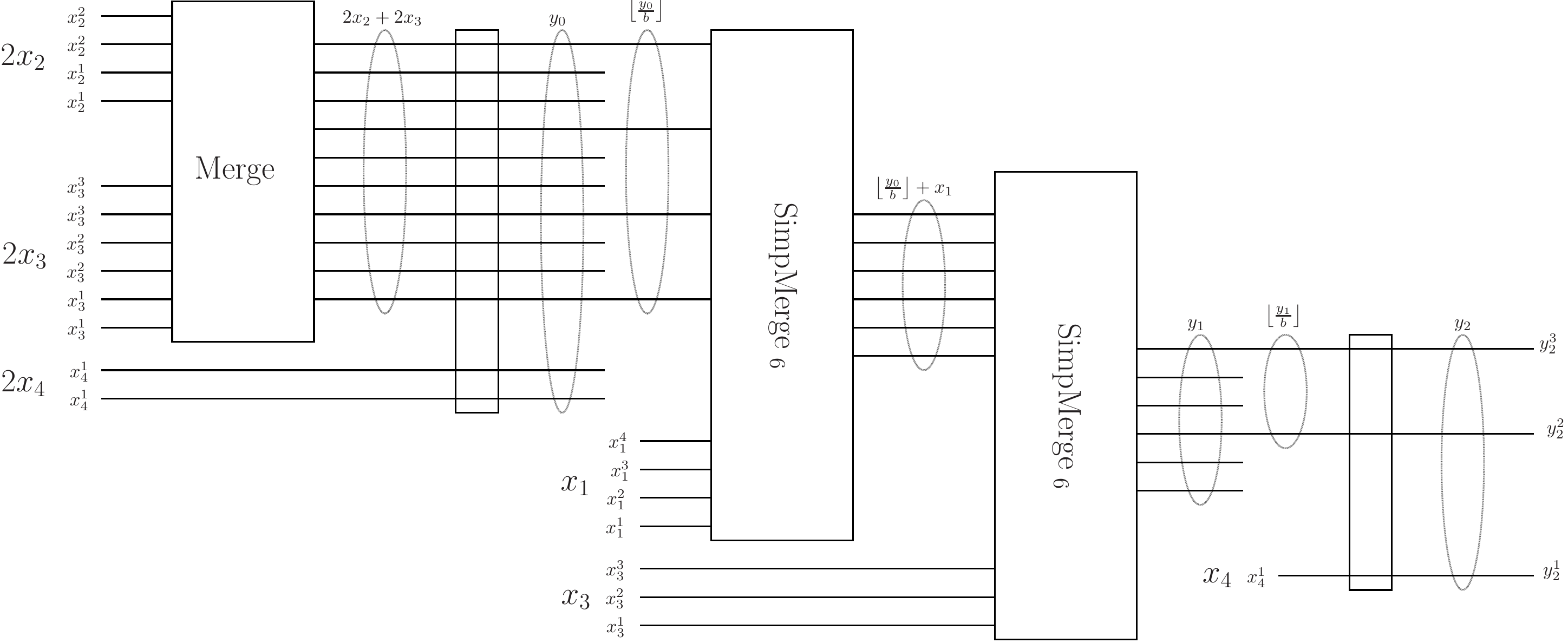}
      \caption{\label{figure-example-merge}Network of the constraint
        $3x_1 + 2 x_2 + 5x_3 + 11 x_4 \leqslant 26$, with $x_4 = 1$
        without reordering the terms.}
    \end{center}
  \end{figure}

  \begin{figure}[t]
    \begin{center}
      \includegraphics[scale=0.4]{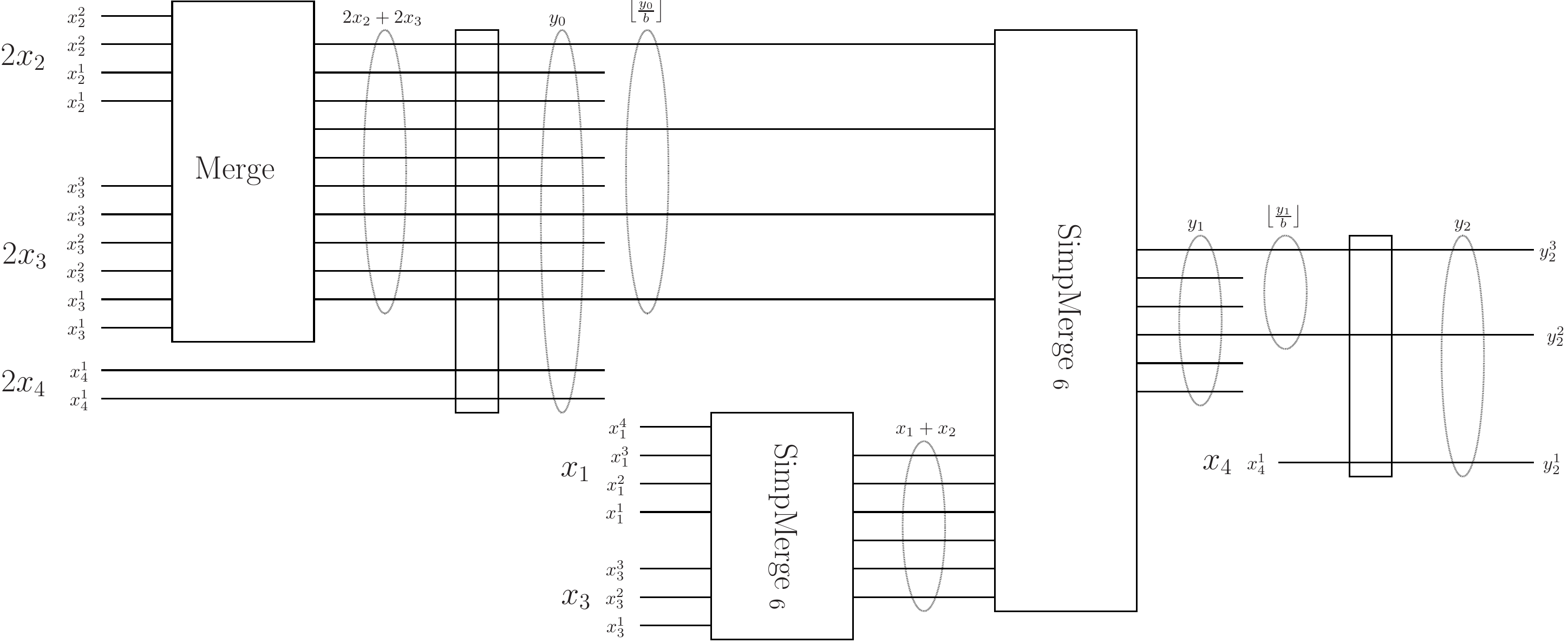}
      \caption{\label{figure-example-simpMerge}Network of the
        constraint $3x_1 + 2 x_2 + 5x_3 + 11 x_4 \leqslant 26$, with
        $x_4 = 1$, after reordering the terms.}
    \end{center}
  \end{figure}

  If we directly apply some method for encoding the sorting networks,
  some of the clauses would be subsumed by other ones. This happens because we
    have duplicated inputs in the network. The formula can be easily simplified
    once obtained. Alternatively, the use of $b=2$ produces simplified formulae
    directly.
\end{exmp}

Finally, we can omit redundant merge networks in the
last layers. Let $\amax=\max \{a_i~|~ i \in 1..n \}$, the largest
coefficient.
We observe that the sorting networks after the $\log(\amax)$th level only merge
already sorted output with the tare. Thus, they can be omitted. This
effectively gives a better size. For instance, in Example
\ref{example-different-networks} in both networks the last sorter is
unnecessary as it merges the tare $x^1_4$ with an already sorted output. In
fact, the output setting $\neg y^3_2$ is equivalent to $\neg y^6_1$. This
discussion reduces the number of layers of sorters to $\log(\amax)$, which
is fewer
than $\log(a_0)$. 

Our implementation encodes each constraint with all the values of $b$
between 2 and 10, and it selects the most compact encoding. Almost
always it is $b=2$, but the cost of trying other bases is negligible.





\begin{theorem} \label{theo-size-SN}
    The encodings \SNTARE{} and \SNOPT{} using the improvements in this Section
    require $O(nd \log n \log d \log \amax)$ variables and clauses, where $d =
    \max \{d_i\}$.
\end{theorem}

\subsection{Encoding Linear Integer Constraints through Partial Sums}
\label{section-sugar}
In this section we explain the encoding of LI constraints 
of \shortciteA{Sugar}. 
Basically, the encoding introduces integer variables
representing the partial sums $s_i = a_i x_i + a_2 x_2 + \cdots + a_i
x_i$, which are encoded with the order encoding, and then simply
encodes $s_n \leqslant a_0$. This encoding is called \emph{Support
  Encoding}, since it encodes the support of the partial sums.

These partial sums can be easily encoded in a recursive way, this is,
using that $s_i = s_{i-1} + a_i x_i$. Since, we are only interested in
the positive polarity of $\ordEnc(s_i)$ (only the lower bounds of the
variables $x_i$ propagate), 
we can simply
encode these equalities with the clauses $$s_{i-1} \geqslant b \wedge
x_i \geqslant c \rightarrow s_i \geqslant b + a_i c.$$ 
The encoding can be simplified as in ~\shortcite{Sugar} since some of these
clauses can be subsumed. All in all, the encoding needs $O(n a_0)$
variables and $O(n d a_0)$ clauses in the worst case, where $d$ is the
size of the largest domain of the variables $x_i$.

The encoding is very similar to the MDD, \MDD, encoding defined in Section
\ref{section-MDD}. In fact, the clauses it introduces are identical to
the MDD encoding from Section
\ref{subsection-encoding-MDDs-into-CNF}. Since this encoding does not
check if two bounds of a partial sum are equivalent, the encoding is
indeed equivalent to the MDD encoding when a non-reduced MDD is used.

In general, however, {\Sup} creates redundant clauses and variables
which are not created using \MDD. For instance, if $a_{i+1}, a_{i+2},
\ldots, a_n$ and $a_0$ are even, $s_i \geqslant 3$ and $s_i \geqslant
4$ are equivalent in the constraint. A reduced MDD will merge these
two variables into a single node, while {\Sup} creates two different
variables. In that aspect, our encoding \MDD{} from Section
\ref{section-MDD} is an improved version of \Sup.

\subsection{Encoding Linear Integer Constraints through the Logarithmic Encoding}
\label{section-logarithmic}

In the previous sections we have seen some methods to handle LI
constraints when integer variables are encoded with the order
encoding. Here, we explain the different encodings when the
logarithmic encoding is used. First, we explain the different
possibilities described in the literature as well as some
generalizations of PB encodings that work as well with LI
constraints. Finally, we introduce a new method, \BDDDec{}, 
more compact than
most of the state-of-the-art encodings, but with a reasonable
propagation strength.

Given a linear constraint
$$C: a_1 x_1 + a_2 x_2 + \cdots + a_n x_n \leqslant a_0,$$ let $[y_{b,i}^0,
  y_{b,i}^1, \ldots, y_{b,i}^{\delta_i}]$ be $\logEnc(x_i)$ for $1
\leqslant i \leqslant n$. In other words,

$$x_i = y_{b,i}^0 + 2 y_{b,i}^1 + \cdots +
2^{\delta_i}y_{b,i}^{\delta_i}, \text{ where } \delta_i = \log d_i
\ \forall i \in \{1, 2, \ldots, n \}.$$

\subsubsection{Linear constraints as multiplication by a constant and summation}

\label{huang}

Perhaps the most obvious way to encode linear constraints is using a binary
encoding of integers and using ripple carry adders to encode both addition
and multiplication by a constant.  This is the method used by
both FznTini~\shortcite{DBLP:conf/cp/Huang08} and Picat SAT by ~\shortciteA{Zhou2017}.
Interestingly more complex adder circuits like carry look ahead adders,
or parallel prefix adders, introduced by circuit designers to make addition
circuits faster, appear to be worse for encoding arithmetic in
SAT

In these methods the linear inequality $C$
is broken into  $z_i = a_i x_i, 1 \leq i \leq n$ and $z_1 + \cdots + z_n \leq a_0$,
and additions and multiplication by a constant are encoded using adder
circuits.

A ripple carry adder encoding the addition of two non-negative $n$-bit logarithmic
integers $u$ and $v$, $w = u + v$, where $\logEnc(u) = [u_b^0, u_b^1,
  \ldots, u_b^{n-1}]$,  $\logEnc(v) = [v_b^0, v_b^1,
  \ldots, v_b^{n-1}]$ and $\logEnc(w) = [w_b^0, w_b^1,
  \ldots, w_b^{n-1}]$  is simply
$$
\begin{array}{rcl}
  (w_b^0, c_1) & = & HA(u_b^0, v_b^0) \\
  (w_b^1, c_2) & = & FA(u_b^1, v_b^1, c_1) \\
  & \vdots  \\
  (w_b^{n-1}, c_n) & = & FA(u_b^{n-1}, v_b^{n-1}, c_{n-1}) \\
\end{array}
$$
where $c_n$ represents the overflow bit. It can be ignored (to implement
fixed width arithmetic), or set to 0 (to force no overflow to occur).
Repeated addition is achieved by recursively breaking the term $z_1 + \cdots
+ z_n$ into two almost equal halves and summing the results of the addition
of the halves.

Multiplication by a constant is implemented by binary addition of shifted
inputs.  Let $w = a u$ where $a$ is a constant. The encoding is
$w = \sum_{i = 0..n-1, bit(i,a) = 1} 2^i u$ where $2^i u$ is calculated by
right shifting the encoding for $u$ $i$ times, and the summation is encoded
as above.

\subsubsection{Linear constraints transformed to psuedo-Boolean constraints}
\label{sectionLIintoPB}

In the remaning methods, the constraint is encoded in two steps: first, it is
transformed into a PB constraint using \emph{logarithmic encodings} of
the integer variables. The PB constraint can subsequently be
translated to CNF yielding a complete method, with one of the methods
explained in Section~\ref{section-PB}.

\ignore{
  Given a linear constraint
$$C: a_1 x_1 + a_2 x_2 + \cdots + a_n x_n \leqslant a_0,$$ let $[y_{b,i}^0,
  y_{b,i}^1, \ldots, y_{b,i}^{\delta_i}]$ be $\logEnc(x_i)$ for $1
\leqslant i \leqslant n$. In other words,

$$x_i = y_{b,i}^0 + 2 y_{b,i}^1 + \cdots +
2^{\delta_i}y_{b,i}^{\delta_i}, \text{ where } \delta_i = \log d_i
\ \forall i \in \{1, 2, \ldots, n \}.$$
}
Give the logarithmic encoding above the linear term is equivalent to
$$a_1 x_1 + a_2 x_2 + \cdots + a_n x_n = a_1 (y_{b,1}^0 + 2y_{b,1}^1 +
\cdots + 2^{\delta_1} y_{b,1}^{\delta_1}) + \cdots + a_n (y_{b,n}^0 +
2y_{b,n}^1 + \cdots + 2^{\delta_n}y_{b,n}^{\delta_n} ),$$ so $C$
is equivalent to
$$C': \sum_{i=0}^n \sum_{j=0}^{\delta_i} 2^j a_i y_{b,i}^j \leqslant a_0.$$

Notice that $C'$ is a pseudo-Boolean constraint, and it can be encoded
with any pseudo-Boolean encoding. The size of $C'$ is $O(n \log d)$, where
$d = \max \{ d_i \}.$ However, the method is not
consistent, as it is shown in the following example:
\begin{exmp}\label{example-logConsistent}
Consider the constraint $$C: 4x_1 + 5 x_2 + 6 x_3\leqslant 14, \text{ where }
x_1, x_2, x_3 \in [0,2].$$ Using the method explained above, the constraint
is transformed into $$C':4 y_{b,1}^0 + 8 y_{b,1}^1 + 5y_{b,2}^0 + 10 y_{b,2}^1 + 6 y_{b,3}^0 + 12 y_{b,3}^1
\leqslant 14.$$

Notice that $C'$ is equivalent to the following set of clauses
$$\begin{array}{lllllll}
  \{ & \neg y_{b,1}^0 \lor \neg y_{b,1}^1 \lor \neg y_{b,2}^0,
  &\quad& \neg y_{b,1}^0 \lor \neg y_{b,1}^1 \lor \neg y_{b,3}^0,
  &\quad&\neg y_{b,1}^0 \lor \neg y_{b,2}^0 \lor \neg y_{b,2}^1,
  \\& \neg y_{b,1}^0 \lor \neg y_{b,2}^0 \lor \neg y_{b,3}^0,
  && \neg y_{b,1}^0 \lor \neg y_{b,2}^1 \lor \neg y_{b,3}^0,
  && \neg y_{b,1}^0 \lor \neg y_{b,3}^1,
  \\& \neg y_{b,1}^1 \lor \neg y_{b,2}^1,
  && \neg y_{b,1}^1 \lor \neg y_{b,2}^0 \lor \neg y_{b,3}^0,
  && \neg y_{b,1}^1 \lor \neg y_{b,3}^1,
  \\& \neg y_{b,2}^0 \lor \neg y_{b,2}^1,
  && \neg y_{b,2}^0 \lor \neg y_{b,3}^1,
  && \neg y_{b,2}^1 \lor \neg y_{b,3}^0,
  \\& \neg y_{b,2}^1 \lor \neg y_{b,3}^1,
  && \neg y_{b,3}^0 \lor \neg y_{b,3}^1 &&&\}
\end{array}$$

The constraint is unsatisfiable if $x_1, x_2, x_3 \geqslant
1$. However, the logarithmic encoding of $x_1, x_2, x_3 \geqslant 1$
is the empty assignment, so unit propagation cannot find any
inconsistency. 
\end{exmp}

The resulting PB constraint can be encoded with any method explained
in Section~\ref{section-PB}. Let us consider the two main approaches
for encoding PB constraints: SNs and BDDs.

Since the resulting method is not consistent, instead of an SN we can use
adders: the propagation strength is similar but the resulting encoding
is much smaller. The resulting encoding, {\Adder}, is the most compact
encoding, since it only needs $O(n \log d \log \amax)$ variables and
clauses where $d = \max \{d_i\}$ and $\amax = \max \{a_i~|~i\geq 1 \}$, but it is
also the worst encoding in terms of propagation strength.

Regarding the BDDs methods, ~\shortciteA{bartzis2006efficient}\label{bartzis}
realized that the BDD size of $C'$ can be reduced by reordering the constraint.
The resulting method, {\BDD}, requires $O(n a_0 \log d)$.

In this paper we improved Bartzis and Bultan's method by also
decomposing the coefficients of $C'$ before reordering.  That generates
a new encoding that we call {\BDDDec}.

\begin{exmp}\label{example-BDDDec}
  Consider the LI constraint $3x_1 + 2x_2 + 5 x_3 \leqslant 15$.
  After encoding the integer variables with the logarithmic encoding,
  the constraint becomes the pseudo-Boolean $3 y_{b,1}^0 + 6 y_{b,1}^1
  + 12 y_{b,1}^2 + 2 y_{b,2}^0 + 4 y_{b,2}^1 + 5 x_{b ,3}^0 + 10
  y_{b,3}^1 \leqslant 15$. 
  \shortciteA{bartzis2006efficient} construct the BDD of the
  pseudo-Boolean $2 y_{b,2}^0 + 3 y_{b,1}^0 + 4 y_{b,2}^1 + 5
  y_{b,3}^0 + 6 y_{b,1}^1 + 10 y_{b,3}^1 + 12 y_{b,1}^2 \leqslant 15.$
  Our method decomposes the coefficients (i.e., considers $y_{b,1}^0 +
  2 y_{b,1}^0$ instead of $3 y_{b,1}^0$) and builds the resulting BDD;
  so we encode the constraint $ y_{b,1}^0 + y_{b,3}^0 + 2 y_{b,2}^0 +
  2 y_{b,1}^0 + 2 y_{b,1}^1 + 2 y_{b,3}^1 + 4 y_{b,2}^1 + 4 y_{b,3}^0
  + 4 y_{b,1}^1 + 4 y_{b,1}^2 + 8 y_{b,3}^1 + 8 y_{b,1}^2 \leqslant
  15.$ 
\end{exmp}

Formally, the {\BDDDec} method encodes LI constraint $a_1 x_1 + \cdots +
a_n x_n \leqslant a_0$ with $x_i \in [0,d_i], 1 \leq i \leq n$ by
first creating the PB constraint
$$\sum_{i = 1}^n \sum_{\substack{0 \leqslant j \leqslant \lfloor
    \log_2 d_i \rfloor\\ (d_i \div 2^j) \!\!\!\!\mod 2 = 1}}
\sum_{\substack{0 \leqslant k \leqslant \lfloor \log_2 a_i
    \rfloor\\ (a_i \div 2^k) \!\!\!\! \mod 2 = 1}} 2^{j+k} \times
y_{b,i}^j \leqslant a_0$$
over the logarithmic encoding variables $y_b$ 
and encoding this using the state-of-the-art encoding for PB 
constraints given by \shortciteA{Abio12}.

\begin{theorem}
  Given a LI constraint $C: a_1 x_1 + a_2 x_2 + \cdots + a_n x_n
  \leqslant a_0,$ {\BDDDec} encodes $C$ with $O(n^2 \log d \log \amax)$,
  where $\amax$ is the largest coefficient and $d$ is the largest domain of
  the integer variables $x_1, x_2, \ldots, x_n$. 
\end{theorem}
Notice that {\BDDDec} encodes a PB with only power-of-two
coefficients. Therefore, the theorem follows immediately from results
of \shortciteA{Abio12}.

Notice that, since $n \log \amax \ll a_0$, {\BDDDec} is more compact
than {\BDD}.

\section{Experimental Results}
\label{section-experiments}

In this section we experimentally compare the main encodings of PB and
LI constraints. We also want to check if our improvements work in
practice: that is, if the preprocessing method explained in Section
\ref{section-groupingCoeffs} improves some of the encodings, if {\MDD}
improves {\Sup} and if {\BDDDec} improves {\BDD}.

All experiments were performed in a 2x2GHz Intel Quad Core Xeon E5405,
with 2x6MB of Cache and 16 GB of RAM. The Barcelogic SAT solver of 
\shortciteA{Bofilletal2008CAV} was used for all the SAT-based methods. We also
compare against lazy clause generation (LCG) ~\shortcite{lazyj} approaches
which directly propagate linear constraints, and explain this propagation
implemented in the Barcelogic SAT solver. 
We also compare against lazy decomposition 
(LD)~\shortcite{DBLP:conf/cp/AbioS12,encodeOrPropagate} methods, which use
LCG propagation by default for all linear constraints, but during runtime
decompose the most important linear constraints using some encoding, also
implemented in the Barcelogic SAT solver. 

Before commenting upon the results, let us explain the different families of
benchmarks used.  We use both PB benchmarks and general LI benchmarks.

\subsection{Benchmarks}
\subsubsection{Pseudo-Boolean Benchmarks}
\paragraph{RCPSP}
Resource-constrained project scheduling problem~\shortcite{Blazewicz198311}
(RCPSP) is one of the most studied scheduling problem. It consists
of tasks consuming one or more resources, precedences between some
tasks, and resources. Here we consider the case of non-preemptive
tasks and renewable resources with a constant resource capacity over
the planning horizon. A solution is a schedule of all tasks so that
all precedences and resource constraints are satisfied.

The objective of RCPSP is to find a solution minimizing the
makespan. The problem is encoded the same way as by  
\shortciteA{SchuttEtal2009CP},
resulting in one pseudo-Boolean constraints per resource and time
slot. These PB constraints are then encoded with the different
methods. Here we have considered the 2040 original RCPSP problems from
PSPlib~\shortcite{PSPlib}.

\paragraph{Pseudo-Boolean Competition 2015}
Another set of problems we have considered is the benchmarks from the
Pseudo-Boolean competition 2015
(\url{http://pbeva.computational-logic.org/}). We have considered the
benchmarks from SMALL-INT optimization which contain Pseudo-Booleans
(this is, we have removed benchmarks with only cardinality constraints
or only clauses).

We have filtered some benchmarks that can be trivially solved by any
method: from the 6266 benchmarks available, we have selected 3993 that
cannot be solved by the PB solver \emph{Clasp} \shortcite{clasp} in 15
seconds. From these benchmarks, we have randomly selected 500.

\paragraph{Sport Leagues Scheduling} 

Another experiment considers scheduling a double round-robin sports
league of $N$ teams. All teams meet each other once in the first $N -
1$ weeks and again in the second $N -1$ weeks, with exactly one match
per team each week. A given pair of teams must play at the home of one
team in one half, and at the home of the other in the other half, and
such matches must be spaced at least a certain minimal number of weeks
apart. Additional constraints include, e.g., that no team ever plays
at home (or away) three times in a row, other (public order, sportive,
TV revenues) constraints, blocking given matches on given days, etc.

Additionally, the different teams can propose a set of constraints
with some importance (low, medium or high). It is desired not only to
maximize the number of these constraints satisfied, but also to assure
that at least some of the constraints of every team are
satisfied. More information can be found in the thesis of \shortciteA{thesisAbio}.

Low-importance constraints are given a weight of 1; medium-importance,
5, and high-importance, 10. For every constraint proposed by a team
$i$, a new Boolean variable $x_{i,j}$ is created. This variable is set
to true if the constraint is violated. For every team, a
pseudo-Boolean constraint $\sum_j w_{i,j} x_{i,j} \leqslant K_i$ is
imposed. The objective function to minimize is $\sum_i \sum_j w_{i,j}
x_{i,j}$. The data is based on real-life instances.

We have considered 10 different problems with 20 random seeds. In all
the problems, the optimal value was found around 30.

\subsubsection{Linear Integer Benchmarks}

\paragraph{Multiple Knapsack}

Here we consider the classic multiple knapsack problem.
$$
\begin{array}{lllllllllll}
\text{Max } & a_1^0 x_1 &+& a_2^0 x_2 &+& \cdots &+& a_n^0 x_n & & & \text{such that} \\
 &a_1^1 x_1 &+& a_2^1 x_2 &+& \cdots &+& a_n^1 x_n &\leqslant &a_0^1 \\
 & \ldots \\
 & a_1^m x_1 &+& a_2^m x_2 &+& \cdots &+& a_n^m x_n &\leqslant& a_0^m, \\
\end{array}$$
where $x_i$ are integer variables with domain $[0, d]$ and the
coefficients belong to $[0, \amax]$.

Since it only consists of linear integer constraints it is ideal for MIP
solvers. We consider this problem since it is easy to modify the
parameters of the constraints, and, therefore, we can easily compare
the encodings in different situations. More precisely, we have
considered different constraint sizes, coefficient sizes and domain
sizes. In these problems, $n$ is the number of variables, $m = 20$ is the
number of LI constraints, $d+1$ is the domain size of the variables;
and $\amax$ is the bound of the coefficients.

For each parameter configuration, 100 benchmarks are considered. The
experiments use $m=20$, $d = 20$, $\amax = 10$ and $n=15$ unless
stated otherwise.

\paragraph{Graph Coloring}

The classical graph coloring problem consists in, given a graph,
assign to each node a color $\{0, 1, \ldots, c-1\}$ such that two
nodes connected by an edge have different colors. Usually, the problem
consists in finding a solution that minimizes the number of colors
(i.e., $c$).  Here we have considered a variant of this problem. Let
us consider a graph that can be colored with $c$ colors: For each node
$\nu$ of the graph, let us define an integer value $a_\nu$. Now, we
want to color the graph with $c$ colors $\{0, 1, \ldots, c-1\}$
minimizing the function $\sum a_\nu x_\nu$, where $x_\nu$ is the color
of the node $\nu$.

We have considered the 80 graph coloring instances from
\url{http://mat.gsia.cmu.edu/COLOR08/} that have less than 500
nodes. For each graph problem, we have considered 4 different
benchmarks: in the $i$-th one, $1 \leqslant a_\nu \leqslant 3i-2$ for
$i = 1,2,3,4$.

\paragraph{MIPLib Benchmarks}
Finally, we have considered the instances from MIPLib
\shortcite{miplib}. These instances come from academia and industry.

Encoding methods perform well in problems with lots of Boolean
variables and clauses and a small group of LI constraints; these
problems mainly contain LI constraints, so we do not expect that our
methods outperform MIP solvers. However, they are the standard tool to
compare MIP solvers, and it is interesting to check how the different
encodings perform in pure LI problems.

We have considered binary only problems and integer only problems 
with only integer coefficients and bounded domains. From
them, we have selected the problems with domains ($d$) bounded by
100000, 
constraint size ($n$) bounded by 20000, coefficients
($\amax$) bounded by 4000 and with less than 100000 constraints. We
obtained 34 benchmarks, that have been run with 20 different random
seeds.

\subsection{Results Presentation}

Multiple knapsack problems are very easy ({\Gurobi} can solve any of
them in less than one second), so results are reported using the
average time for solving the instances. Timeouts are treated as a 300s
response in computing average times.

In all the other problems we report the so-called
\emph{pseudo-harmonic average distance from the solutions found and
  the best solutions known}. Note that we do not use solving times but
distances of solutions, and we do not compute the (arithmetic) average
but the pseudo-harmonic one.

The usage of distances is motivated by the nature of the studied
problems: they are much harder, most of the time they cannot be solved in
3600 seconds. Therefore, the solving time is not suitable: in most
cases it would be 3600 seconds, and would give no information. For a
method and a benchmark, we compute the distance from the best solution
found by that method and the best solution known for the benchmark. If
the method finds no solution, the distance is considered to be
infinite. Notice that in this case we can compare two methods even if
neither of them has found the optimal value.

Regarding the arithmetic mean, we do not use it for two reasons:
first, since some distances are infinite, the mean will be infinite
no matter the other values, making it impossible to compare 
methods. Second, the mean is highly affected for outliers, 
and we don't wish the results to be so sensitive to them.
For example, if we have 10 benchmarks and method 1 finds the optimal
solution nine times and a solution at distance 1000 in the other one,
and method 2 finds a solution at distance 100 in every case, both
methods will have the same mean distance; but we want to conclude 
that method 1 is better than method 2.

The harmonic average solves these problems: it can be computed when
some numbers are infinite, and it is more stable regarding
outliers. However, the harmonic average cannot be computed if some
value is 0. This is a problem in our case, since when a method finds
the optimal solution, the distance of the solution found and the best
solution is 0. For this reason, we have defined the
\emph{pseudo-harmonic average} as follows: given some non-negative
numbers $x_1, x_2, \ldots, x_n$, let $H$ be the harmonic average of
$x_1 + 1, x_2 + 1, \ldots, x_n+1$. Then, the pseudo-harmonic average
of $x_1, x_2, \ldots, x_n$ is $H-1$.

More formally, the pseudo-harmonic average of $x_1, x_2, \ldots, x_n$
is
$$\frac{n}{\frac{1}{x_1+1} + \frac{1}{x_2+1} + \cdots + \frac{1}{x_n+1}} - 1$$

When presenting some results, a value must be more than a $5\%$ better
than another to be considered significantly better. If the difference
of two values is less than $5\%$, we consider that there is no
significant difference between them.

\subsection{Grouping Variables with the Same Coefficient}
In this section we test the impact of grouping variables with the same
coefficient as a preprocessing technique in both Pseudo-Boolean and LI
constraints, as explained in Section
\ref{section-groupingCoeffs}. This technique can be used with any
consistent encoding or even with an LCG solver. Therefore, this
technique is separately evaluated with the MDD encoding ({\MDD})
explained in Section \ref{section-MDD}, the sorting network encoding
({\SN}) explained in Section \ref{section-sortingEncoding}, with the
support encoding ({\Sup}) from \shortcite{Sugar} explained in Section
\ref{section-sugar} and with an LCG approach ({\LCG}).

Tables \ref{table-originalRcpsp-prep}-\ref{table-miplib-prep} contain
the results of the different approaches. For each method, the best
result is shown in bold (if it improves by more than 5\%).

\begin{table}
   \def\tabcolsep{5pt}
   \begin{center}{\fontsize{8pt}{1em}\selectfont
      \begin{tabular}{|l|lllll|}
         \hline
         & 15s & 60s & 300s & 900s & 3600s \\
         \hline
         \MDDEncAwC & 0.616 & \textbf{0.462} & \textbf{0.284} & \textbf{0.201} & \textbf{0.16} \\
         \MDDEncAnC & {0.639} & {0.541} & {0.351} & {0.253} & {0.189} \\
         \hline
         \CNwC & \textbf{0.618} & \textbf{0.38} & \textbf{0.244} & \textbf{0.197} & \textbf{0.152} \\
         \CNnC & {0.692} & {0.439} & {0.263} & {0.228} & {0.18} \\
         \hline
         \SupwC & {0.662} & {0.557} & {0.346} & {0.234} & \textbf{0.128} \\
         \SupnC & {0.674} & {0.544} & \textbf{0.326} & {0.244} & {0.173} \\
         \hline \hline
         \LCGwC & {0.655} & {0.535} & {0.325} & {0.21} & \textbf{0.118} \\
         \LCGnC & \textbf{0.335} & \textbf{0.166} & \textbf{0.117} & \textbf{0.124} & {0.13} \\

         \hline
      \end{tabular}}
      \caption{\label{table-originalRcpsp-prep} Grouping coefficients: pseudo-harmonic average distance from the 2040 original RCPSP benchmarks.}
   \end{center}
\end{table}
\begin{table}
   \def\tabcolsep{5pt}
   \begin{center}{\fontsize{8pt}{1em}\selectfont
      \begin{tabular}{|l|lllll|}
         \hline
         & 15s & 60s & 300s & 900s & 3600s \\
         \hline
         \MDDEncAwC & {27.2} & {15.8} & {9.12} & {6.86} & {4.67} \\
         \MDDEncAnC & \textbf{17.6} & \textbf{10.1} & \textbf{6.13} & \textbf{4.45} & \textbf{3.08} \\
         \hline
         \CNwC & {27.3} & {14.5} & {8.29} & {6.43} & {4.21} \\
         \CNnC & \textbf{23.7} & \textbf{9.75} & \textbf{6.3} & \textbf{4.48} & \textbf{3.36} \\
         \hline
         \SupwC & {38.6} & \textbf{26.4} & {21.7} & {17.8} & {14.9} \\
         \SupnC & {38.2} & {30.4} & {21.9} & \textbf{15.8} & \textbf{13.9} \\
         \hline \hline
         \LCGwC & {19.4} & {8.37} & {4.99} & {3.88} & {2.91} \\
         \LCGnC & \textbf{15.1} & \textbf{7.24} & {4.96} & {3.90} & {3.00} \\
         \hline
      \end{tabular}}
     \caption{\label{table-pbcomp-prep}Grouping coefficients: pseudo-harmonic average distance from the 500 small-int PB optimization instances of the pseudo-Boolean competition 2015.}
   \end{center}
\end{table}
\begin{table}
  \def\tabcolsep{5pt}
  \begin{center}{\fontsize{8pt}{1em}\selectfont
      \begin{tabular}{|l|lllll|}
        \hline
        & 15s & 60s & 300s & 900s & 3600s \\ \hline

        \MDDEncAwC & \textbf{917} & \textbf{331} & \textbf{38} & \textbf{7.84} & \textbf{5.56} \\
        \MDDEncAnC & {1044} & {459} & {88.8} & {31.2} & {6.32} \\ \hline
        \CNwC & \textbf{27.9} & \textbf{12.1} & \textbf{2.5} & \textbf{1.28} & \textbf{0.748} \\
        \CNnC & {791} & {188} & {16.2} & {7.88} & {5.51} \\ \hline
        \SupwC & {951} & {364} & \textbf{33.6} & \textbf{8.21} & {5.88} \\
        \SupnC & \textbf{888} & {347} & {70.8} & {15.9} & {5.75} \\ \hline \hline
        \LCGwC & \textbf{960} & \textbf{281} & \textbf{31.0} & \textbf{8.20} & \textbf{5.86} \\
        \LCGnC & {1517} & {342} & {47.4} & {13.0} & {7.51} \\

        \hline

     \end{tabular}}
    \caption{\label{table-leagues-prep}Grouping coefficients:
      pseudo-harmonic average distance from 200 sport scheduling
      league benchmarks.}
  \end{center}
\end{table}
\begin{table}
  \def\tabcolsep{2pt}
  \begin{center}{\fontsize{8pt}{1em}\selectfont
      \begin{tabular}{|l|llllll||llllll||llllll|}
        \hline
        & \multicolumn{6}{|l||}{Different values of $n$} &
        \multicolumn{6}{|l||}{Different values of $\amax$} &
        \multicolumn{6}{|l|}{Different values of $d$} \\
        \hline
        & $5$ & $ 10$ & $20$ & $ 40$ & $ 80$ & $ 160$ 
        & $ 1$ & $ 2$ & $ 4$ & $ 8$ & $16$ & $32$ 
        & $ 1$ & $ 2$ & $ 4$ & $10$ & $ 25$ & $100$ \\ \hline
        \MDDEncAwC & \textbf{0.05} & \textbf{6.45} & \textbf{175} & {268} & {290} & 300
                   & 51.8 & \textbf{57.2} & \textbf{68.9} & \textbf{103} & \textbf{107} & \textbf{119}
                   & \textbf{0.01} & \textbf{0.03} & \textbf{0.17} & \textbf{18.9} & \textbf{80.6} & {258} \\
        \MDDEncAnC & 0.17 & 9.62 & 187 & 276 & 299 & 300
                   & \textbf{44.8} & 66.0 & 79.6 & 124 & 120 & 127
                   & 0.09 & 0.10 & 0.36 & 26.0 & 93.6 & 265 \\ \hline

        \CNwC & 0.24 & \textbf{10.4} & \textbf{181} & {267} & 289 & 300
              & {56.0} & 65.5 & \textbf{67.6} & \textbf{122} & \textbf{118} & \textbf{122}
              & 0.14 & 0.23 & 0.47 & \textbf{20.3} & \textbf{93.8} & 261 \\
        \CNnC & \textbf{0.13} & 19.0 & 196 & 272 & 289 & 300
              & 58.7 & {64.3} & 87.0 & 143 & 130 & 135 
              & \textbf{0.04} & \textbf{0.08} & \textbf{0.39} & 26.7 & 102 & 262 \\ \hline

        \SupwC & 0.24 & 17.2 & 196 & 277 & 300 & 300
               & 108 & 99.4 & 98.3 & 142 & 134 & 147
               & 0.05 & 0.17 & 0.60 & \textbf{30.7} & 110 & 279 \\
        \SupnC & \textbf{0.12} & \textbf{16.0} & 197 & 278 & 300 & 300
               & \textbf{53.9} & \textbf{78.6} & \textbf{90.3} & 142 & 133 & {145}
               & \textbf{0.02} & \textbf{0.07} & {0.57} & 32.8 & {108} & {272} \\ \hline \hline

        \LCGwC & 0.26 & 15.2 & 191 & 274 & 288 & 300
               & 183 & 144 & 118 & 146 & 107 & 79.7
               & 0.05 & 0.12 & 0.51 & 31.9 & 105 & 270 \\
        \LCGnC & \textbf{0.01} & \textbf{4.87} & \textbf{173} & {265} & 288 & 300
               & \textbf{117} & \textbf{97.2} & \textbf{88.0} & \textbf{118} & \textbf{94.0} & \textbf{70.6}
               & \textbf{0.02} & \textbf{0.01} & \textbf{0.13} & \textbf{22.9} & \textbf{75.3} & \textbf{242} \\ \hline
    \end{tabular}}
    \caption{\label{table-mks-prep}Grouping coefficients: multiple
      knapsack average solving time.}
  \end{center}
\end{table}

\begin{table}
   \def\tabcolsep{5pt}
   \begin{center}{\fontsize{8pt}{1em}\selectfont
      \begin{tabular}{|l|lllll|}
         \hline
         & 15s & 60s & 300s & 900s & 3600s \\
         \hline
         \MDDEncAwC & \textbf{3.70} & \textbf{3.06} & \textbf{2.72} & \textbf{2.42} & \textbf{2.17} \\
         \MDDEncAnC & 4.61 & 3.87 & 3.10 & 2.72 & 2.51 \\ \hline
         \CNwC & {3.50} & 2.99 & 2.64 & 2.45 & 2.45 \\
         \CNnC & 3.57 & 2.99 & 2.67 & {2.41} & \textbf{2.14} \\ \hline
         \SupwC & \textbf{3.24} & \textbf{2.45} & \textbf{2.06} & \textbf{1.82} & \textbf{1.54} \\
         \SupnC & 4.69 & 3.63 & 3.02 & 2.68 & 2.39 \\ \hline \hline
         \LCGwC & \textbf{2.98} & \textbf{2.32} & \textbf{1.99} & \textbf{1.65} & \textbf{1.10} \\
         \LCGnC & {3.18} & {2.94} & {2.53} & {2.24} & {2.02} \\
         \hline
      \end{tabular}}
      \caption{\label{table-graph-prep} Grouping coefficients: pseudo-harmonic average
        distance from 320 graph coloring benchmarks.}
   \end{center}
\end{table}
\begin{table}
   \def\tabcolsep{5pt}
   \begin{center}{\fontsize{8pt}{1em}\selectfont
      \begin{tabular}{|l|lllll|}
         \hline
         & 15s & 60s & 300s & 900s & 3600s \\
         \hline
         \MDDEncAwC & {3.43} & {2.47} & {1.76} & {1.31} & \textbf{0.767} \\
         \MDDEncAnC & \textbf{2.40} & \textbf{2.03} & \textbf{1.42} & \textbf{1.15} & {0.981} \\ \hline
         \CNwC & {6.12} & \textbf{3.92} & {2.75} & \textbf{2.27} & \textbf{1.58} \\
         \CNnC & {6.32} & {4.30} & {2.71} & {2.42} & {1.82} \\ \hline
         \SupwC & {14.1} & {8.48} & {3.91} & {3.3} & {2.63} \\
         \SupnC & \textbf{6.69} & \textbf{3.54} & \textbf{2.42} & \textbf{1.72} & \textbf{1.12} \\ \hline \hline
         \LCGwC & {11.5} & {5.49} & {3.33} & {1.98} & {1.62} \\
         \LCGnC & \textbf{1.85} & \textbf{1.39} & \textbf{1.19} & \textbf{1.01} & \textbf{0.90} \\ \hline
      \end{tabular}}
     \caption{\label{table-miplib-prep} Grouping coefficients:
       pseudo-harmonic average distance from 680 MIPLib instances}
   \end{center}
\end{table}

Preprocessing clearly improves {\MDD} results in Tables
\ref{table-originalRcpsp-prep}, \ref{table-leagues-prep},
\ref{table-mks-prep} and \ref{table-graph-prep}, and the basic method
is better in Tables \ref{table-pbcomp-prep} and
\ref{table-miplib-prep}. It is worth noticing that SAT-based methods
are not competitive method in Tables \ref{table-pbcomp-prep} and
\ref{table-miplib-prep} (see Section \ref{section-exp-other}), so
{\MDD} would not be the right choice anyway. Preprocessing helps
{\MDD} in all the cases where they are competitive, and,
therefore, {\MDD} should always be used with
preprocessing. Accordingly, we will always use preprocessing on {\MDD}
in the following sections.

For {\SN}, the basic method is only better than the method with
preprocessing in Table \ref{table-pbcomp-prep}; in the other cases,
preprocessing is either better or both methods are roughly equal
(Table \ref{table-graph-prep}). Again, it is worth noticing that for
the problems of Table \ref{table-pbcomp-prep} {\SN} is not a
competitive choice. Preprocessing is helping {\SN} in all the cases
where encodings are competitive, and, therefore, {\SN} should always
be used with preprocessing. Accordingly, we will always use
preprocessing on {\SN} in the following sections.

{\SupwC} is clearly better than {\Sup} in Table \ref{table-graph-prep}
and slightly better in Table \ref{table-leagues-prep}, while the basic
method is better in Tables \ref{table-mks-prep} and
\ref{table-miplib-prep}. Again, the basic method is only better at
cases where encodings should not be used: accordingly, we will always
use preprocessing on {\Sup} in the following sections.

Preprocessing clearly improves {\LCG} in Tables
\ref{table-leagues-prep} and \ref{table-graph-prep}, and the basic
method is better in the other cases. In this case, it is not clear if
Preprocessing helps or not: accordingly, we will always use the method
without preprocessing in the following sections.

\subsection{BDDs with Logarithmic Encoding}

As explained in Section \ref{section-logarithmic}, we can use the
logarithmic encoding to transform an LI constraint into a PB. 
\shortciteA{bartzis2006efficient} realized that this PB can be
reordered and then encoded into a BDD. This BDD can finally be encoded
into SAT. We call {\BDD} the resulting encoding. In Section
\ref{section-logarithmic} we propose an improvement of this method,
called {\BDDDec}, where the coefficients of the PB are also
decomposed before producing the BDD: in this way, the BDD is smaller,
so the resulting encoding generates fewer variables and clauses.

In this section we compare these two methods, to check if there is a
significant improvement of our encoding {\BDDDec} over {\BDD} in
practice.

Tables \ref{table-mks-bdds}-\ref{table-miplib-bdds} contains the
results of both methods in the different LI problems. Notice we have
not compared these methods in the PB problems since they use the
logarithmic encoding. In fact, the encoding {\BDD} of a PB constraint
would be equivalent to the encoding {\MDD}.

\begin{table}
  \def\tabcolsep{2pt}
  \begin{center}{\fontsize{8pt}{1em}\selectfont
      \begin{tabular}{|l|llllll||llllll||llllll|}
        \hline
        & \multicolumn{6}{|l||}{Different values of $n$} &
        \multicolumn{6}{|l||}{Different values of $\amax$} &
        \multicolumn{6}{|l|}{Different values of $d$} \\
        \hline
        & $5$ & $ 10$ & $20$ & $ 40$ & $ 80$ & $ 160$ 
        & $ 1$ & $ 2$ & $ 4$ & $ 8$ & $16$ & $32$ 
        & $ 1$ & $ 2$ & $ 4$ & $10$ & $ 25$ & $100$ \\ \hline
        \BDD & \textbf{0.04} & 7.11 & 185 & 272 & 298 & 300
             & 21.1 & 39.2 & 59.1 & 111 & 110 & 129
             & \textbf{0.01} & \textbf{0.06} & 0.58 & 26.8 & 77.6 & 220 \\
        \BDDDec & 0.12 & \textbf{4.73} & \textbf{163} & 267 & 295 & 300
                & \textbf{10.5} & \textbf{25.6} & \textbf{47.7} & \textbf{90.9} & \textbf{84.8} & \textbf{82.5}
                & 0.07 & 0.13 & \textbf{0.46} & \textbf{18.6} & \textbf{56.3} & \textbf{202} \\
        \hline
    \end{tabular}}
    \caption{\label{table-mks-bdds}BDDs methods: multiple
      knapsack average solving time.}
  \end{center}
\end{table}
\begin{table}
   \def\tabcolsep{5pt}
   \begin{center}{\fontsize{8pt}{1em}\selectfont
      \begin{tabular}{|l|lllll|}
         \hline
         & 15s & 60s & 300s & 900s & 3600s \\
         \hline
         \BDD & {11.5} & {10.4} & {9.35} & {7.32} & {5.82} \\
         \BDDDec & \textbf{10.2} & \textbf{7.06} & \textbf{6.19} & \textbf{5.22} & \textbf{4.60} \\
         \hline


      \end{tabular}}
      \caption{\label{table-graph-bdds} BDDs methods: pseudo-harmonic average
        distance from 320 graph coloring benchmarks.}
   \end{center}
\end{table}
\begin{table}
   \def\tabcolsep{5pt}
   \begin{center}{\fontsize{8pt}{1em}\selectfont
      \begin{tabular}{|l|lllll|}
         \hline
         & 15s & 60s & 300s & 900s & 3600s \\
         \hline
         \BDD & \textbf{6.59} & \textbf{4.53} & \textbf{2.89} & \textbf{2.49} & \textbf{2.20} \\
         \BDDDec & {9.88} & {8.86} & {5.77} & {4.78} & {3.56} \\
         \hline
      \end{tabular}}
     \caption{\label{table-miplib-bdds} BDDs methods: pseudo-harmonic
       average distance from 680 MIPLib instances}
   \end{center}
\end{table}

{\BDDDec} is clearly better at multiple knapsack and graph coloring,
while {\BDD} is the best choice for the MIPLib instances.

\subsection{Domain consistent encodings for LI constraints.}
In this section we compare the two domain consistent encodings for LI
(and PB) constraints: {\Sup}, from \shortciteA{Sugar} explained in Section
\ref{section-sugar}, and our method {\MDD}, explained in Section
\ref{section-MDD}. As discussed in Section \ref{section-sugar}, both
methods are basically equivalent except that {\Sup} generates a
non-reduced MDD, and, therefore, produces redundant variables and
clauses. We want to test here whether this redundancy makes {\MDD}
better than {\Sup} in practice.

Tables \ref{table-originalRcpsp-DC}-\ref{table-miplib-DC} contain the
results of both methods in the different benchmarks. According to the
results of the previous section, both methods use the preprocessing method
described in Section \ref{section-groupingCoeffs}.

\begin{table}
   \def\tabcolsep{5pt}
   \begin{center}{\fontsize{8pt}{1em}\selectfont
      \begin{tabular}{|l|lllll|}
         \hline
         & 15s & 60s & 300s & 900s & 3600s \\
         \hline
         \MDD & \textbf{0.616} & \textbf{0.462} & \textbf{0.284} & \textbf{0.201} & {0.160} \\
         \Sup & {0.662} & {0.557} & {0.346} & {0.234} & \textbf{0.128} \\


         \hline
      \end{tabular}}
      \caption{\label{table-originalRcpsp-DC} Domain consistent methods: pseudo-harmonic average distance from the 2040 original RCPSP benchmarks.}
   \end{center}
\end{table}
\begin{table}
   \def\tabcolsep{5pt}
   \begin{center}{\fontsize{8pt}{1em}\selectfont
      \begin{tabular}{|l|lllll|}
         \hline
         & 15s & 60s & 300s & 900s & 3600s \\
         \hline
         \MDD & \textbf{27.2} & \textbf{15.8} & \textbf{9.12} & \textbf{6.86} & \textbf{4.67} \\
         \Sup & {38.6} & {26.4} & {21.7} & {17.8} & {14.9} \\
         \hline
      \end{tabular}}
  \caption{\label{table-pbcomp-DC}Domain consistent methods:
    pseudo-harmonic average distance from the 500 small-int PB
    optimization instances of the pseudo-Boolean competition 2015.}
   \end{center}
\end{table}

\begin{table}
  \def\tabcolsep{5pt}
  \begin{center}{\fontsize{8pt}{1em}\selectfont
      \begin{tabular}{|l|lllll|}
        \hline
        & 15s & 60s & 300s & 900s & 3600s \\ \hline
        \MDD & {917} & \textbf{331} & {38} & {7.84} & \textbf{5.56} \\
        \Sup & {951} & {364} & \textbf{33.6} & {8.21} & {5.88} \\

        \hline

    \end{tabular}}
    \caption{\label{table-leagues-DC}Domain consistent methods:
      pseudo-harmonic average distance from 200 sport scheduling
      league benchmarks.}
  \end{center}
\end{table}
\begin{table}
  \def\tabcolsep{2pt}
  \begin{center}{\fontsize{8pt}{1em}\selectfont
      \begin{tabular}{|l|llllll||llllll||llllll|}
        \hline
        & \multicolumn{6}{|l||}{Different values of $n$} &
        \multicolumn{6}{|l||}{Different values of $\amax$} &
        \multicolumn{6}{|l|}{Different values of $d$} \\
        \hline
        & $5$ & $ 10$ & $20$ & $ 40$ & $ 80$ & $ 160$ 
        & $ 1$ & $ 2$ & $ 4$ & $ 8$ & $16$ & $32$ 
        & $ 1$ & $ 2$ & $ 4$ & $10$ & $ 25$ & $100$ \\ \hline
        \MDD & \textbf{0.05} & \textbf{6.45} & \textbf{175} & {268} & {290} & 300
             & {51.8} & \textbf{57.2} & \textbf{68.9} & \textbf{103} & \textbf{107} & \textbf{119}
             & \textbf{0.01} & \textbf{0.03} & \textbf{0.17} & \textbf{18.9} & \textbf{80.6} & \textbf{258} \\
        \Sup & 0.12 & 16.0 & 197 & 278 & 300 & 300
             & 53.9 & 78.6 & 90.3 & 142 & 133 & 145
             & 0.02 & 0.07 & 0.57 & 32.8 & 108 & 272 \\ \hline
    \end{tabular}}
    \caption{\label{table-mks-DC}Domain consistent methods: multiple
      knapsack average solving time.}
  \end{center}
\end{table}
\begin{table}
   \def\tabcolsep{5pt}
   \begin{center}{\fontsize{8pt}{1em}\selectfont
      \begin{tabular}{|l|lllll|}
         \hline
         & 15s & 60s & 300s & 900s & 3600s \\
         \hline
         \MDD & {3.70} & {3.06} & {2.72} & {2.42} & {2.17} \\
         \Sup & \textbf{3.24} & \textbf{2.45} & \textbf{2.06} & \textbf{1.82} & \textbf{1.54} \\


         \hline
      \end{tabular}}
      \caption{\label{table-graph-DC} Domain consistent methods:
        pseudo-harmonic average distance from 320 graph coloring
        benchmarks.}
   \end{center}
\end{table}
\begin{table}
   \def\tabcolsep{5pt}
   \begin{center}{\fontsize{8pt}{1em}\selectfont
      \begin{tabular}{|l|lllll|}
         \hline
         & 15s & 60s & 300s & 900s & 3600s \\
         \hline
         \MDD & \textbf{3.43} & \textbf{2.47} & \textbf{1.76} & \textbf{1.31} & \textbf{0.767} \\
         \Sup & {14.1} & {8.48} & {3.91} & {3.30} & {2.63} \\
         \hline
      \end{tabular}}
     \caption{\label{table-miplib-DC} Domain consistent methods:
       pseudo-harmonic average distance from 680 MIPLib instances}
   \end{center}
\end{table}

{\Sup} is better at graph coloring, but {\MDD} is better at the other
instances. The experimental results, then, prove that {\MDD} improves
{\Sup}. Accordingly, we will use {\MDD} as the domain consistent
encoding for LI constraints in the following experiments.

\subsection{Comparison against Other Methods}\label{section-exp-other}
In this section we compare the best encodings with other methods for
solving MIP problems. The encodings selected are {\Adder}, explained
in Sections \ref{section-PB-adders} and \ref{section-logarithmic}; {\BDDDec},
explained in Section \ref{section-logarithmic}; {\SN}, explained in
Section \ref{section-sortingEncoding}; and {\MDD}, explained in
Section \ref{section-MDD}.

Other methods considered here are the LCG solver Barcelogic
\shortcite{Bofilletal2008CAV} for LI constraints (\LCG); a lazy
decomposition solver based on the {\SN} encoding, {\LDSN}; a lazy
decomposition solver based on the {\MDD} encoding, {\LDMDD}; and the
MIP solver {\Gurobi} \shortcite{gurobi}.

Methods {\MDD}, {\SN}, {\LDSN} and {\LDMDD} include the
preprocessing technique described in Section
\ref{section-groupingCoeffs}. Results are given in Tables
\ref{table-originalRcpsp-all}-\ref{table-miplib-all}. For each family of
benchmarks, the best encoding is underlined and the best method is
shown in bold.

\begin{table}
   \def\tabcolsep{5pt}
   \begin{center}{\fontsize{8pt}{1em}\selectfont
      \begin{tabular}{|l|lllll|}
         \hline
         & 15s & 60s & 300s & 900s & 3600s \\
         \hline
         \Adder & {16.7} & {11.1} & {2.99} & {2.37} & {1.43} \\
         \SN & \underline{0.618} & \underline{0.380} & \underline{0.244} & \underline{0.197} & \underline{0.152} \\
         \MDD & \underline{0.616} & {0.462} & {0.284} & \underline{0.201} & \underline{0.160} \\
         \hline
         \LDSN & \textbf{0.343} & {0.176} & \textbf{0.0655} & \textbf{0.0406} & \textbf{0.000} \\
         \LDMDD & \textbf{0.337} & \textbf{0.169} & {0.0783} & {0.0626} & {0.0516} \\
         \LCG & \textbf{0.335} & \textbf{0.166} & {0.117} & {0.124} & {0.13} \\
         \Gurobi & {1.90} & {1.10} & {0.537} & {0.404} & {0.341} \\

         \hline
      \end{tabular}}
      \caption{\label{table-originalRcpsp-all}Pseudo-harmonic average distance from the 2040 original RCPSP benchmarks.}
   \end{center}
\end{table}
\begin{table}
  \def\tabcolsep{5pt}
   \begin{center}{\fontsize{8pt}{1em}\selectfont
      \begin{tabular}{|l|lllll|}
         \hline
         & 15s & 60s & 300s & 900s & 3600s \\
         \hline
         \Adder & {54.7} & {43.0} & {19.2} & {11.0} & {7.9} \\
         \SN & \underline{27.3} & \underline{14.5} & \underline{8.29} & \underline{6.43} & \underline{4.21} \\
         \MDD & \underline{27.2} & {15.8} & {9.12} & {6.86} & {4.67} \\
         \hline
         \LDSN & {14.6} & {7.23} & {3.87} & {3.16} & {2.46} \\
         \LDMDD & {12.4} & {5.52} & {2.88} & {2.21} & {1.69} \\
         \LCG & {15.1} & {7.24} & {4.96} & {3.90} & {3.00} \\
         \Gurobi & \textbf{4.82} & \textbf{2.64} & \textbf{1.83} & \textbf{1.27} & \textbf{0.447} \\
         \hline
      \end{tabular}}
     \caption{\label{table-pbcomp-all}Pseudo-harmonic average distance from the 500 small-int PB optimization instances of the pseudo-Boolean competition 2015.}
   \end{center}
\end{table}

\begin{table}
  \def\tabcolsep{5pt}
  \begin{center}{\fontsize{8pt}{1em}\selectfont
      \begin{tabular}{|l|lllll|}
        \hline
        & 15s & 60s & 300s & 900s & 3600s \\ \hline
        \Adder & $\infty$ & $\infty$ & {953} & {39.8} & {9.37} \\
        \SN & \textbf{\underline{27.9}} & \textbf{\underline{12.1}} & \textbf{\underline{2.50}} & \textbf{\underline{1.28}} & \textbf{\underline{0.748}} \\
        \MDD & {917} & {331} & {38.0} & {7.84} & {5.56} \\
        \hline
        \LDSN & {753} & {198} & {14.2} & {7.40} & {4.42} \\
        \LDMDD & {846} & {297} & {35.6} & {7.78} & {5.02} \\
        \LCGnC & {1517} & {342} & {47.4} & {13.0} & {7.51} \\
        \Gurobi & {$\infty$} & {$\infty$} & {$\infty$} & {$\infty$} & {$\infty$} \\

        \hline
    \end{tabular}}
    \caption{\label{table-leagues-all}Pseudo-harmonic average distance from 200 sport
      scheduling league benchmarks.}
  \end{center}
\end{table}
\begin{table}
  \def\tabcolsep{2pt}
  \begin{center}{\fontsize{8pt}{1em}\selectfont
      \begin{tabular}{|l|llllll||llllll||llllll|}
        \hline
        & \multicolumn{6}{|l||}{Different values of $n$} &
        \multicolumn{6}{|l||}{Different values of $\amax$} &
        \multicolumn{6}{|l|}{Different values of $d$} \\
        \hline
        & $5$ & $ 10$ & $20$ & $ 40$ & $ 80$ & $ 160$ 
        & $ 1$ & $ 2$ & $ 4$ & $ 8$ & $16$ & $32$ 
        & $ 1$ & $ 2$ & $ 4$ & $10$ & $ 25$ & $100$ \\ \hline
        \Adder & \underline{0.05} & 9.56 & 186 & 276 & 296 & 300
               & 57.3 & 60.4 & 74.9 & 115 & {105} & {117}
               & 0.02 & 0.15 & 1.84 & 32.7 & {80.3} & {215} \\
        \BDDDec & 0.12 & \underline{4.73} & \underline{163} & {267} & {295} & 300
                & \underline{10.5} & \underline{25.6} & \underline{47.7} & \underline{90.9} & \underline{84.8} & \underline{82.5}
                & 0.07 & 0.13 & {0.46} & \underline{18.6} & \underline{56.3} & \underline{202} \\
        \SN & 0.24 & {10.4} & {181} & {267} & {289} & 300
              & {56.0} & 65.5 & {67.6} & {122} & {118} & {122}
              & 0.14 & 0.23 & 0.47 & {20.3} & {93.8} & 261 \\
        \MDD & \underline{0.05} & {6.45} & {175} & {268} & {290} & 300
                   & {51.8} & {57.2} & {68.9} & {103} & {107} & {119}
                   & \textbf{\underline{0.01}} & \underline{0.03} & \underline{0.17} & \underline{18.9} & {80.6} & {258} \\ \hline
        \LDSN & 0.28 & 9.77 & 182 & 266 & 288 & 300
              & 43.9 & 56.4 & 70.0 & 111 & 108 & 106
              & 69.0 & 0.11 & 0.40 & 21.7 & 82.0 & 244 \\
        \LDMDD & \textbf{0.01} & 3.23 & 165 & 264 & 287 & 300
               & 44.3 & 48.2 & 59.4 & 90.0 & 91.5 & 91.1
               & 42.0 & \textbf{0.01} & 0.09 & 16.4 & 63.1 & 244 \\
        \LCG & \textbf{0.01} & {4.87} & {173} & {265} & 288 & 300
               & {117} & {97.2} & {88.0} & {118} & {94.0} & {70.6}
               & {0.02} & \textbf{0.01} & {0.13} & {22.9} & {75.3} & {242} \\
        \Gurobi & \textbf{0.01} & \textbf{0.09} & \textbf{0.02} & \textbf{0.03} & \textbf{0.02} & \textbf{0.03}
                & \textbf{0.01} & \textbf{0.01} & \textbf{0.01} & \textbf{0.01} & \textbf{0.02} & \textbf{0.02}
                & \textbf{0.01} & {0.02} & \textbf{0.01} & \textbf{0.01} & \textbf{0.02} & \textbf{0.01} \\ \hline
    \end{tabular}}
    \caption{\label{table-mks-all}Multiple knapsack average solving
      time.}
  \end{center}
\end{table}
\begin{table}
   \def\tabcolsep{5pt}
   \begin{center}{\fontsize{8pt}{1em}\selectfont
      \begin{tabular}{|l|lllll|}
         \hline
         & 15s & 60s & 300s & 900s & 3600s \\
         \hline
         \Adder & {7.80} & {5.76} & {4.06} & {3.54} & {3.21} \\
         \BDDDec & {10.2} & {7.06} & {6.19} & {5.22} & {4.60} \\
         \SN & \underline{3.50} & \underline{2.99} & \underline{2.64} & {2.45} & {2.45} \\
         \MDD & {3.7} & {3.06} & {2.72} & \underline{2.42} & \underline{2.17} \\
         \hline
         \LDSN & \textbf{2.98} & {2.88} & \textbf{2.51} & \textbf{2.28} & {2.28} \\
         \LDMDD & \textbf{3.06} & {2.90} & \textbf{2.58} & {2.36} & {2.14} \\
         \LCG & {3.18} & {2.94} & \textbf{2.53} & \textbf{2.24} & \textbf{2.02} \\
         \Gurobi & {3.09} & \textbf{2.67} & \textbf{2.47} & {2.38} & {2.25} \\

         \hline
      \end{tabular}}
      \caption{\label{table-graph-all} Pseudo-harmonic average
        distance from 320 graph coloring benchmarks.}
   \end{center}
\end{table}
\begin{table}
   \def\tabcolsep{5pt}
   \begin{center}{\fontsize{8pt}{1em}\selectfont
      \begin{tabular}{|l|lllll|}
         \hline
         & 15s & 60s & 300s & 900s & 3600s \\
         \hline
         \Adder & {13.28} & {12.76} & {10.3} & {7.86} & {7.64} \\
         \BDDDec & {9.88} & {8.86} & {5.77} & {4.78} & {3.56} \\
         \SN & {6.12} & {3.92} & {2.75} & {2.27} & {1.58} \\
         \MDD & \underline{3.43} & \underline{2.47} & \underline{1.76} & \underline{1.31} & \underline{0.767} \\ \hline
         \LDMDD & {1.73} & {1.09} & {1.01} & {0.773} & {0.614} \\
         \LDSN & {1.59} & {1.11} & {1.02} & {0.703} & {0.577} \\
         \LCG & {11.5} & {5.49} & {3.33} & {1.98} & {1.62} \\
         \Gurobi & \textbf{1.32} & \textbf{0.77} & \textbf{0.568} & \textbf{0.513} & \textbf{0.365} \\ \hline
      \end{tabular}}
     \caption{\label{table-miplib-all} Pseudo-harmonic average distance from 680 MIPLib
       instances}
   \end{center}
\end{table}

{\Gurobi} is clearly the best method in Tables \ref{table-pbcomp-all},
\ref{table-mks-all} and \ref{table-miplib-all}. In all the other
cases, {\LDSN} improves or is not far from {\LCG}. For Table
\ref{table-leagues-all}, {\SN} is clearly the best method.

\section{Conclusion}
\label{section-conclusion}

In this paper we have investigated how to best encode a linear integer (LI)
constraint into SAT. Since encoding methods for LI constraints are based on
those for cardinality constraints and Pseudo-Boolean (PB) constraints, we
also include a detailed survey of those methods.
We introduce three new approaches for encoding LI constraints:
\begin{itemize}
\item \MDD{} based on mapping an order encoding of the integers to an MDD,
\item \SN{} based on using sorting networks applied to a logarithmic
  encoding of the coefficients and order encoding of the integers, and
\item \BDDDec{} based on encoding both integers and coefficients using
  logarithmic encodings.
\end{itemize}
We have compared these approaches with existing methods, and found that
\MDD{} improves the state of the art for domain consistent encodings of LI
constraints, \SN{} provides a robust consistent encoding method for LI
constraints which provided the best solution for challenging sports
scheduling, 
and \BDDDec{} provides a method for robustly encoding LI constraints with
large coefficients and domains. 
The lazy decomposition versions of \MDD{} and \SN{} are also highly
competitive.

\newpage

\section{Bibliography}

\bibliography{bibNew}
\bibliographystyle{theapa}

\newpage

\appendix

\section{Proofs}

\def\thetheorem{\ref{prop-up-propagator}}
\begin{proposition}
  Let $D$ be a domain on the variables $\cal X$, and let $c$ be a
  constraint on $\cal X$. Let $({\cal Y}, F, e)$ be an encoding of
  $\cal X$ and $({\cal Y}_c, F_c)$ an encoding of $c$. Then $D_1
  \mapsto D_1 \sqcap e^{-1} \circ \pi_{\vert {\cal Y}} \circ
  \textrm{up}_{F_c} \circ e(D_1)$ is a correct propagator of $c$,
  where $\pi_{\vert {\cal Y}}$ is the projection from ${\cal Y}_c$ to
  $\cal Y$ and $\textrm{up}_{F_c}$ is the unit propagation on $F_c$.
\end{proposition}
\begin{proof}
  Let $D_1, D_2$ be domains on $\cal X$. We have to see that
  \begin{enumerate}
    \item If $D_1 \sqsubseteq D_2$, then $D_1 \sqcap e^{-1} \circ
      \pi_{\vert {\cal Y}} \circ \textrm{up}_{F_c} \circ e (D_1)
      \sqsubseteq D_2 \sqcap e^{-1} \circ \pi_{\vert {\cal Y}} \circ
      \textrm{up}_{F_c} \circ e (D_2)$.
    \item $D_1 \sqcap e^{-1} \circ \pi_{\vert {\cal Y}} \circ
      \textrm{up}_{F_c} \circ e (D_1) \sqsubseteq D_1$.
    \item $\{ \solns(c) \st \solns(c) \sqsubseteq D_1 \} = \{
      \solns(c) \st \solns(c) \sqsubseteq D_1 \sqcap e^{-1} \circ
      \pi_{\vert {\cal Y}} \circ \textrm{up}_{F_c} \circ e (D_1) \}$.
  \end{enumerate}

  \begin{enumerate}
    \item $e$ and $e^{-1}$ are monotonically decreasing functions by
      definition. Unit propagation is monotonically decreasing since
      it is a propagator, and $\pi_{\vert {\cal Y}}$ is monotonically
      decreasing since it is a projection. Therefore, $e^{-1} \circ
      \pi_{\vert {\cal Y}} \circ \textrm{up}_{F_c} \circ e$ is
      monotonically decreasing.

      Since $D_1 \sqsubseteq D_2$ and $e^{-1} \circ \pi_{\vert {\cal
          Y}} \circ \textrm{up}_{F_c} \circ e (D_1) \sqsubseteq e^{-1}
      \circ \pi_{\vert {\cal Y}} \circ \textrm{up}_{F_c} \circ e
      (D_2)$, their intersection also satisfies the inequality.
    \item The result is obvious.
    \item Obviously, $$\{ \solns(c) \st \solns(c) \sqsubseteq D_1 \}
      \supset \{ \solns(c) \st \solns(c) \sqsubseteq D_1 \sqcap e^{-1}
      \circ \pi_{\vert {\cal Y}} \circ \textrm{up}_{F_c} \circ e (D_1)
      \},$$ so let us prove the other inequality. Let us take $$D' \in
      \{ \solns(c) \st \solns(c) \sqsubseteq D_1 \}.$$ By definition
      of encoding, $e(D')$ is a complete assignment of $\cal Y$ and
      $c$ is satisfiable on $e(D')$: therefore, $\textrm{up}_{F_c}
      \circ e (D') \neq \emptyset$.

      Since $e(D')$ is a complete assignment of $\cal Y$,
      $$\pi_{\vert {\cal Y}} \circ \textrm{up}_{F_c} \circ e (D') =
      e(D').$$ So $$e^{-1} \circ \pi_{\vert {\cal Y}} \circ
      \textrm{up}_{F_c} \circ e (D') = e^{-1} (e(D')) = D'.$$
      Therefore $$D' = e^{-1} \circ \pi_{\vert {\cal Y}} \circ
      \textrm{up}_{F_c} \circ e (D') \sqsubseteq e^{-1} \circ
      \pi_{\vert {\cal Y}} \circ \textrm{up}_{F_c} \circ e (D_1).$$
  \end{enumerate}
\end{proof}

\def\thetheorem{\ref{prop-intervals}}
\begin{proposition}
  Let $\mathcal M$ be the MDD of a LI constraint $a_1 x_1 + \cdots + a_n x_n
  \leqslant a_0$. Then, the following holds:
  \begin{enumerate}
  \item The interval of the true node $\tnode$ is $[0, \infty)$.
  \item The interval of the false node $\fnode$ is $(-\infty, -1]$.
  \item Let $\nu$ be a node with selector variable $x_i$ and
    children $\{\nu_0, \nu_1, \ldots, \nu_{d_i} \}$. Let $[\beta_j,
      \gamma_j]$ be the interval of $\nu_j$. Then, the interval of
    $\nu$ is $[\beta, \gamma]$, with
    $$\beta = \max \{\beta_r + r a_i \st 0 \leqslant r \leqslant d_i \}, \qquad
    \gamma = \min \{\gamma_r + r a_i \st 0 \leqslant r \leqslant d_i \}.$$
  \end{enumerate}
\end{proposition}

\begin{proof}
  \begin{enumerate}
    \item $0 \leqslant \alpha$ is true (i.e., represented by the MDD
      $\tnode$) if and only if $\alpha \in [0, \infty)$. Therefore,
      the interval of $\tnode$ is $[0, \infty)$.
    \item Analogously, $0 \leqslant \alpha$ is false if and only if
      $\alpha \in (-\infty, -1]$, so the interval of $\fnode$ is
      $(-\infty, -1]$.
    \item
      \begin{description}
        \item[$\mathbf{\subseteq:}$] Given $h \in [\beta, \gamma]$ we
          have to show that $h$ belongs to the interval of $\nu$. Let
          $\{x_j = v_j\}_{j=i}^n$ be an assignment of the variables
          $x_i, x_{i+1}, \ldots, x_n$. We have to show that the
          assignment satisfies the constraint $$\sum_{j=i}^n a_j x_j
          \leqslant h$$ if and only if the path defined by the
          assignment goes from the node $\nu$ to $\tnode$.

      Since $\beta \leqslant h \leqslant \gamma$, by definition of
      $\beta$ and $\gamma$,
      $$\beta_{v_i} + v_i a_i \leqslant \beta \leqslant h \leqslant
      \gamma \leqslant \gamma_{v_i} + v_i a_i$$ So $h - v_i a_i \in
             [\beta_{v_i}, \gamma_{v_i}]$.

             Since $[\beta_{v_i}, \gamma_{v_i}]$ is the interval of
             $\nu_{v_i}$, the assignment $\{x_j = v_j\}_{j=i+1}^n$
             goes from $\nu_{v_i}$ to $\tnode$ if and only if
             $\sum_{j=i+1}^n a_j v_j \leqslant h - v_i
             a_i$. Therefore, the assignment $\{x_j = v_j\}_{j=i}^n$
             goes from $\nu$ to $\tnode$ if and only if $\sum_{j=i}^n
             a_j v_j \leqslant h$ as we wanted to prove.
           \item[$\mathbf{\supseteq:}$] Let $h$ be in the interval of
             $\nu$. We have to show that $h \in [\beta, \gamma]$, this
             is,
             $$\max \{\beta_r + r a_i \st 0 \leqslant r \leqslant d_i
             \} \leqslant h \leqslant \min \{\gamma_r + r a_i \st 0
             \leqslant r \leqslant d_i \}.$$ Take $r$ in $0 \leqslant
             r \leqslant d_i$, we have to show that $\beta_r + r a_i
             \leqslant h \leqslant \gamma_r + r a_i$.

             Let $\{x_j=v_j\}_{j=i+1}^n$ be an assignment going from
             $\nu_r$ to $\tnode$. Then, $\{x_i = r \} \cup
             \{x_j=v_j\}_{j=i+1}^n$ goes from $\nu$ to $\tnode$. Since
             $h$ belongs to the interval of $\nu$,
             $$a_i r + \sum_{j=i+1}^n a_j v_j \leqslant h.$$

             Let $\{x_j=v_j\}_{j=i+1}^n$ be an assignment going from
             $\nu_r$ to $\fnode$. Then, $\{x_i = r \} \cup
             \{x_j=v_j\}_{j=i+1}^n$ goes from $\nu$ to $\fnode$. Since
             $h$ belongs to the interval of $\nu$,
             $$a_i r + \sum_{j=i+1}^n a_j v_j > h.$$

             Therefore, any assignment goes from $\nu_r$ to $\tnode$
             if and only if $\sum_{j=i+1}^n a_j v_j > h - a_i r$. By
             definition, $h - a_i r$ belongs to the interval of
             $\nu_r$, so $\beta_r \leqslant h - a_i r \leqslant
             \gamma_r$ as we wanted to prove.
    \end{description}
  \end{enumerate}
\end{proof}

\def\thetheorem{\ref{lemma-mdd-consistency}}
\begin{lemma}
  Let $A = \{ x_j \geqslant v_j \}_{j=i}^n$ be a partial assignment on
  the last variables. Let $\nu$ be a node of $\mathcal M$ with
  selector variable $x_i$.

  Then, $\mddenc(\mu)$ and $A$ propagates (by unit propagation) $\neg
  z_\nu$ if and only if $A$ is incompatible with $\nu$ (this is, the
  constraint defined by an MDD rooted at $\nu$ does not have any
  solution satisfying $A$).
\end{lemma}
\begin{proof}
  Let us prove the result by induction on $n+1-i$. If $i = n+1$, $\nu$
  can only be $\tnode$ and $\fnode$, and the result is trivial.

  Let us prove the inductive case. Let us denote $\nu_k = \child(\nu,
  k)$, and let $r$ be $v_i$.

  \begin{description}
    \item[$\mathbf{\Rightarrow:}$] Let us assume that $A$ and $\nu$
      are compatible, and let us prove that $\neg z_\nu$ is not
      propagated.

      $z_\nu$ only appears with negative polarity in the clauses
      $$\neg z_\nu \lor \neg y_i^k \lor z_{\nu_k}, \ 0 \leqslant k
      \leqslant d_i.$$ For $k > r$, $y_i^k$ is undefined so these
      clauses cannot propagate $\neg z_\nu$. For $j \leqslant r$, for
      monotonicity, $A \setminus \{x_i \geqslant r\}$ is compatible
      with $\nu_j$, so, by induction hypothesis, $z_{\nu_j}$ is not
      propagated to false. Therefore, these clauses cannot propagate
      $\neg z_\nu$.

    \item[$\mathbf{\Leftarrow:}$] $A \setminus \{x_i \geqslant r \}$
      is incompatible with $\nu_r$ so, by induction hypothesis,
      $z_{\nu_r}$ is propagated to false. Then, the clause $$\neg
      z_\nu \lor \neg y_i^r \lor z_{\nu_r}$$ propagates $\neg z_\nu$.
  \end{description}
\end{proof}

\def\thetheorem{\ref{th-consistency}}
\begin{theorem}
  Unit propagation on $\mddenc(\mu)$ is domain consistent.
\end{theorem}
\begin{proof} We prove the result by induction on $n$. The case $n=0$ is
    trivial, so let us prove the inductive case. Let $A = \{x_j \geqslant
    v_j\}_{j=1}^n$ be a partial assignment which is compatible with $\mathcal
    M$ if and only if $x_i < r$. We have to prove that unit propagation on $A$
    and $\mddenc(\mu)$ propagates $\neg y_i^r$.

  Let us denote $\nu_k =\child(\mu,k)$.
  \begin{description}
    \item[$\mathbf{i>1:}$] For monotonicity of the MDD, $A' = A
      \setminus \{ x_1 \geqslant v_1 \}$ is compatible with
      $\nu_{v_1}$ if and only if $x_i < r$.

      Notice that $z_{\nu_{v_1}}$ is propagated by the clauses $z_\mu$
      and $$\neg z_\mu \lor \neg y_1^{v_1} \lor z_{\nu_{v_1}}.$$ By
      induction hypothesis, $\neg y_i^r$ is propagated on
      $\mddenc(\nu_{v_1})$.

    \item[$\mathbf{i=1:}$] The MDD rooted at $\nu_r$ is incompatible
      with $A \setminus \{x_1 \geqslant v_1 \}$. By the previous
      Lemma, $\neg z_{\nu_r}$ is propagated. Therefore, clauses
      $z_\mu$ and $\neg z_\mu \lor \neg y_1^r \lor z_{\nu_r}$
      propagate $\neg y_i^r$.
  \end{description}
\end{proof}

\def\thetheorem{\ref{lemma-yiff}}
\begin{lemma}
  Let $A = \{ x_{i,j} \geqslant v_{i,j} \}_{1 \leqslant i \leqslant n, 0
    \leqslant j \leqslant m}$ be an assignment. Then,
  $$\sum\limits_{j=0}^m b^j \Big( A_{1,j} v_{1,j} + A_{2,j} v_{2,j} +
  \cdots + A_{n,j} v_{n,j} \Big) > b^{m+1}-1$$ if and only if
  $y_m^{b}$ is propagated to true.
\end{lemma}
\begin{proof}
  \begin{description}
    \item[$\Rightarrow$]
      $$\begin{array}{lllll}
        b^{m+1}-1 & < && & \sum\limits_{j=0}^m b^j \Big( A_{1,j} v_{1,j} + A_{2,j} v_{2,j} + \cdots + A_{n,j} v_{n,j} \Big)\\
        & = & & y_0 + & \sum\limits_{j=1}^m b^j \Big( A_{1,j} v_{1,j} + A_{2,j} v_{2,j} + \cdots + A_{n,j} v_{n,j} \Big)\\
        & \leqslant & b^1 - b^0 + & b \left \lfloor{\frac{y_{0}}{b}}\right \rfloor + & \sum\limits_{j=1}^m b^j \Big( A_{1,j} v_{1,j} + A_{2,j} v_{2,j} + \cdots + A_{n,j} v_{n,j} \Big)\\
        & = & b^1 - b^0 + & b y_1 + & \sum\limits_{j=2}^m b^j \Big( A_{1,j} v_{1,j} + A_{2,j} v_{2,j} + \cdots + A_{n,j} v_{n,j} \Big)\\
        & \leqslant & b^2 - b^0 + & b^2 \left \lfloor{\frac{y_{1}}{b}}\right \rfloor + & \sum\limits_{j=2}^m b^j \Big( A_{1,j} v_{1,j} + A_{2,j} v_{2,j} + \cdots + A_{n,j} v_{n,j} \Big)\\
        & = & b^2 - b^0 + & b^2 y_2 + & \sum\limits_{j=3}^m b^j \Big( A_{1,j} v_{1,j} + A_{2,j} v_{2,j} + \cdots + A_{n,j} v_{n,j} \Big)\\
        & \leqslant & & \cdots\\
        & \leqslant & b^m - b^0 + & b^m y_m. 
      \end{array}$$

      Therefore, $b^{m+1} - b^m > b^m y_m$ so $y_m \geqslant b$. By
      the previous proposition, $(y_m^1, \ldots)$ is a domain
      consistent encoding of the order encoding of $y_m$, so $y_m^b$
      will be propagated.

    \item[$\Leftarrow$]

      By the previous proposition, $(y_m^1, \ldots)$ is a domain
      consistent encoding of the order encoding of $y_m$. Therefore,
      if $y_m^b$ is propagated, $y_m \geqslant b$, so $b^{m+1} - 1 <
      b^m y_m$. Therefore:

      $$\begin{array}{llll}
        b^{m+1}-1 & < & b^m y_m\\
        & = & b^m \left \lfloor{\frac{y_{m-1}}{b}}\right \rfloor +& \sum\limits_{j=m}^m b^j \Big( A_{1,j} v_{1,j} + A_{2,j} v_{2,j} + \cdots + A_{n,j} v_{n,j} \Big)\\
        & \leqslant & b^{m-1} y_{m-1} &+ \sum\limits_{j=m}^m b^j \Big( A_{1,j} v_{1,j} + A_{2,j} v_{2,j} + \cdots + A_{n,j} v_{n,j} \Big)\\
        & = & b^{m-1} \left \lfloor{\frac{y_{m-2}}{b}}\right \rfloor +& \sum\limits_{j=m-1}^m b^j \Big( A_{1,j} v_{1,j} + A_{2,j} v_{2,j} + \cdots + A_{n,j} v_{n,j} \Big)\\
        & \leqslant & \cdots\\
        & \leqslant & b \left \lfloor{\frac{y_{0}}{b}}\right \rfloor+& \sum\limits_{j=1}^m b^j \Big( A_{1,j} v_{1,j} + A_{2,j} v_{2,j} + \cdots + A_{n,j} v_{n,j} \Big)\\
        & \leqslant & y_0 +& \sum\limits_{j=1}^m b^j \Big( A_{1,j} v_{1,j} + A_{2,j} v_{2,j} + \cdots + A_{n,j} v_{n,j} \Big)\\
        & = & & \sum\limits_{j=0}^m b^j \Big( A_{1,j} v_{1,j} + A_{2,j} v_{2,j} + \cdots + A_{n,j} v_{n,j} \Big)\\
      \end{array}$$ 
  \end{description}
\end{proof}

\def\thetheorem{\ref{lemma-opt-consistency}}
\begin{lemma}
  Given a partial assignment $A = \{ x_i \geqslant v_i \}$ such that
  $$\sum_{j=0}^m \sum_{i=1}^n b^j A_{ij} v_i = \sum_{j=0}^m b^j
  \varepsilon_j,$$ the following variables are assigned due to unit
  propagation:
  \begin{enumerate}
    \item $o_j^{\varepsilon_j}$ for all $0 \leqslant j \leqslant m$ with $\varepsilon_j >0$.\label{aaaa}
    \item $\neg o_j^{\varepsilon_j+1}$ for all $0 \leqslant j \leqslant
      m$ with $\varepsilon_j < b-1$.\label{bbbb}
    \item $\neg x_i^{v_i+1}$ for all $1 \leqslant i \leqslant n$ with
      some $A_{ij} \neq 0$ (i.e., $x_i \leqslant v_i$).\label{cccc}
  \end{enumerate}
\end{lemma}
\begin{proof}
  Let us prove the result by induction on $m$. If $m=0$ the results
  are obvious, so let us prove the general case.

  Since $$\sum_{j=0}^m \sum_{i=1}^n b^j A_{ij} v_i = \sum_{j=0}^m
  b^j \varepsilon_j,$$ $\sum\limits_{i=1}^n A_{i0} v_i = \varepsilon_0 + b
  \lambda$ for some integer $\lambda \geqslant 0$.  Due to the
  properties of sorting networks, $y_0^{l}$ is propagated to true for
  all $1 \leqslant l \leqslant \varepsilon_0 + b\lambda$. Therefore,
  constraint
  $$\sum_{l=0}^{\left\lfloor \frac{e_0}{b} \right\rfloor} y_0^{bl} +
  \sum_{j=0}^{m-1} \sum_{i=1}^n b^j A_{ij+1} x_i \leqslant
  \sum_{j=0}^{m-1} b^j \varepsilon_{j+1}$$ satisfies the hypothesis of
  the lemma. By induction hypothesis, the following literals are
  propagated:
  \begin{enumerate}
    \item $o_j^{\varepsilon_j}$ for all $0 < j \leqslant m$ with $\varepsilon_j>0$.
    \item $\neg o_j^{\varepsilon_j+1}$ for all $0 < j \leqslant
      m$ with $\varepsilon_j < b-1$.
    \item $\neg x_i^{v_i+1}$ for all $1 \leqslant i \leqslant n$ with
      some $A_{ij} \neq 0$ with $j>0$.

      $\neg y_0^{bl}$ with $l > \lambda$.
  \end{enumerate}
  Therefore, since $y_0^{\varepsilon_0 + b\lambda}$ is true and
  $y_0^{b(\lambda+1)}$ is false, equation (\ref{def-ojk}) propagates
  $o_0^{\varepsilon_0}$, and the point \ref{aaaa} is proved.

  If $\varepsilon_0 =b-1$, point \ref{bbbb} of the lemma is already
  proved. Besides, $y_0^{\varepsilon_0 + b\lambda+ 1}$ is false, since
  $\varepsilon_0 + b\lambda+ 1 = b (\lambda+1)$.

  If $\varepsilon_0 < b-1$, equation (\ref{clauses-bound}) contain the
  clause $$\neg o_0^{\varepsilon_0+1} \vee
  \bigvee_{\substack{j>0\\ \varepsilon_j>0}} \neg
  o_j^{\varepsilon_j}.$$ Since $o_j^{\varepsilon_j}$ has been
  propagated to true for all $j>0$, then $o_0^{\varepsilon_0+1}$ is
  propagated to false, and \ref{bbbb} is proved. Besides, since
  $y_0^{b(\lambda+1)}$ is false, equation (\ref{def-ojk}) propagates
  $y_0^{\varepsilon_0 + b\lambda+ 1}$ to false.

  Finally, since the sorting network with output $y_0$ has $\varepsilon_0
  + b\lambda$ true inputs and its $\varepsilon_0 + b\lambda+ 1$-th output
  is false, all the other inputs are propagated to false. That proves
  the last point of the lemma.
\end{proof}

\def\thetheorem{\ref{theo-network-optimization}}
\begin{theorem}
    The \SNOPT{} encoding presented in Section
    \ref{section-optimization-networks} is consistent.
\end{theorem}
\begin{proof}
  Let $A = \{x_i \geqslant v_i \}$ be a partial assignment which cannot
  be extended to a full assignment satisfying $C$. In that case,
  $$\sum_{j=0}^m \sum_{i=1}^n b^j A_{ij} v_i > \sum_{j=0}^m b^j
  \varepsilon_j.$$ We have to show that unit propagation finds an
  inconsistency. As before, we prove the result by induction on
  $m$. Case $m=0$ is a direct consequence of last Lemma, so let us prove
  the general case
.

  First of all, let us assume that $$\left \lfloor \frac{1}{b}
  \sum_{j=0}^m \sum_{i=1}^n b^j A_{ij} v_i \right \rfloor >
  \left \lfloor \frac{1}{b}\sum_{j=0}^m b^j \varepsilon_j\right \rfloor.$$

  Notice that $$\left \lfloor \frac{1}{b}\sum_{j=0}^m b^j
  \varepsilon_j\right \rfloor = \sum_{j=0}^{m-1} b^j \varepsilon_{j+1},$$ and
  $$\left \lfloor \frac{1}{b} \sum_{j=0}^m \sum_{i=1}^n b^j A_{ij}
  v_i \right \rfloor = \sum_{l=0}^{\left\lfloor \frac{e_0}{b} \right\rfloor} y_0^{bl} +
  \sum_{j=0}^{m-1} \sum_{i=1}^n b^j A_{ij+1} v_i.$$

  Therefore, $$\sum_{l=0}^{\left\lfloor \frac{e_0}{b} \right\rfloor}
  y_0^{bl} + \sum_{j=0}^{m-1} \sum_{i=1}^n b^j A_{ij+1} v_i >
  \sum_{j=0}^{m-1} b^j \varepsilon_{j+1},$$ and, by induction hypothesis,
  unit propagation finds a conflict.

  Assume now that $$\left \lfloor \frac{1}{b}
  \sum_{j=0}^m \sum_{i=1}^n b^j A_{ij} v_i \right \rfloor \leqslant
  \left \lfloor \frac{1}{b}\sum_{j=0}^m b^j \varepsilon_j\right \rfloor.$$

  Let $\mu, \lambda$ be the two integers such that $\sum_{i=1}^n
  A_{i0} v_i = \mu + b \lambda$, with $0 \leqslant \mu < b$.

  We now have that $$\left \lfloor \frac{1}{b} \sum_{j=0}^m \sum_{i=1}^n
  b^j A_{ij} v_i \right \rfloor = \left \lfloor
  \frac{1}{b}\sum_{j=0}^m b^j \varepsilon_j\right \rfloor$$
  and
  $$\sum_{j=0}^m \sum_{i=1}^n b^j A_{ij} v_i - b \left \lfloor
  \frac{1}{b} \sum_{j=0}^m \sum_{i=1}^n b^j A_{ij} v_i \right
  \rfloor > \sum_{j=0}^m b^j - b \left \lfloor \frac{1}{b}\sum_{j=0}^m
  b^j \varepsilon_j\right \rfloor.$$

  As before, the first equality can be transformed into
  $$\sum_{l=0}^{\left\lfloor \frac{e_0}{b} \right\rfloor}
  y_0^{bl} + \sum_{j=0}^{m-1} \sum_{i=1}^n b^j A_{ij+1} v_i =
  \sum_{j=0}^{m-1} b^j \varepsilon_{j+1},$$ so we can apply the previous Lemma:
  \begin{enumerate}
    \item $o_j^{\varepsilon_j}$ is propagated for all $0 < j \leqslant m$
      with $\varepsilon_j>0$.
    \item $\neg o_j^{\varepsilon_j+1}$ is propagated for all $0 < j
      \leqslant m$ with $\varepsilon_j < b-1$.
    \item $\neg x_i^{v_i+1}$ is propagated for all $1 \leqslant i \leqslant n$ with
      some $A_{ij} \neq 0$ with $j>0$.

      $\neg y_0^{bl}$ is propagated for all $l > \lambda$.
  \end{enumerate}

  By hypothesis, $$\mu = \sum_{j=0}^m \sum_{i=1}^n b^j A_{ij} v_i
  - b \left \lfloor \frac{1}{b} \sum_{j=0}^m \sum_{i=1}^n b^j
  A_{ij} v_i \right \rfloor > \sum_{j=0}^m b^j \varepsilon_j - b
  \left \lfloor \frac{1}{b}\sum_{j=0}^m b^j \varepsilon_j\right
  \rfloor = \varepsilon_0.$$ Notice also that $y_0^l$ is true for all
  $1 \leqslant l \leqslant \mu + b \lambda$; therefore,
  $y_0^{\varepsilon_0+1+b\lambda}$ is true.

  Equation (\ref{clauses-bound}) contains the clause
  $$\neg o_0^{\varepsilon_0+1} \vee
  \bigvee_{\substack{j>0\\ \varepsilon_j>0}} \neg o_j^{\varepsilon_j},$$
  that propagates $\neg o_0^{\varepsilon_0+1}$.

  On the other hand, since $y_0^{\varepsilon_0+1+b\lambda}$ is true and
  $y_0^{b(\lambda+1)}$ is false, equation (\ref{def-ojk}) propagates
  $o_0^{\varepsilon_0+1}$, causing a conflict.
\end{proof}

\def\thetheorem{\ref{theo-size-SN}}
\begin{theorem} 
    Encoding \SNTARE{} and \SNOPT{} applying the improvements in Section
    \ref{section-pract-improvements} require $O(nd \log n \log d \log \amax)$
    variables and clauses, where $d = \max \{d_i\}$. 
\end{theorem}
\begin{proof}
  First, define $\sm(a,b)$ as the number of variables
  needed to encode a simplified merge with inputs of sizes $a$ and
  $b$. Notice that, from \shortciteA{AsinNOR11}, $\sm(a,b) \leqslant
  O(\max\{a,b\} \log (\max\{a,b\}))$.

  Let us prove that the encoding requires $O(nd \log n \log d \log
  \amax$ Boolean variables in the non-tare case. The number of clauses
  can be computed in a similar way. The tare case is very similar.

  First, the encoding defines $$y_0 = A_{1,0} x_1 + A_{2,0} x_2 +
  \cdots + A_{n,0} x_n.$$ Given $1 \leqslant a < b \leqslant n$, let
  us define $$z_0^{a,b} = x_{a} + x_{a+1} + \cdots + x_{b}.$$ Notice
  that the number of Boolean variables needed to define $y_0$ are
  bounded by the number of Boolean variables needed to define
  $z_0^{1,n} = x_1 + x_2 + \cdots + x_n$, which can be
  computed as
  $$
  \begin{array}{lll}
    z_0^{1,2} &=& \smerge(\ordEnc(x_1),\ordEnc(x_2))\\
    z_0^{3,4} &=& \smerge(\ordEnc(x_3),\ordEnc(x_4))\\
    & \cdots& \\
    z_0^{1,4} &=& \smerge(z_0^{1,2},z_0^{3,4})\\
    z_0^{5,8} &=& \smerge(z_0^{5,6},z_0^{7,8})\\
    & \cdots& \\
    z_0^{1,8} &=& \smerge(z_0^{1,4},z_0^{5,8})\\
    & \cdots& \\
  \end{array}$$

So the number of variables to define $y_0$ is bounded by
$$\begin{array}{l}
  \sm(d_1,d_2)+\sm(d_3, d_4) + \cdots + \sm(d_1+d_2,d_3+d_4) + \cdots \leqslant\\
  \frac{n}{2} \sm(d,d) + \frac{n}{4}\sm(2d,2d) + \cdots \leqslant \\
  \frac{n}{2} O(d \log d) +  \frac{n}{4} O(2d \log 2d) +  \cdots =  \\
  \frac{n}{2} O(d \log d) +  \frac{n}{2} O(d \log d) +  \cdots = \\
  O(nd\log d \log n)
\end{array}$$

The encoding then defines $$y_1 = \left \lfloor{\frac{y_{0}}{b}}\right
\rfloor + A_{1,1} x_1 + A_{2,1} x_2 + \cdots + A_{n,1}
x_n$$ Again, we can bound the number of Boolean variables needed to
define $y_1$ with the number of Boolean variables needed to define
$$\left \lfloor{\frac{y_{0}}{b}}\right \rfloor + x_1 + x_2 + \cdots +
x_n,$$ which can be computed as $\smerge(\left
\lfloor{\frac{y_{0}}{b}}\right \rfloor,z_0^{1,n})$. We need
$O(nd\log d \log n) + \sm(nd,nd)= O(nd\log d \log n)$ variables.

Analogously, we need $O(nd\log d \log n)$ variables for the definition of the
other $y_j$ variables. Therefore, to introduce variables $y_0, y_1, \ldots,
y_m$ we need $O(nd\log d \log n \log a_0)$ Boolean variables, since $m = \log_b
a_0$. Using the improvement to the construction, we can remove the merge
networks from layer $(\log_b \amax)+1$ onwards, since they are redundant. 

Finally, we need to introduce variables
  $$o_j^k := \bigvee_{\substack{1\leqslant l \leqslant e_j\\l \equiv k
    \ (\text{mod } b)}} \Big(y_j^l \wedge \neg y_j^{l+b-k}\Big) \qquad
0 \leqslant j \leqslant m, \ 1 \leqslant k < b.$$ We can introduce
$o_j^k$ through Tseytin transformation; it needs
$e_j$ extra variables. Since $e_j \leqslant nd$, and we have
$(m+1)(b-1)$ variables $o_j^k$, we need $O(nd\log \amax)$.

All in all, the encoding needs $$O(nd\log d \log n \log \amax +
O(nd\log \amax) = O(nd\log d \log n \log \amax)$$ Boolean variables.
\end{proof}

\end{document}